\documentclass[oneside, 12pt,a4paper, fleqn]{article}
\usepackage{eurosym}
\usepackage{amssymb,amsmath,xcolor,graphicx,xspace,colortbl,rotating}
\usepackage{natbib}
\usepackage{appendix}
\usepackage{caption}
\usepackage{rotating}
\usepackage{amsthm}
\usepackage{enumerate}
\usepackage{tocbibind}
\usepackage{hyperref}
\hypersetup{
bookmarks=true,         
unicode=false,          
pdftoolbar=true,        
pdfmenubar=true,        
pdffitwindow=false,     
pdfstartview={FitH},    
pdfnewwindow=true,      
colorlinks=true,       
linkcolor=cyan,          
citecolor=blue,        
filecolor=magenta,      
urlcolor=cyan           }
\theoremstyle{theorem}

\newtheorem{proposition}{Proposition} [section]
\newtheorem{corollary}{Corollary} [section]

\newtheorem{assumption}{Assumption} [section]

\usepackage{titling}
\begin{document}

\title{{\Large Identification Robust Inference for the Risk Premium in Term
Structure Models }}
\author{Frank Kleibergen\thanks{
Amsterdam School of Economics, University of Amsterdam, Roetersstraat 11,
1018 WB Amsterdam, The Netherlands. Email: f.r.kleibergen@uva.nl.} \and %
Lingwei Kong\thanks{
Faculty of Economics and Business, University of Groningen, Nettelbosje 2,
9747 AE Groningen, The Netherlands. Email: l.kong@rug.nl. } }
\date{October 16, 2022} 
\maketitle

\begin{abstract}
We propose identification robust statistics for testing hypotheses on the
risk premia in dynamic affine term structure models. We do so using the
moment equation specification proposed for these models in \nocite%
{adrian2013pricing}Adrian et al. (2013). We extend the subset (factor)
Anderson-Rubin test from Guggenberger et al. (2012) to models with multiple
dynamic factors and time-varying risk prices. Unlike projection based tests,
it provides a computationally tractable manner to conduct identification
robust tests on a larger number of parameters. We analyze the potential
identification issues arising in empirical studies. Statistical inference
based on the three-stage estimator from Adrian et al. (2013) requires
knowledge of the factors' quality and is misleading without full-rank $\beta 
$'s or with sampling errors of comparable size as the loadings. Empirical
applications show that some factors, though potentially weak, may drive the
time variation of risk prices, and weak identification issues are more
prominent in multi-factor models.
\end{abstract}



\thispagestyle{empty}\bigskip


\bigskip


\section{Introduction}

A variety of Dynamic affine term structure models (DATSMs) have been
developed since the foundational work by Vasicek (1977)\nocite%
{vasicek1977equilibrium} and Cox et al. (1985).\nocite{cox198fiftheory} They
help to understand the movements of bond yields and to analyze financial
markets. DATSMs are empirically appealing for their smooth tractability and
simple characterization of how risks get priced. There are many studies
employing this framework. To list a few: Cochrane and Piazzesi (2005)\nocite%
{cochrane2005bond} apply affine term structure models to study time
variation in expected excess bond returns using a single powerful
explanatory factor; Wu and Xia (2016)\nocite{wu2016measuring} use affine
models to summarize the macroeconomic effects of unconventional monetary
policy; Ang and Piazzesi (2003)\nocite{ang2003no} investigate how macro
economic variables affect bond prices and the dynamics of the yield curve,
Buraschi and Jiltsov (2005)\nocite{buraschi2005inflation} study the
properties of the nominal and real risk premia of the term structure of
interest rates and \nocite{golinski2016long}Goli{\'{n}}ski and Zaffaroni
(2016) incorporate long memory state variables into the term structure
model. We adopt the DATSMs setup developed by \nocite{adrian2013pricing}%
Adrian et al. (2013) which nests a general class of linear asset pricing
models and can be regarded as a linear asset pricing model with time-varying
risk premia and dynamic factors.

Many recent studies have developed approaches to estimate DATSMs. Most of
them involve a time-consuming numerical optimization procedure which results
from the high non-linearity. The inference concerning the coefficients
suffers similar challenges. Another undesirable feature, as pointed out by,
e.g., Hamilton and Wu (2012)\nocite{hamilton2012identification}, is that the
identification can be problematic. Lack of identification, e.g., due to
unspanned factors (\nocite{adrian2013pricing}Adrian et al. (2013)), and thus
a relatively flat surface of the likelihood also leads to unsatisfying
inference results (e.g., Kan and Zhang (1999),\nocite{kan1999two} \nocite%
{gospodinov2017spurious}Gospodinov et al. (2017), \nocite%
{kleibergen2009tests}Kleibergen (2009), Hamilton and Wu (2012)\nocite%
{hamilton2012identification}, \nocite{dovonon2013testing}Dovonon and Renault
(2013), Beaulieu et al. (2016)\nocite{beaulieu2016less}, Khalaf and Schaller
(2016),\nocite{khalaf2016identification} \nocite{cattaneo2022beta}Cattaneo
et al. (2022)).

Unspanned factors refer to those that only affect the dynamics of bond
prices under the historical measure but not the risk neutral one (see, e.g.,
Joslin et al. (2014)\nocite{joslin2014risk}). The presence of such factors
has been documented in empirical studies (e.g., Ludvigson and Ng (2009)%
\nocite{ludvigson2009macro}, \nocite{adrian2013pricing}Adrian et al.
(2013)). They can lead to identification challenges since varying the
parameters of the risk neutral pricing measure - risk premia parameters -
associated with these factors does not strongly influence bond prices. 
\nocite{adrian2013pricing}Adrian et al. (2013) allow for the presence of
unspanned factors with the prior knowledge of knowing which factors are
unspanned. Their proposed estimation procedure differs for cases with and
without unspanned factors.

Because of the identification issues, traditional inference methods based on
t-tests and Wald statistics become unreliable for conducting inference on
the risk prices in DATSMs (e.g., \nocite{stock2000gmm}Stock and Wright
(2000) \nocite{kleibergen2005testing}Kleibergen (2005), \nocite%
{antoine2020testing}Antoine and Renault (2020), \nocite%
{andrews2012estimation}Andrews and Cheng (2012), \nocite%
{antoine2022identification}Antoine and Lavergne (2022)). We therefore
propose easy-to-implement identification robust test procedures which are
valid even when the model is not identified due to unspanned factors. The
test procedures we provide can be used to study the time-varying risk-premia
for linear asset pricing models. Our proposed inference procedures use the
framework presented in \nocite{adrian2013pricing}Adrian et al. (2013), where
the risk of bond prices is modeled as a linear functional in observed
factors. The risk of bond prices can then be decomposed into two parts: a
time-constant and a time-varying part. We propose statistics for testing
hypothezes specified on all parameters of the time-varying components and on
just subsets of them.

The paper is organized as follows: Section 2 introduces the DATSM. Section 3
states the three step estimation procedure from Adrian et al. (2013)\nocite%
{adrian2013pricing} and shows that it is sensitive to identification issues.
Section 3 also shows the empirical relevance of these identification issues
using the data from Adrian et al. (2013). Section 4 introduces the
identification robust tests for the time-varying risk premia. It conducts a
small simulation experiment and applies then for an one time-varying risk
factor setting. Section 6 introduces the identification robust tests for
testing hypothezes on subsets of the time-varying risk premia. It applies
them to a variety of multi-factor settings using the data from Adrian et al.
(2013). The final sixth section concludes.

We use the following notation throughout the paper: $\otimes \ $and vec$%
(\cdot )$ represent respectively the Kronecker product and vectorization
operator; $\Sigma ^{\frac{1}{2}}$ is the lower triangular Cholesky
decomposition of the positive definite symmetric matrix $\Sigma $ such that $%
\Sigma =\Sigma ^{\frac{1}{2}}\Sigma ^{\frac{1}{2}^{\prime }}$; for a $%
N\times K$ dimensional full rank matrix $A:P_{A}=A(A^{\prime
}A)^{-1}A^{\prime }$ and $M_{A}=I_{N}-P_{A}\text{.}$

\section{Dynamic affine term structure models (DATSMs)}

We briefly discuss the popular class of DATSMs with observed factors.
Instead of working directly with the implied yields on an $n$-period bond as
usually done in the term structure literature, we make use of the excess
holding return of an $n$-period bond as in \nocite{adrian2013pricing}Adrian
et al. (2013).

We first illustrate the model set-up following \nocite{adrian2013pricing}%
Adrian et al. (2013) and thereafter consider tests associated with the risk
premia. For $P_{t,n}$ the price at time $t$ of a zero-coupon bond maturing
at time $t+n\text{,}$ the pricing kernel, $M_{t+1},$ is such that 
\begin{equation*}
P_{t,n}=E_{t}(M_{t+1}P_{t+1,n-1})\text{.}
\end{equation*}%
For $r_{t}$ the one-period short rate and $\lambda _{t}$ the market price of
risk, the pricing kernel is assumed exponential affine in innovation factors 
$v_{t}\sim _{i.i.d}N(0,\Sigma _{v}):$ 
\begin{equation*}
M_{t+1}=\exp \left( -r_{t}-\frac{1}{2}\lambda _{t}^{\prime }\lambda
_{t}-\lambda _{t}^{\prime }\Sigma _{v}^{-\frac{1}{2}}v_{t+1}\right) ,
\end{equation*}%
where the market price of risk $\lambda _{t}$ is an affine function of the
observed factors $X_{t}:$ 
\begin{equation*}
\lambda _{t}=\Sigma _{v}^{-\frac{1}{2}}\left( \lambda _{0}+\Lambda
_{1}X_{t}\right) ,
\end{equation*}%
with $\lambda _{0}$ and $\Lambda _{1}$ resp. a $K$-dimensional vector and a $%
K\times K$ dimensional matrix, and the $K$-dimensional vector of state
variables $X_{t}$ results from a VAR(1): 
\begin{equation*}
X_{t+1}=\mu +\Phi_1 X_{t}+v_{t+1}\text{.}
\end{equation*}%
For the one-period (log) excess holding return of a $n$-period bond at $t+1:$
\begin{equation*}
r_{t+1,n}=\ln \left( P_{t+1,n}\right) -\ln \left( P_{t,n+1}\right) -r_{t}%
\text{,}
\end{equation*}%
with $r_{t}=\ln P_{t,1}$, the structure of the pricing kernel implies that: 
\begin{equation*}
E_{t}\left[ \exp \left( -r_{t+1,n}-\frac{1}{2}\lambda _{t}^{\prime }\lambda
_{t}-\lambda _{t}^{\prime }\Sigma _{v}^{-\frac{1}{2}}v_{t+1}\right) \right]
=1.
\end{equation*}%
Assuming that $(r_{t+1,n}\text{,}$ $v_{t+1})$ are jointly normal, Adrian et
al. (2013) show that: 
\begin{equation*}
E_{t}\left( r_{t+1,n}\right) =\beta ^{(n)\prime }\left( \lambda _{0}+\Lambda
_{1}X_{t}\right) -\frac{1}{2}var(r_{t+1,n})\text{,}
\end{equation*}%
with $\beta ^{(n)}=\Sigma _{v}^{-1}cov(v_{t+1},r_{t+1,n})\text{.}$ When
decomposing, $R_{t+1,n}$ into a component correlated with $v_{t+1}$ and an
uncorrelated component/prediction error $e_{t+1,n}:$ 
\begin{equation*}
\begin{array}{c}
r_{t+1,n}-E_{t}\left( r_{t+1,n}\right) =\beta ^{(n)\prime }v_{t+1}+e_{t+1,n};%
\end{array}%
\end{equation*}%
so%
\begin{equation*}
\begin{array}{c}
r_{t+1,n}=\beta ^{(n)\prime }\left( \lambda _{0}+\Lambda _{1}X_{t}\right)
+g^{(n)}(\beta ,\Sigma _{v},\Sigma _{e})+\beta ^{(n)\prime }v_{t+1}+e_{t+1,n}%
\text{,}%
\end{array}%
\end{equation*}%
where $g^{(n)}(\beta ,\Sigma _{v},\Sigma _{e})=-\frac{1}{2}var(r_{t+1,n})%
\text{,}$ thus, for example, in case $\Sigma _{e}=$var($e_{t+1,n})=\sigma
_{e}^{2}\text{,}$ $g^{(n)}(\beta ,\Sigma _{v},\Sigma _{e})=-\frac{1}{2}%
\left( \beta ^{(n)\prime }\Sigma _{v}\beta ^{(n)}+\sigma _{e}^{2}\right) 
\text{. }$Additional restrictions are often imposed on the parameters
because of the cross-sectional term structure, but these restrictions are
dropped in \nocite{adrian2013pricing}Adrian et al. (2013)'s approach.
Assumption \ref{assum: model specification} summarizes the model setting.

\begin{assumption}
\label{assum: model specification}~\newline
(a) Consider a $K\times 1$ vector of factors $X_{t},$ $t=0,\cdots ,T$ that
results from a stationary vector autoregression of order 1: 
\begin{equation*}
\begin{array}{c}
X_{t+1}=\mu +\Phi_1 X_{t}+v_{t+1},%
\end{array}%
\end{equation*}%
where $v_{t}$ are the innovation shocks (or innovation factors). The log
excess holding return $r_{t+1,n-1}$ satisfies: 
\begin{equation*}
{\small 
\begin{array}{c}
r_{t+1,n-1}=\beta ^{(n-1)^{\prime }}\left( \lambda _{0}+\Lambda
_{1}X_{t}\right) +g^{(n-1)}\left( \beta ,\Sigma _{v},\Sigma _{e}\right)
+\beta ^{(n-1)^{\prime }}v_{t+1}+e_{t+1,n-1}%
\end{array}%
}
\end{equation*}%
with $g^{(n-1)}(\cdot )$ a parametric function. Furthermore:%
\begin{equation*}
(v_{t+1}^{\prime },e_{t+1}^{\prime })^{\prime }|\left\{ X_{s}\right\}
_{s=0}^{t}\sim i.i.d.N\left( 0,\text{diag}(\Sigma _{v},\Sigma _{e})\right) .
\end{equation*}%
(b) For $\Sigma _{e}=\sigma _{e}^{2},$ $g^{(n-1)}\left( \beta ,\Sigma
_{v},\Sigma _{e}\right) =\frac{1}{2}\left( \beta ^{(n-1)^{\prime }}\Sigma
_{v}\beta ^{(n-1)}+\sigma _{e}^{2}\right) $.
\end{assumption}


\section{Regression estimator and Wald based inference}

To estimate the price of risk, Adrian et al. (2013) propose a three-step
procedure akin to the two pass procedure from \nocite{fama1973risk}Fama and
MacBeth (1973):

\begin{enumerate}
\item Estimate: 
\begin{equation*}
X_{t+1}=\mu +\Phi_1 X_{t}+v_{t+1}\text{,}
\end{equation*}%
by least squares to obtain $\hat{\mu}\text{,}$ $\hat{\Phi}\text{,}$ $\hat{v}%
_{t}=X_{t}-\hat{\mu}-\hat{\Phi}_1X_{t-1}\text{,}$ $t=1,\ldots ,T$ and $\hat{%
\Sigma}_{v}=\frac{1}{T}\sum_{t=1}^{T}\hat{v}_{t}\hat{v}_{t}^{\prime }\text{.}
$

\item Estimate: 
\begin{equation*}
r_{t+1,n}=a^{(n)}+d^{(n)\prime }X_{t}+\beta ^{(n)\prime }\hat{v}%
_{t+1}+e_{t+1,n}\text{,}
\end{equation*}%
by least squares to obtain $\hat{a}^{(n)}\text{,}$ $\hat{d}^{(n)}$ and $\hat{%
\beta}^{(n)}\text{,}$ $n=1,\ldots ,N\text{.}$

\item Construct $\hat{a}=(\hat{a}^{(1)}\ldots \hat{a}^{(N)})^{\prime },\text{
}$ $\hat{\beta}=(\hat{\beta}^{(1)}\ldots $ $\hat{\beta}^{(N)})^{\prime },$ $%
\hat{d}=(\hat{d}^{(1)}\ldots d^{(n)})^{\prime },$ $\hat{g}=(\hat{g}%
^{(1)}\ldots \hat{g}^{(n)})$ for $\hat{g}^{(n)}=g^{(n)}(\hat{\beta},\hat{%
\Sigma}_{v},\hat{\Sigma}_{e}),$ $n=1,\ldots ,N\text{,}$ and estimate $%
\lambda _{0}$ and $\Lambda _{1}$ using: 
\begin{equation*}
\begin{array}{cl}
\hat{\lambda}_{0}= & \left( \hat{\beta}^{\prime }\hat{\beta}\right)
^{^{\prime }-1}\hat{\beta}^{\prime }(\hat{a}+\hat{g}+\hat{\beta}\hat{\mu})
\\ 
\hat{\Lambda}_{1}= & \left( \hat{\beta}^{\prime }\hat{\beta}\right)
^{^{\prime }-1}\hat{\beta}^{\prime }(\hat{d}+\hat{\beta}\hat{\Phi}).%
\end{array}%
\end{equation*}
\end{enumerate}

The three-step procedure essentially regresses transformed returns on
estimated $\beta $'s. Two issues can hamper the reliability of the final
stage: the sampling error of estimates resulting from previous stages and
the quality of the $\beta $'s. The first issue is negligible when the
information dominates the asymptotically vanishing sampling errors, and only
slight modifications are necessary for the asymptotic variance estimator
used in test statistics to accommodate for it. However, when modeling the
unspanned factors, \nocite{adrian2013pricing}Adrian et al. (2013) show that
the entries in the $\beta $'s corresponding to unspanned factors are zero
(Assumption \ref{assum:unspanned factors}), which leads to identification
problems (e.g., \nocite{kleibergen2009tests}Kleibergen (2009), Beaulieu et
al. (2013)\nocite{beaulieu2013identification},\nocite%
{kleibergen2015unexplained}Kleibergen and Zhan (2015)).

\begin{assumption}
\label{assum:unspanned factors} Potential existence of (nearly) unspanned
factors: $\beta ^{\prime }=(B,C)$, $B$ represents the spanned factors and is
of full rank $K_{B}\leq K$ while $C=O(1/\sqrt{T})$ reflects the (nearly)
unspanned factors. If $K_{B}=K,$ there are no (nearly) unspanned factors.
\end{assumption}

\nocite{adrian2013pricing}Adrian et al. (2013) show that Assumption \ref%
{assum:unspanned factors} embeds the unspanned factors that do not affect
the dynamics of bonds under the historical pricing measure. The rows of $%
\beta $ comprised of close to zero values correspond to unspanned factors. 
\nocite{adrian2013pricing} Adrian et al. (2013) assume that the location of
these rows is known, and their three-step estimation procedure excludes the
unspanned factors in the second step by only including spanned factors. The
third step would otherwise encounter identification issues resulting from a
classical multicollinearity problem. It is straightforward to show that zero 
$\beta $'s lead to an identification problem because any value of the $%
\lambda $'s associated with the zero elements in $\beta $ goes for the
values of excess returns. When we instead just have small $\beta $'s, which
are comparable in magnitude to the estimation error, we similarly encounter
such an identification problem (e.g., \nocite{kleibergen2009tests}Kleibergen
(2009), \nocite{antoine2009efficient}Antoine and Renault (2009, 2012), 
\nocite{antoine2012efficient} \nocite{kleibergen2020robust}Kleibergen and
Zhan (2020), \nocite{kleibergen2019identification}Kleibergen et al. (2022)).

Proposition \ref{prop2} states some well-known results from the weak
identification literature. It shows that the risk premia estimator $\hat{%
\Lambda}$ becomes inconsistent in the presence of weak factors because it
converges to a non-normal distribution. This results since varying the value
of the risk premia associated with the weak factors does not change the
excess returns. The asymptotic distribution of the conventional Wald
statistic for testing the null hypothesis $H_{0}:\Lambda _{1}=\Lambda
_{1}^{0}$ then no longer converges to a $\chi ^{2}$ -distribution, and the
same holds for subset testing based on this estimator. Therefore, the
conventional test statistics can be misleading in the presence of unspanned
factors.

\begin{proposition}
\label{prop2}~\newline
Under Assumptions \ref{assum: model specification}.(a), denote $\Lambda
=[\lambda _{0},\Lambda _{1}]$:\newline
(a) If $\beta $ is of full rank: 
\begin{equation*}
\begin{array}{c}
\sqrt{T}\left[ 
\begin{matrix}
\text{vec}\left( \hat{\beta}^{\prime }-\beta ^{\prime }\right) \\ 
\text{vec}\left( \hat{\Lambda}-\Lambda \right)%
\end{matrix}%
\right] \rightarrow _{d}N\left( \left[ 
\begin{matrix}
0 \\ 
0%
\end{matrix}%
\right] ,\left[ 
\begin{matrix}
\mathcal{V}_{\beta } & \mathcal{C}_{\Lambda ,\beta }^{\prime } \\ 
\mathcal{C}_{\Lambda ,\beta } & \mathcal{V}_{\Lambda }%
\end{matrix}%
\right] \right) ,%
\end{array}
\label{eq:propo1}
\end{equation*}%
where $\mathcal{V}_{\beta },\mathcal{C}_{\Lambda ,\beta },\mathcal{V}%
_{\Lambda }$ are specified in the proof.\newline
(b) If Assumption \ref{assum:unspanned factors} holds with $K_{B}<K$, then $%
\hat{\Lambda}\rightarrow _{d}\Lambda +\epsilon,$ where $\epsilon $ follows a
non-standard distribution, so $\hat{\Lambda}$ no longer converges to the
true value $\Lambda $ at rate $\sqrt{T}$.
\end{proposition}

\begin{proof}
	See the Online Supplementary Appendix.
\end{proof}

We focus on inference concerning the time-varying component of risk prices $%
\Lambda _{1}$. We therefore demean the one-period (log) excess holding
returns by subtracting its time-series average: 
\begin{equation*}
\bar{r}_{t+1,n}=\beta ^{(n)\prime }\left( \Lambda _{1}\bar{X}_{t}\right)
+\beta ^{(n)\prime }\bar{v}_{t+1}+\bar{e}_{t+1,n}\text{,}
\end{equation*}%
with $\bar{z}_{t+1,n}=z_{t+1,n}-\bar{z},$ with $\bar{z}=\frac{1}{T}%
\sum_{t=1}^{T}z_{t},$ for $z=r,$ $X,$ $v$ and $e$ resp. When stacking the
equations for $N$ different maturities: 
\begin{equation*}
R_{t+1}=\beta \left( \Lambda _{1}\bar{X}_{t}\right) +\beta \bar{v}%
_{t+1}+e_{t+1}\text{,}
\end{equation*}%
where $R_{t+1}=(\bar{r}_{t+1,1}\ldots \bar{r}_{t+1,N})^{\prime }\text{,}$ $%
\beta =(\beta ^{(1)}\ldots \beta ^{(n)})^{\prime }\text{,}$ $e_{t}=(\bar{e}%
_{t,1}\ldots \bar{e}_{t,N})^{\prime }\text{,}$ the pricing equation closely
resembles the beta-pricing model for the return on (portfolios of) assets: 
\begin{equation*}
\text{r}_{t+1}=\beta \lambda +\beta \text{F}_{t+1}+u_{t+1}\text{,}
\end{equation*}%
with r$_{t}$ an $N$-dimensional vector with the returns on $N$ assets, $%
\beta $ the $N\times K$ dimensional beta matrix and $F_{t}$ a $K$
-dimensional vector of risk factors. A further important similarity that
both models imply is the reduced rank structure that becomes obvious using a
slight re-specification: 
\begin{equation*}
\begin{array}{ccc}
R_{t+1}=\beta \left( \Lambda _{1}\ \text{{}}\vdots \ \text{{}}I_{K}\right)
\left( 
\begin{array}{c}
\bar{X}_{t} \\ 
\bar{v}_{t+1}%
\end{array}%
\right) +e_{t} &  & 
\end{array}%
\end{equation*}%
and 
\begin{equation*}
\begin{array}{ccc}
\text{r}_{t+1}=\beta \left( \lambda \ \text{{}}\vdots \ \text{{}}%
I_{K}\right) \left( 
\begin{array}{c}
1 \\ 
F_{t+1}%
\end{array}%
\right) +u_{t}, &  & 
\end{array}%
\end{equation*}%
where the $N\times 2K$ and $N\times (K+1)$ dimensional matrices $\beta
(\Lambda _{1}\ \text{{}}\vdots \ \text{{}}I_{K})$ and $\beta (\lambda \ 
\text{{}}\vdots \ \text{{}}I_{K})$ are each at most of rank $K\text{,}$ so
except for the largest $K$ singular values, all, $K$ and 1 resp., singular
values of these matrices are zero. Further for $\Lambda _{1}$ and $\lambda $
to be well defined, $\beta $ should be of full rank. When $\beta $ is near a
reduced rank value, or in other words, if some factors are weak/unspanned,
we encounter an identification issue which is also reflected by more than
just $K$ or 1 resp. singular values of the above matrices being equal or
close to zero.\bigskip

\noindent {{\small 
\begin{tabular}{c|c|c|c|c|cc}
\hline\hline
& $\hat{\beta}_{1}$ & $\hat{\beta}_{2}$ & $\hat{\beta}_{3}$ & $\hat{\beta}%
_{4}$ & $\hat{\beta}_{5}$ &  \\ \hline
{\ (1)} & {\ -0.0094} & {\ 0.0031 } & {\ -0.0008 } & {\ 0.0002 } & {\ 0.0000}
&  \\ 
& (-3.6027) & (2.5058) & (-1.4886) & (0.4344) & (0.0347) &  \\ 
{\ (2)} & {\ -0.0213 } & {\ 0.0057 } & {\ -0.0007 } & {\ -0.0003} & {\
0.0002 } &  \\ 
& (-8.2107) & (4.5416) & (-1.2228) & (-0.7341) & (0.5816) &  \\ 
{\ (3)} & {\ -0.0446} & {\ 0.0070} & {\ 0.0010 } & {\ -0.0005} & {\ -0.0001}
&  \\ 
& (-17.1745) & (5.6040) & (1.8889) & (-1.3535) & (-0.4551) &  \\ 
{\ (4)} & {\ -0.0656} & {\ 0.0048 } & {\ 0.0024 } & {\ 0.0000} & {\ -0.0003}
&  \\ 
& (-25.2648) & (3.8500) & (4.5279) & (0.0386) & (-0.9643) &  \\ 
{\ (5)} & {\ -0.0843} & {\ 0.0003} & {\ 0.0028 } & {\ 0.0007 } & {\ -0.0001}
&  \\ 
& (-32.4670) & (0.2023) & (5.2857) & (1.7192) & (-0.1732) &  \\ 
{\ (6)} & {\ -0.1011 } & {\ -0.0059} & {\ 0.0022} & {\ 0.0010 } & {\ 0.0003}
&  \\ 
& (-38.9474) & (-4.6910) & (4.1584) & (2.6824) & (1.0713) &  \\ 
{\ (7)} & {\ -0.1164 } & {\ -0.0130 } & {\ 0.0008} & {\ 0.0010} & {\ 0.0006}
&  \\ 
& (-44.8511) & (-10.3464) & (1.5602) & (2.5469) & (1.9102) &  \\ 
{\ (8)} & {\ -0.1305 } & {\ -0.0206} & {\ -0.0011} & {\ 0.0005 } & {\ 0.0005}
&  \\ 
& (-50.2742) & (-16.4072) & (-2.0105) & (1.2971) & (1.8297) &  \\ 
{\ (9)} & {\ -0.1435 } & {\ -0.0284} & {\ -0.0033} & {\ -0.0003} & {\ 0.0002}
&  \\ 
& (-55.2826) & (-22.6118) & (-6.1006) & (-0.8978) & (0.6354) &  \\ 
{\ (10)} & {\ -0.1556 } & {\ -0.0361} & {\ -0.0056} & {\ -0.0015} & {\
-0.0005} &  \\ 
& (-59.9307) & (-28.7667) & (-10.3483) & (-3.7935) & (-1.6457) &  \\ 
{\ (11)} & {\ -0.1669} & {\ -0.0436} & {\ -0.0078} & {\ -0.0028} & {\
-0.0014 } &  \\ 
& (-64.2706) & (-34.7284) & (-14.4913) & (-7.1319) & (-4.8558) &  \\ \hline
& [0.0000] & [0.0000] & [0.0000] & [0.0000] & [0.0001] &  \\ \hline\hline
\end{tabular}
}} ~\bigskip

\noindent {\small \textbf{Table 1}: The least squares estimates of the $%
\beta $'s associated with the excess returns for bonds with 11 different
maturities of $6,12,18\cdots ,60$ and 84, 120 months over the sample period
1987:01-2011:12, and the factors are the first five principal components
generated using the cross-section of bond yields for maturities $3,\cdots
,120$ months (same data taken from Adrain et al. (2013)). Based on
Proposition \ref{prop2}.(a), we report LS estimates of $\beta $'s with
t-statistics in round brackets; and$p$-values of $F$-tests in square
brackets, testing the null hypothesis $H_{0}:\beta _{j}=0$ that each column
is jointly zero.. The Kleibergen-Paap rank statistic (see Kleibergen and
Paap (2006)) testing $H_{0}:$ rank($\beta $)=4: 1.6561 [0.9764]. }

\paragraph{Illustrating unspanned and weak factors in an empirical study}

In reality, there is no direct prior knowledge of the number of unspanned or
weak factors. To show their presence and empirical relevance, we use the
data from \nocite{adrian2013pricing}Adrian et al. (2013), i.e., the zero
coupon yield data constructed by G\"{u}rkaynak et al. (2007).\nocite%
{gurkaynak2007us} Table 1 shows the factor loading estimates and relevant
tests corresponding to the five principal component (PCA) factors employed
in Adrian et al. (2013). Table 1 shows that many elements of $\beta $ are
small and not statistically significant from zero at the 5\% significance
level. Though the $F$-tests on the columns of $\beta $ are significant, the
rank test of the $\beta $-matrix indicates potential identification
problems, because for these factor loadings we can not reject the reduced
rank null hypothesis at the 5\% significance level.
 
\begin{figure}[htbp!]
	\includegraphics[width=\columnwidth]{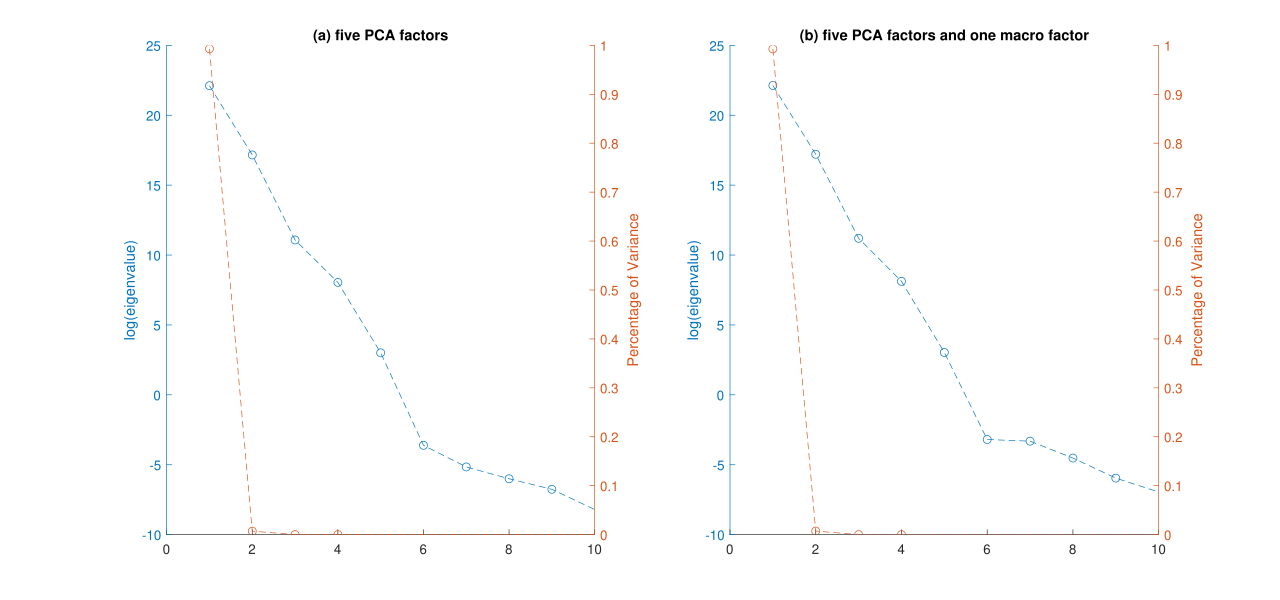}
	\caption*{{\small \textbf{Figure 1}: (log) singular values (blue) and the
			percentage of variance explained (red) of $\displaystyle\frac{1}{T}
			\sum_{t=1}^{T}R_{t}\left( 
			\begin{array}{c}
				\bar{X}_{t} \\ 
				\hat{v}_{t+1}%
			\end{array}
			\right) ^{\prime }\left[ \frac{1}{T}\sum_{t=1}^{T}\left( 
			\begin{array}{c}
				\bar{X}_{t} \\ 
				\hat{v}_{t+1}%
			\end{array}
			\right) \left( 
			\begin{array}{c}
				\bar{X}_{t} \\ 
				\hat{v}_{t+1}%
			\end{array}
			\right) ^{\prime }\right] .$ $R_{t}$ uses the demeaned excess returns of the
			same bonds as in Table 1, $\bar{X}_{t}$ involves the five PCA factors in (a)
			and uses the five PCA factors with one additional macro factor in (b). }}
\end{figure}

\bigskip

Figure 1.(a) uses the same data set as Table 1 and shows a scree plot of the
(log) singular values of 
\begin{equation*}
\frac{1}{T}\sum_{t=1}^{T}R_{t}\left( 
\begin{array}{c}
\bar{X}_{t} \\ 
\hat{v}_{t+1}%
\end{array}%
\right) ^{\prime }\left[ \frac{1}{T}\sum_{t=1}^{T}\left( 
\begin{array}{c}
\bar{X}_{t} \\ 
\hat{v}_{t+1}%
\end{array}%
\right) \left( 
\begin{array}{c}
\bar{X}_{t} \\ 
\hat{v}_{t+1}%
\end{array}%
\right) ^{\prime }\right] ,
\end{equation*}%
which estimates $\beta (\Lambda _{1}\ \text{{}}\vdots \ \text{{}}I_{K})$. In
case of strong spanned factors, the smallest $K,$ is five, singular values
should be close to zero while Figure 1.(a) shows that the smallest six
singular values are close to zero which indicates a weak/unspanned factor
problem which is further reflected by the rank test not rejecting rank($%
\beta $)=4 in Table 1 at the 5\% significance level.{~}

~\bigskip

\noindent {{\small \ 
\begin{tabular}{c|c|c|c|c|c|c}
\hline\hline
& $\hat{\beta}_{1}$ & $\hat{\beta}_{2}$ & $\hat{\beta}_{3}$ & $\hat{\beta}
_{4}$ & $\hat{\beta}_{5}$ & $\hat{\beta}_{\mathit{macro}}$ \\ \hline
{\ (1)} & -0.0094 & 0.0032 & -0.0008 & 0.0002 & 0.0000 & 0.0001 \\ 
& (-3.4930) & (2.4109) & (-1.4749) & (0.4253) & (0.0309) & (0.0516) \\ 
{\ (2)} & -0.0213 & 0.0057 & -0.0007 & -0.0003 & 0.0002 & 0.0000 \\ 
& (-7.9353) & (4.3501) & (-1.2115) & (-0.7259) & (0.5717) & (0.0254) \\ 
{\ (3)} & -0.0446 & 0.0070 & 0.0010 & -0.0005 & -0.0001 & -0.0000 \\ 
& (-16.5873) & (5.3631) & (1.8715) & (-1.3372) & (-0.4497) & (-0.0037) \\ 
{\ (4)} & -0.0656 & 0.0049 & 0.0024 & 0.0000 & -0.0003 & 0.0001 \\ 
& (-24.4149) & (3.6980) & (4.4871) & (0.0346) & (-0.9529) & (0.0655) \\ 
{\ (5)} & -0.0843 & 0.0003 & 0.0028 & 0.0007 & -0.0001 & 0.0001 \\ 
& (-31.3724) & (0.2100) & (5.2383) & (1.6937) & (-0.1732) & (0.0801) \\ 
{\ (6)} & -0.1011 & -0.0059 & 0.0022 & 0.0010 & 0.0003 & 0.0000 \\ 
& (-37.6206) & (-4.4799) & (4.1208) & (2.6462) & (1.0539) & (0.0333) \\ 
{\ (7)} & -0.1164 & -0.0130 & 0.0008 & 0.0010 & 0.0006 & -0.0000 \\ 
& (-43.3100) & (-9.8999) & (1.5456) & (2.5142) & (1.8813) & (-0.0252) \\ 
{\ (8)} & -0.1305 & -0.0206 & -0.0011 & 0.0005 & 0.0005 & -0.0001 \\ 
& (-48.5424) & (-15.7013) & (-1.9930) & (1.2811) & (1.8021) & (-0.0512) \\ 
{\ (9)} & -0.1435 & -0.0284 & -0.0033 & -0.0003 & 0.0002 & -0.0000 \\ 
& (-53.3870) & (-21.6288) & (-6.0457) & (-0.8872) & (0.6243) & (-0.0204) \\ 
{\ (10)} & -0.1557 & -0.0361 & -0.0056 & -0.0015 & -0.0005 & 0.0001 \\ 
& (-57.8990) & (-27.4937) & (-10.2541) & (-3.7506) & (-1.6271) & (0.0745) \\ 
{\ (11)} & -0.1670 & -0.0435 & -0.0078 & -0.0028 & -0.0014 & 0.0002 \\ 
& (-62.1305) & (-33.1560) & (-14.3589) & (-7.0557) & (-4.7981) & (0.2297) \\ 
\hline
& [0.0000] & [0.0000] & [0.0000] & [0.0000] & [0.0002] & [1.0000] \\ 
\hline\hline
\end{tabular}%
}} ~\bigskip

\noindent {\small \textbf{Table 2}: Using the same data as in Table 1 with
one additional macro factor (real activity) that constructed following Ang
and Piazzesi (2003), we report estimates of the $\beta $'s with t-statistics
in round brackets; and $p$-values of $F$-tests in square brackets, testing
the null hypothesis $H_{0}:\beta _{j}=0$ that each column is jointly zero.
Kleibergen-Paap rank statistic testing $H_{0}:$ rank($\beta $ )=5: 0.0025
[1.0000]. \bigskip }

As previously noted in the literature (e.g., \nocite{kleibergen2020robust}
Kleibergen and Zhan (2020), Kleibergen et al. (2022)), weak identification
issues are often present when macro factors are used. Following \nocite%
{ang2003no}Ang and Piazzesi (2003), we construct one macro factor, the real
activity measure, which is the first principal component resulting from four
variables that capture real US macro activity: the "Help Wanted Advertising
in Newspapers (HELP)"\footnote{%
We use the HELP-Wanted index from \nocite{barnichon2010building}Barnichon
(2010) to match the time periods of the excess returns.} index, unemployment
(UE), the growth rate of employment (EMPLOY), and the growth rate of
industrial production (IP). As shown in Table 2, the macro factor is much
less correlated with returns and thus is more likely to result in an
identification issue. This is further reflected in Figure 1(b) containing
the scree plot which shows that while there are six factors, the smallest
seven singular values are close to zero. Table 2 also shows a tiny value for
the rank test which provides another indication of a weak/unspanned factor
problem. \bigskip

The identification problems revealed in Figures 1-2 and Tables 1-2 not only
affect the validity of the estimators but also the reliability of
traditional inference procedures. \nocite{adrian2013pricing} Adrian et al.
(2013) assume that the positions of the zero rows in $\beta $ are known when
dealing with unspanned factors. We remain agnostic about this and provide
testing procedures concerning $\Lambda _{1}$ which are identification robust
without the need of prior knowledge of the unspanned factors.

\section{Identification robust tests of time-varying risk premia}

The identification robust tests are based on the sample moment vector. The
sample moment vector results from the observed factors having no predictive
power for the prediction error, $\bar{e}_{t+1,n},$ so our sample moment
vector for $\Lambda _{1}$ is: 
\begin{equation*}
\begin{array}{c}
f_{T}(\Lambda _{1},X)=\frac{1}{T}\sum_{t=1}^{T}(\bar{X}_{t-1}\otimes (\bar{R}%
_{t}-\hat{\beta}\hat{V}_{t}))-\left( \hat{Q}_{XX}\otimes \hat{\beta}\right) 
\text{vec(}\Lambda _{1}),%
\end{array}%
\end{equation*}%
with $\hat{Q}_{XX}=\frac{1}{T}\sum_{t=1}^{T}\bar{X}_{t-1}\bar{X}%
_{t-1}^{\prime },$ and its derivative with respect to vec($\Lambda _{1})$ is 
\begin{equation*}
\begin{array}{c}
q_{T}(X)=-\left( \hat{Q}_{XX}\otimes \hat{\beta}\right) \text{.}%
\end{array}%
\end{equation*}%
We next make an assumption regarding the large sample behavior of the sample
moment vector and its derivative.

\begin{assumption}
\label{assum3}~\newline
Under $H_{0}:\Lambda _{1}=\Lambda _{1}^{0}$, 
\begin{equation*}
\sqrt{T}\left( 
\begin{array}{c}
f_{T}(\Lambda _{1}^{0},X) \\ 
\text{vec}(q_{T}(X)-J)%
\end{array}%
\right) \underset{d}{\rightarrow }\left( 
\begin{array}{c}
\psi _{f} \\ 
\psi _{q}%
\end{array}%
\right) \text{,}
\end{equation*}%
where the Jacobian $J=-(Q_{XX}\otimes \beta ),$ $\hat{Q}_{XX}\underset{p}{%
\rightarrow }Q_{XX}\text{,}$ and $\psi _{f}$ and $\psi _{q}$ are $N\times K$
and $NK\times K^{2}$ dimensional random vectors:%
\begin{equation*}
\left( 
\begin{array}{c}
\psi _{f} \\ 
\psi _{q}%
\end{array}%
\right) \sim N\left( 0,V(\Lambda _{1}^{0})\right) \text{,}
\end{equation*}%
with 
\begin{equation*}
V(\Lambda _{1}^{0})=\left( 
\begin{array}{cc}
V_{ff}(\Lambda _{1}^{0}) & V_{qf}(\Lambda _{1}^{0})^{\prime } \\ 
V_{qf}(\Lambda _{1}^{0}) & V_{qq}(\Lambda _{1}^{0})%
\end{array}%
\right) ,
\end{equation*}%
where $V_{ff}(\Lambda _{1}^{0}),$ $V_{qf}(\Lambda _{1}^{0})$ and $%
V_{qq}(\Lambda _{1}^{0})$ are $NK\times NK,$ $NK^{3}\times NK$ and $%
NK^{3}\times NK^{3}$ dimensional matrices.
\end{assumption}

Assumption \ref{assum3} is a high-level assumption which resembles
Assumption 1 in Kleibergen (2005) and holds under mild conditions. \nocite%
{kleibergen2005testing} Assumption \ref{assum3} holds true irrespective of
Assumption \ref{assum:unspanned factors}. Assumption \ref{assum: model
specification} is sufficient for Assumption \ref{assum3}, but our proposed
test statistics can be applied to more general cases than the model implied
in Assumption \ref{assum: model specification}. For our setting: 
\begin{equation*}
\psi _{f}=\psi _{\bar{f}}+\Psi _{q}\text{vec(}\Lambda _{1}^{0})\text{,}
\end{equation*}%
with $\Psi _{q}=$vecinv($\psi _{q})$ and 
\begin{equation*}
\begin{array}{c}
\sqrt{T}\left( \frac{1}{T}\sum_{t=1}^{T}(\bar{X}_{t-1}\otimes (\bar{R}_{t}-%
\hat{\beta}\hat{V}_{t}))-\left( Q_{XX}\otimes \beta \right) \text{vec}%
(\Lambda _{1}^{0})\right) \underset{d}{\rightarrow }\psi _{\bar{f}}\text{.}%
\end{array}%
\end{equation*}%
We also have 
\begin{equation*}
\psi _{q}=\text{vec}((Q_{XX}\otimes \Psi _{\beta })+(\Psi _{XX}\otimes \beta
))\text{,}
\end{equation*}%
where 
\begin{equation*}
\begin{array}{cl}
\sqrt{T}\text{vech}(\frac{1}{T}\sum_{t=1}^{T}\bar{X}_{t}\bar{X}_{t}^{\prime
}-Q_{XX})\underset{d}{\rightarrow }\psi _{XX}\text{,} & \Psi _{XX}=\text{
vechinv}(\psi _{XX})\text{,} \\ 
\sqrt{T}\text{vec}(\hat{\beta}-\beta )\underset{d}{\rightarrow }\psi _{\beta
}\text{ ,} & \Psi _{\beta }=\text{vecinv}(\psi _{\beta })\text{,}%
\end{array}%
\end{equation*}%
with vech$(A)$ containing the unique elements of a symmetric matrix $A$.
Since $\psi _{q}$ has $NK^{3}$ elements while the number of unique elements
in $\Psi _{\beta }$ and $\Psi _{XX}$ equals $NK+\frac{1}{2}K(K+1)\text{,}$
the joint normal distribution of $(\psi _{f}\text{,}$ $\psi _{q})$ is
further allowed to be degenerate. When $v_{t}$ is normal (as assumed in
Assumption \ref{assum: model specification}) or its third moment equals
zero, $\psi _{\beta }$ and $\psi _{XX}$ are also independently distributed.

The identification robust statistics use an estimator of the Jacobian whose
limit behavior under H$_{0}:\Lambda _{1}=\Lambda _{1}^{0}$ is independent of
the limit behavior of the sample moment, see Kleibergen (2005): 
\begin{equation*}
\begin{array}{rl}
\hat{D}_{T}(\Lambda _{1},X)= & (\hat{D}_{1,T}(\Lambda _{1},X)\ldots \hat{D}%
_{1,T}(\Lambda _{k},X)) \\ 
\text{vec(}\hat{D}_{T}(\Lambda _{1},X))= & \text{vec}(q_{T}(X))-\hat{V}%
_{qf}(\Lambda _{1})\hat{V}_{ff}(\Lambda _{1})^{-1}f_{T}(\Lambda _{1},X) \\ 
\sqrt{T}\text{vec(}\hat{D}_{T}(\Lambda _{1}^{0},X)-J)\underset{d}{%
\rightarrow } & \psi _{q.f}\sim N(0,V_{qq.f}(\Lambda _{1}^{0}))%
\end{array}%
\end{equation*}%
with $V_{qq.f}(\Lambda _{1})=V_{qq}-V_{qf}(\Lambda _{1})V_{ff}(\Lambda
_{1})^{-1}V_{qf}(\Lambda _{1})^{\prime }\text{,}$ $\hat{V}_{qf}(\Lambda
_{1}) $ and $\hat{V}_{ff}(\Lambda _{1})$ consistent estimators of $%
V_{qf}(\Lambda _{1}^{0})$ and $V_{ff}(\Lambda _{1}^{0}),$ and $\psi _{q.f}$
independent of $\psi _{f}.$

We can next define the identification robust Factor Anderson-Rubin (FAR),
(Kleibergen) Lagrange multiplier (KLM) and JKLM\ statistics for testing H$%
_{0}:\Lambda _{1}=\Lambda _{1}^{0}:$ 
\begin{equation*}
\begin{array}{cl}
\text{FAR(}\Lambda _{1}^{0})= & T\times f_{T}(\Lambda _{1}^{0},X)^{\prime }%
\hat{V}_{ff}(\Lambda _{1}^{0})^{-1}f_{T}(\Lambda _{1}^{0},X)\underset{d}{%
\rightarrow }\chi ^{2}(KN) \\ 
\text{KLM(}\Lambda _{1}^{0})= & T\times f_{T}(\Lambda _{1}^{0},X)^{\prime }%
\hat{V}_{ff}(\Lambda _{1}^{0})^{-\frac{1}{2}}P_{\hat{V}_{ff}(\Lambda
_{1}^{0})^{-\frac{1}{2}}\hat{D}_{T}(\Lambda _{1}^{0},X)} \\ 
\text{\thinspace \thinspace } & \text{\quad \quad \quad \quad \quad \quad }%
\hat{V}_{ff}(\Lambda _{1}^{0})^{-\frac{1}{2}}f_{T}(\Lambda _{1}^{0},X)%
\underset{d}{\rightarrow }\chi ^{2}(K^{2}) \\ 
\text{JKLM(}\Lambda _{1}^{0})= & T\times f_{T}(\Lambda _{1}^{0},X)^{\prime }%
\hat{V}_{ff}(\Lambda _{1}^{0})^{-\frac{1}{2}}M_{\hat{V}_{ff}(\Lambda
_{1}^{0})^{-\frac{1}{2}}\hat{D}_{T}(\Lambda _{1}^{0},X)} \\ 
\text{\thinspace \thinspace } & \text{\quad \quad \quad \quad \quad \quad }%
\hat{V}_{ff}(\Lambda _{1}^{0})^{-\frac{1}{2}}f_{T}(\Lambda _{1}^{0},X)%
\underset{d}{\rightarrow }\chi ^{2}(K(N-K)).%
\end{array}%
\end{equation*}%
The limiting distributions are a direct result of Assumption \ref{assum3}
and do not depend on the rank of the Jacobian or $\beta $, so the limiting
distributions hold regardless of Assumption \ref{assum:unspanned factors}.

\subsection{Illustrative simulation and empirical results}

We conduct a single-factor model simulation study to illustrate the
performance of the proposed robust joint tests. For the data generating
process (DGP), we consider 
\begin{equation*}
R_{t}=c+\beta \left( \Lambda _{1}X_{t-1}+v_{t}\right) +e_{t}\text{,}
\end{equation*}%
where the parameters are calibrated to data from Adrian et al. (2013). In
particular, we fix the sample size to be $T=300$ and use the eleven excess
returns as in Table 1. We calibrate with the first PCA factor from Adrian et
al. (2013) to mimic the strong identified case and the third one for the
weak identification setting. Figure 2 shows power curves of the conventional
t-statistic and the robust test statistics in both strong and weakly
identified cases. For both settings, FAR, KLM and JKLM tests are all size
correct, while the Wald test is size distorted under weak identification and
size correct but biased for stong identification. For weak identification,
the KLM test has some power loss away from the hypothesized value because of
which it is preferred to combine it in a conditional or unconditional manner
with the J-test to improve power, see Moreira (2003) and Kleibergen (2005).

We use the robust tests to analyze the time-varying component of the risk
premium. A detailed description of the involved excess returns and risk
factors has been discussed previously for Tables 1-2. Figure 3 shows the $p$%
-values for testing the risk premium associated with all six factors in a
single factor model with the identification robust tests. A $p$-value larger
than, say, 5\%, implies that we could not reject the null at the 5\%
significance level.\bigskip

\begin{figure}[htbp!]
	\includegraphics[width=\columnwidth]{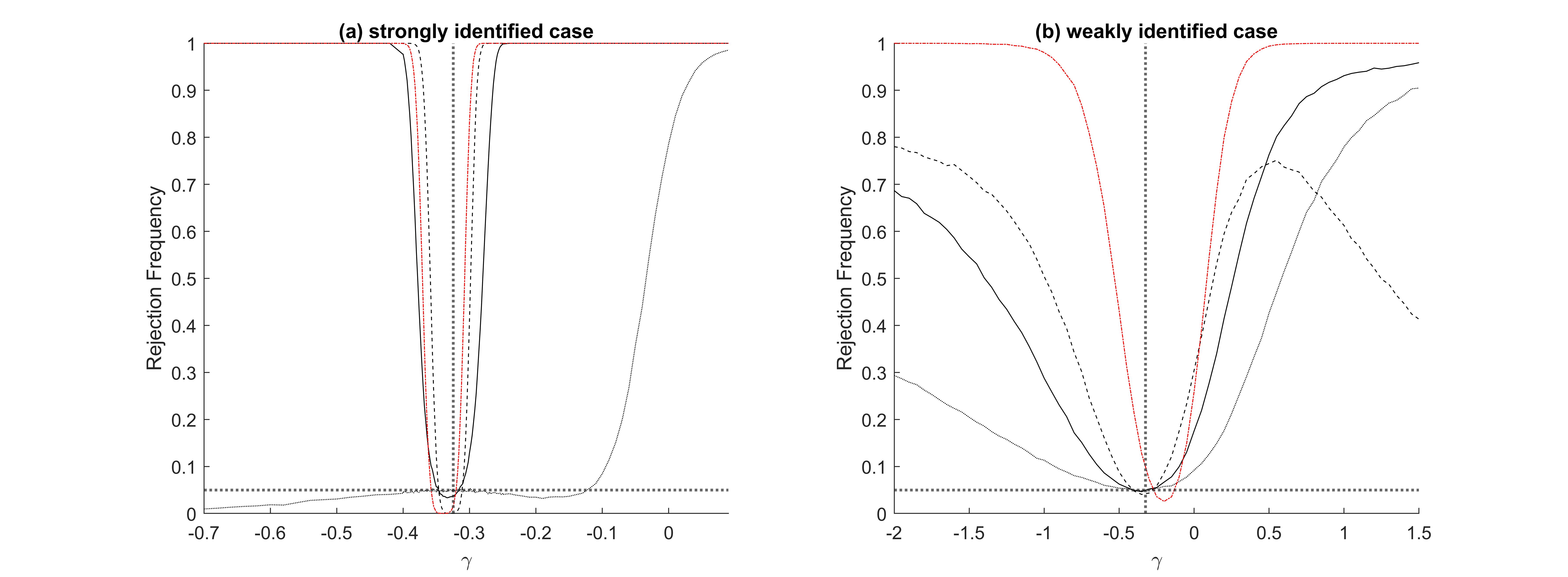}
	\caption*{{\small \textbf{Figure 2}: the left panel (a) is the strong
			identified case, while the right panel (b) is the weak identified case.
			Power curves (rejection frequencies) of Wald (dash-dotted, red), FAR
			(solid), KLM (dashed), and JKLM (dotted) testing the $\Lambda _{1}$ (scalar)
			with values on the horizontal axis; horizontal dashed line at 5\% and
			vertical one at the calibrated value of $\Lambda _{1}$.}}
\end{figure}
  \bigskip

Figure 3 shows that for most cases the JKLM test leads to unbounded 95\%
confidence sets even for strong factors, e.g., the first PCA factor, since
the $p$-value curves are above the 5\% line over the whole interval of
analyzed values of the risk premium which is consistent with the smaller
power observed in our simulation exercises and results since the JKLM test
primarily tests misspecification. For all factors, the FAR and KLM tests
provide bounded 95\% confidence sets since only for bounded regions the $p$%
-values are above the 5\% line, even for those potentially weak factors such
as the fifth PCA factor and the macro factor. The latter implies that in a
single factor setting, all these risk premia are identified. For the
high-order PCA factors, the robust tests, however, result in 95\% confidence
sets that differ from those resulting from the Wald test. Most striking is
that a zero value for the risk premium is not rejected for strong factors
such as the first and second PCA factors but rejected for potentially weak
factors. For example, the null hypothesis that $\Lambda _{1}=0$ is rejected
by the FAR and KLM test for both the fifth factor and the macro factor. This
is partly in line with Adrian et al. (2013), which highlight the role of the
higher-order principal components as the time variation may be largely
driven by, e.g., the fifth principal component. Therefore, some factors may
be weak but have some importance for interpreting the (time-varying)
expected returns.

\begin{figure}[htbp!]
	\includegraphics[height=0.8\textheight,width=\columnwidth]{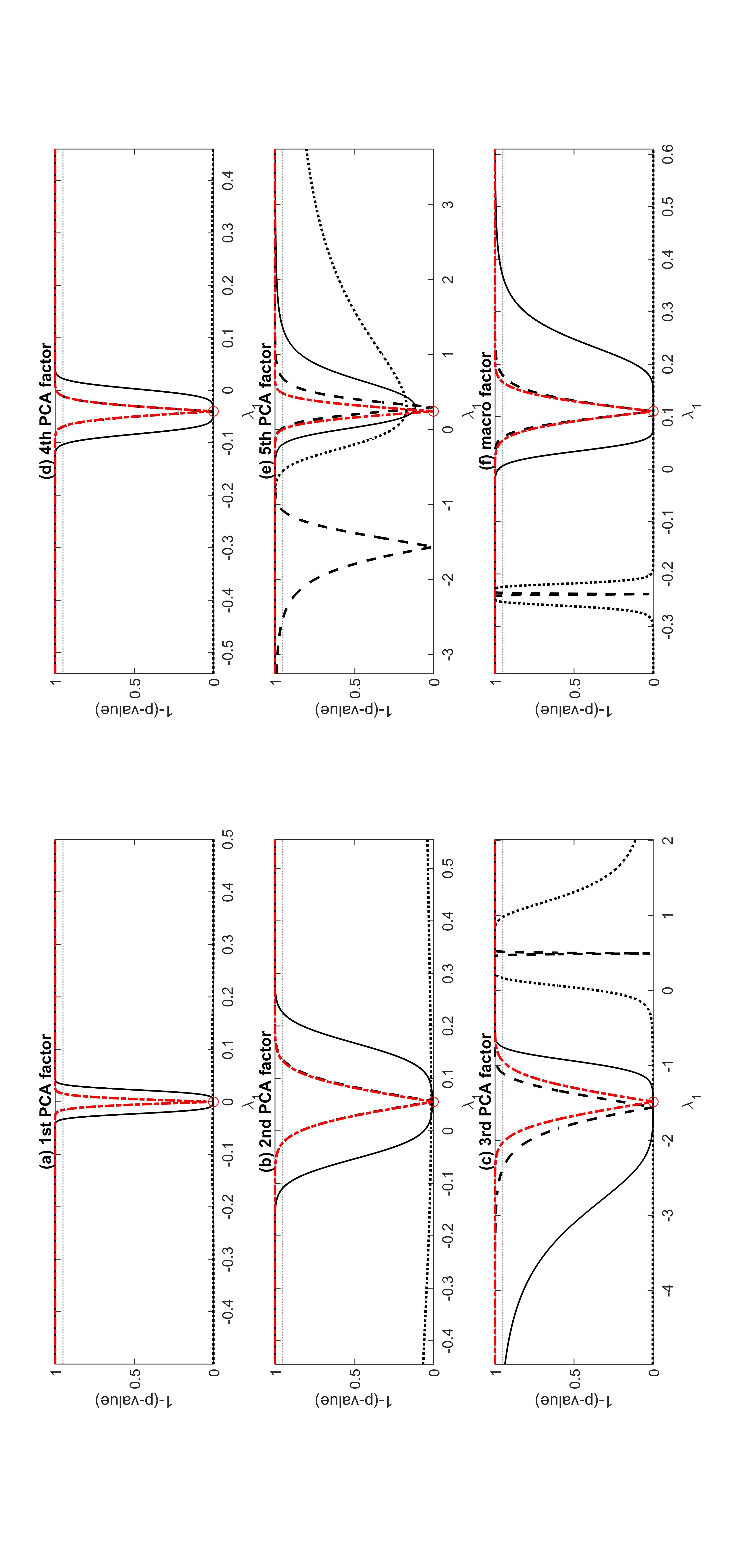}
	\caption*{{\small \textbf{Figure 3}: one-minus-$p$-value curves of Wald (dash-dotted,
			red); FAR (solid); KLM (dashed); JKLM (dotted) for testing the $\Lambda _{1}$
			(scalar) in a single factor model with values on the horizontal axis and
			dotted line at 95\%. This figure uses the same data as in Table 2.}}
\end{figure}

\newpage
\section{Identification robust sub-vector testing}

The identification robust tests introduced in the previous section are for
testing hypotheses specified on all elements of $\Lambda _{1}.$ We are often
interested in testing hypotheses specified on just subsets of the
parameters. When we analyze multi-factor models, testing whether or not a
certain factor risk premium exhibits time variation would require testing a
specific row of $\Lambda _{1}$, while testing whether a factor drives the
time variation would require to test the corresponding column of $\Lambda
_{1}$. Under our current settings, projection-based versions of the
identification robust tests would allow us to test such hypotheses whilst
preserving the size of the test, see \nocite{dufour2005projection}Dufour and
Taamouti (2005). These tests, however, lead to reduced power so we extend
the robust subset FAR test (sFAR) of \nocite{guggenberger2012asymptotic}%
Guggenberger et al. (2012) for testing hypotheses on all elements in a row
of $\Lambda _{1}$ which can be similarly extended to test for all elements
in a column of $\Lambda _{1}$.

Without loss of generality, we consider testing the hypothesis that the risk
premia associated with one specific factor, say the first, are all equal to $%
\lambda _{1}^{0}$: 
\begin{equation*}
\text{H}_{0}:\lambda _{1}=\lambda _{1}^{0},
\end{equation*}%
for $\Lambda _{1}=\binom{\lambda _{1}^{0\prime }}{\Lambda _{2}},$ $\lambda
_{1}:K\times 1,$ $\Lambda _{2}:(K-1)\times K.$ Under H$_{0}:\lambda
_{1}=\lambda _{1}^{0},$ the $N\times 2K$ dimensional reduced rank parameter
matrix in the equation for the stacked returns becomes: 
\begin{equation*}
\Phi =\left( \beta _{1}\ \text{{}}\vdots \ \text{{}}\beta _{2}\right) \left( 
\begin{array}{ccc}
\lambda _{1}^{0\prime } & 1 & 0 \\ 
\Lambda _{2} & 0 & I_{K-1}%
\end{array}%
\right) ,
\end{equation*}%
so post-multiplying by $\left(\begin{matrix}
	{I_{K}} &{-\lambda _{1}^{0}}&{0} \\
	 {0}&0&{I_{K-1}} 
\end{matrix}   \right)^{\prime } 
$ yields the $N\times (2K-1)$ matrix:%
\begin{equation*}
\Phi \left( 
\begin{array}{cc}
I_{K} & 0 \\ 
-\lambda _{1}^{0\prime } & 0 \\ 
0 & I_{K-1}%
\end{array}%
\right) =\left( \beta _{1}\ \text{{}}\vdots \ \text{{}}\beta _{2}\right)
\left( 
\begin{array}{cc}
0 & 0 \\ 
\Lambda _{2} & I_{K-1}%
\end{array}%
\right) =\beta _{2}\left( 
\begin{array}{cc}
\Lambda _{2} & I_{K-1}%
\end{array}%
\right) ,
\end{equation*}%
which, since the rank of $\beta _{2}\left( 
\begin{array}{cc}
\Lambda _{2} & I_{K-1}%
\end{array}%
\right) $ equals $K-1,$ shows that H$_{0}$ implies that the smallest $K$
singular values of $\Phi $ times $\left(\begin{matrix}
	{I_{K}} &{-\lambda _{1}^{0}}&{0} \\
	{0}&0&{I_{K-1}} 
\end{matrix}   \right) ^{\prime }$ equal zero. The sFAR
statistic for testing H$_{0}:\lambda _{1}=\lambda _{1}^{0}:$ 
\begin{equation*}
\begin{array}{cc}
\text{sFAR(}\lambda _{1})= & \min_{\Lambda _{2}}\text{FAR}(\Lambda
_{1}(\lambda _{1}^{0},\Lambda _{2})),%
\end{array}%
\end{equation*}%
therefore corresponds with a rank test of H$_{0}:$ rank$\left( \Phi \left( 
\begin{array}{cc}
I_{K} & 0 \\ 
-\lambda _{1}^{0\prime } & 0 \\ 
0 & I_{K-1}%
\end{array}%
\right) \right) =K-1.\bigskip $

The bounding distribution of the limiting distribution of the sFAR statistic
relies upon a Kronecker product structure (KPS) asymptotic covariance matrix
of the least squares estimator of the linear model (\nocite%
{guggenberger2012asymptotic} see Guggenberger et al. (2012)): 
\begin{equation*}
\hat{\Phi}=\frac{1}{T}\sum_{t=1}^{T}R_{t}\left( 
\begin{array}{c}
\bar{X}_{t} \\ 
\hat{v}_{t+1}%
\end{array}%
\right) ^{\prime }\left[ \frac{1}{T}\sum_{t=1}^{T}\left( 
\begin{array}{c}
\bar{X}_{t} \\ 
\hat{v}_{t+1}%
\end{array}%
\right) \left( 
\begin{array}{c}
\bar{X}_{t} \\ 
\hat{v}_{t+1}%
\end{array}%
\right) ^{\prime }\right] ^{-1}=\left( \hat{d}\ \text{{}}\vdots \ \text{{}}%
\hat{\beta}\right) .
\end{equation*}%
The KPS thus concerns the asymptotic variance of 
\begin{equation*}
\begin{array}{cc}
\frac{1}{\sqrt{T}}\sum_{t=1}^{T}\left( 
\begin{array}{c}
\bar{X}_{t} \\ 
\hat{v}_{t+1}%
\end{array}%
\right) \otimes \tilde{e}_{t}, & \tilde{e}_{t}=e_{t}+\beta (\bar{v}_{t}-\hat{%
v}_{t}).%
\end{array}%
\end{equation*}%
We note that $\hat{v}_{t}$ is not directly observed so it adds additional
sampling error when imputing estimates of $\hat{v}_{t}.$ To implement the
sFAR test, we therefore make the following assumption.

\begin{assumption}
\label{assum:KPS}There exists $\Omega \in \mathbb{R}^{2K\times 2K}$ and $%
\Sigma \in \mathbb{R}^{N\times N}$ symmetric positive definite matrices such
that $S=\Omega \otimes \Sigma $ and 
\begin{equation*}
\begin{array}{c}
\frac{1}{\sqrt{T}}\sum_{t=1}^{T}\left( \left( 
\begin{array}{c}
\bar{X}_{t} \\ 
\hat{v}_{t+1}%
\end{array}%
\right) \otimes \tilde{e}_{t}\right) \rightarrow _{d}N(0,S).%
\end{array}%
\end{equation*}
\end{assumption}

The asymptotic normality stated in Assumption \ref{assum:KPS} is a direct
result of Assumption \ref{assum: model specification}. If $\hat{v}_{t}$ is
directly observed, so $\tilde{e}_{t}=e_{t}$, Assumption \ref{assum: model
specification} implies that $\Omega =\mathbb{E}\left( (\bar{X}_{t}^{\prime }%
\text{ }\vdots \text{ }\hat{v}_{t+1}^{\prime })^{\prime }(\bar{X}%
_{t}^{\prime }\text{ }\vdots \text{ }\hat{v}_{t+1}^{\prime })\right) $ and $%
\Sigma =\text{var}(e_{t})$. Because of the additional sampling error due to
the generated regressor $\hat{v}_{t+1}$, Assumption \ref{assum:KPS} does,
however, not provide the exact specifications of $\Omega $ and $\Sigma $.
~\bigskip

\noindent {{\small \ 
\begin{tabular}{c|c|c|c|cc}
\hline\hline
& (1) & (2) & (3) & (4) &  \\ \hline
KPST & 212.0808 & 201.8551 & 197.1613 & 205.8054 &  \\ 
$p$-value & [0.0511] & [ 0.1265] & [ 0.1809] & [0.0910] &  \\ \hline\hline
\end{tabular}
}} ~\bigskip

\noindent {\small \textbf{Table 3}: KPS test (KPST) statistics for testing
the null hypothesis that $H_{0}:S=\Omega \otimes \Sigma $ for some $\Omega
\in \mathbb{\ \ R}^{2K\times 2K}$ and $\Sigma \in \mathbb{R}^{N\times N}$
symmetric positive definite matrices. All four cases use excess returns on
bonds with maturities 3, 12, 24, 60, 90, 120 months from Adrian et al.
(2013). For the factors $X_{t}$, (1) uses the macro factor (real activity)
and the level factor (first PCA factor), (2) uses the level factor (first
PCA factor) and the slope factor (second PCA factor), (3) uses the macro
factor (real activity) and the slope factor (second PCA factor) and (4) uses
the macro factor (real activity) and the curvature factor (third PCA
factor). \bigskip }

We use the KPS test (KPST) from \nocite{guggenberger2022test}Guggenberger et
al. (2022) to test for the proximity of a KPS matrix to the covariance
matrix, $S.$ Table 3 reports the KPST results, and shows that the KPS
restriction for $S$ is a realistic assumption since none of these tests
reject the null hypothesis that the covariance matrix has a KPS at the 5\%
significance level. A by-product of the KPS test is the KPS factorization
for $\hat{S},$ see \nocite{guggenberger2022test}Guggenberger et al.(2022).

\begin{proposition}
\label{prop:KPS covariance} Under Assumption \ref{assum:KPS}, and when there
is a consistent estimator for $S$, $\hat{S}$, then in large samples $\hat{S}%
\approx (\hat{\Omega}\otimes \hat{\Sigma}),$ where 
\begin{equation*}
\hat{\Omega}=\displaystyle \text{vecinv}\left( \left( 
\begin{array}{l}
\hat{L}_{11} \\ 
\hat{L}_{21}%
\end{array}%
\right) /\hat{L}_{11}\right),\hat{\Sigma}=\text{vecinv}(\hat{L}_{11}\hat{%
\sigma}_{1}\hat{N}_{1}^{\prime }),
\end{equation*}
and $\hat{L}_{11},\hat{L}_{21},\hat{\sigma}_{1},\hat{N}_{1}$ are specified
in the proof. The KPS covariance estimator $\hat{\Omega}\otimes \hat{\Sigma}$
provides a consistent estimator for $S$.
\end{proposition}

\begin{proof}
		See Appendix.
	\end{proof}

Because of the KPS covariance structure, $\hat{\Omega}\otimes \hat{\Sigma},$
we can compute the sFAR statistic using the characteristic polynomial stated
in Proposition \ref{prop:sFAR}.

\begin{proposition}
\label{prop:sFAR} Under Assumptions \ref{assum3} and \ref{assum:KPS}, let $%
\hat{V}_{\hat{\Phi}}=(\hat{\Psi}\otimes \hat{\Sigma}),$ $\hat{\Psi}=\hat{W}%
^{-1}\hat{\Omega}\hat{W}^{-1}=\left( 
\begin{array}{cc}
\hat{\Psi}_{X} & \hat{\Psi}_{XV} \\ 
\hat{\Psi}_{VX} & \hat{\Psi}_{V}%
\end{array}%
\right) $, $\hat{W}=\frac{1}{T}\sum_{t=1}^{T}\binom{\bar{X}_{t}}{\hat{v}%
_{t+1}}\binom{\bar{X}_{t}}{\hat{v}_{t+1}}^{\prime },$ sFAR$(\lambda
_{1}^{0}),$ for testing $\text{H}_{0}:\lambda _{1}=\lambda _{1}^{0},$ for $%
\Lambda _{1}=\binom{\lambda _{1}^{\prime }}{\Lambda _{2}},$ $\lambda
_{1}:K\times 1,$ $\Lambda _{2}:(K-1)\times K,$ equals $T$ times the sum of
the $K$ smallest roots of the characteristic polynomial: 
\begin{equation*}
\begin{array}{lc}
\left\vert \mu \left( 
\begin{array}{cc}
I_{K} & 0 \\ 
-\lambda _{1}^{0\prime } & 0 \\ 
0 & I_{K-1}%
\end{array}%
\right) ^{\prime }\hat{\Psi}\left( 
\begin{array}{cc}
I_{K} & 0 \\ 
-\lambda _{1}^{0\prime } & 0 \\ 
0 & I_{K-1}%
\end{array}%
\right) -\right. &  \\ 
\left. \qquad \qquad \qquad \qquad \left( 
\begin{array}{cc}
I_{K} & 0 \\ 
-\lambda _{1}^{0\prime } & 0 \\ 
0 & I_{K-1}%
\end{array}%
\right) ^{\prime }\hat{\Phi}^{\prime }\hat{\Sigma}^{-1}\hat{\Phi}\left( 
\begin{array}{cc}
I_{K} & 0 \\ 
-\lambda _{1}^{0\prime } & 0 \\ 
0 & I_{K-1}%
\end{array}%
\right) \right\vert & =0,%
\end{array}%
\end{equation*}%
and $\lim_{T\rightarrow \infty }\text{sFAR(}\lambda _{1})\prec \chi
^{2}(K(N-(K-1)).$ The bound on the limiting distribution holds regardless of
Assumption \ref{assum:unspanned factors}.
\end{proposition}

\begin{proof}
		See Appendix.
	\end{proof}
\noindent

\bigskip
 
\begin{figure}[htbp!]
	\includegraphics[width=\columnwidth,height=0.3\textheight]{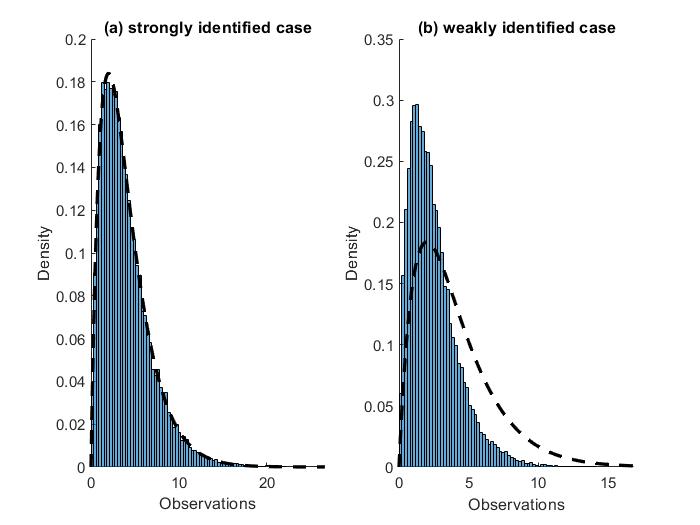}
	\caption*{{\small \textbf{Figure 4:} simulated density plots of the sFAR test
			statistic (shadowed bins) and the density function of the $\chi ^{2}$%
			-distribution (dashed black curve). The left panel (a) is the strong
			identified case, while the right panel (b) is the weak identified case. }}
\end{figure}%
\bigskip

Figure 4 illustrates the $\chi ^{2}$ bound on the limiting distribution of
the sFAR statistic stated in Proposition \ref{prop:sFAR}. The density of the
limiting distribution of the sFAR statistic is simulated for a two-factor
model, which uses the first two PCA factors to mimic the strongly identified
case and the third and the fifth PCA factors to mimic weak identification.
The limiting distribution of the sFAR statistic is $\chi ^{2}$ when the
model is strongly identified. In the weakly identified case, as shown in
Figure 4, the limiting distribution is bounded by the $\chi ^{2}$
distribution. Using $\chi ^{2}$ critical values for the sFAR test thus
controls the size of the test. \bigskip

For projection-based tests on $\lambda _{1},$ the involved test has to be
computed over a grid of points concerning the partialled out parameters.
This becomes computationally burdensome when the number of partialled out
parameters increases because of an increased dimension of $\Lambda _{1}$
resulting from more factors. It makes our proposed approach involving the
sFAR test more empirically appealing because it does not involve an
extensive grid search. In practice, Assumption \ref{assum:KPS} can be
relaxed by using the KPST as a pre-test for conducting robust subset testing
as described in Guggenberger et al. (2022).\bigskip

The value of the sFAR statistic at parameter values distant from zero
provide a diagnostic to indicate if the confidence sets of the hypothesized
parameters are bounded. These tests are therefore indicative of weak
identification.

\begin{proposition}
\label{prop:sFAR far away}\label{prop:min boundary} For tests of $\text{H}%
_{0}:\lambda _{1}=c\lambda _{1}^{0}$ with $\lambda _{1}^{0}$ a fixed vector
of length one and $c$ a scalar, $\lim_{c\rightarrow \infty }$ sFAR($c\lambda
_{1}^{0}$), the realized value of the sFAR statistic at a distant value of $%
\lambda _{1}$ in the direction of $\lambda _{1}^{0}$, equals $T$ times the
sum of the K smallest roots of 
\begin{equation*}
\begin{array}{rl}
\left\vert \left( 
\begin{array}{cc}
\bar{\lambda}_{1,\perp }^{0} & 0 \\ 
0 & I_{K}%
\end{array}%
\right) ^{\prime }\left[ \mu \hat{\Psi}-\hat{\Phi}^{\prime }\hat{\Sigma}^{-1}%
\hat{\Phi}\right] \left( 
\begin{array}{cc}
\bar{\lambda}_{1,\perp }^{0} & 0 \\ 
0 & I_{K}%
\end{array}%
\right) \right\vert & =0,%
\end{array}%
\end{equation*}%
where $\bar{\lambda}_{1,\perp }^{0}$ is a $K\times (K-1)$ orthonormal matrix
that is orthogonal to $\lambda _{1}^{0}$. The limit sFAR statistic is
uniformly bounded from below by the minimum eigenvalue of $T\hat{\Psi}%
_{V}^{-1/2\prime }\hat{\beta}^{\prime }\hat{\Sigma}^{-1}\hat{\beta}\hat{\Psi}%
_{V}^{-1/2}$.
\end{proposition}

\begin{proof}
	See Appendix.
\end{proof}Proposition \ref{prop:sFAR far away} provides a way of verifying
whether the confidence sets resulting from the sFAR statistic are bounded or
unbounded in specific directions (\nocite{dufour1997some}Dufour (1997), 
\nocite{kleibergen2020robust}Kleibergen and Zhan (2020), Kleibergen (2021), 
\nocite{khalaf2016identification}Khalaf and Schaller (2016), \nocite%
{kleibergen2019identification}Kleibergen et al. (2022)). The minimum
eigenvalue of $T\hat{\Psi}_{V}^{-1/2\prime }\hat{\beta}^{\prime }\hat{\Sigma}%
^{-1}\hat{\beta}\hat{\Psi}_{V}^{-1/2}$ is a rank test statistic concerning
the rank of the factor loading matrix $\beta $ (see \nocite%
{kleibergen2006generalized}Kleibergen and Paap (2006)), so Proposition \ref%
{prop:sFAR far away} shows that the sFAR test evaluated at distant values
relates to the rank of $\beta .$ Proposition \ref{prop:sFAR far away} also
explains that when we encounter weak identification issues with $\beta $'s
close to reduced rank, we have unbounded confidence sets. The lower bound is
sharp when $K=1,$ as indicated in the proof of {Proposition} \ref{prop:min
boundary}, for which case also Theorem 12 in \nocite{kleibergen2021efficient}%
Kleibergen (2021) applies. When $K=1$ and the Kleibergen-Paap rank test is
significant at the 5\% significance level, Proposition \ref{prop:min
boundary} implies that the sFAR test leads to bounded 95\% confidence sets
of $\lambda _{1}$. \bigskip

Table 4 reports the Kleibergen-Paap rank test for different factor settings
for the data from Adrian et al. (2013). It shows that, in line with Figure
3, all single-factor model have bounded 95\% confidence sets for the
time-varying risk premia, which are less likely to be bounded when we
include more than three factors. The fifth factor, though identified in a
single-factor setting, suffers from weak identification problems when we
include other factors. Table 4 also shows that the rank test statistic is a
good indicator of unboundedness as small values of the rank test statistics
suggest unbounded confidence sets. \newline
~\bigskip \noindent {{\small 
\begin{tabular}{ll|ll|ll|ll|ll}
\hline\hline
(1) & rank & (2) & rank & (3) & rank & (4) & rank & (5) & rank \\ 
& test &  & test &  & test &  & test &  & test \\ \hline
1* & 30075 & 1,2 & 1290 & 1,2,3 & 711.1 & 1,2,3, & 932.9 & 1,2,3, & {9.105}
\\ 
& {[0.000]} &  & {[0.000]} &  & {[0.000]} & 4 & {[0.000]} & 4,5$\dagger
\dagger $ & {[0.003]} \\ 
2* & 4409 & 1,3 & 1487 & 1,2,4 & 1317 & 1,2,3, & {5.103} &  &  \\ 
& {[0.000]} &  & {[0.000]} &  & {[0.000]} & 5$\dagger \dagger $ & {[0.078]}
&  &  \\ 
3* & 973.0 & 1,4 & 1073 & 1,2,5 & 26.03 & 1,2,4, & 76.53 &  &  \\ 
& {[0.000]} &  & {[0.000]} &  & {[0.002]} & 5 & {[0.000]} &  &  \\ 
4* & 703.2 & 1,5 & 39.82 & 1,3,4 & 603.8 & 1,3,4, & {12.36} &  &  \\ 
& {[0.000]} &  & {[0.000]} &  & {[0.000]} & 5 & {[0.002]} &  &  \\ 
5* & 74.62 & 2,3 & 864.7 & 1,3,5 & {13.22} & 2,3,4, & 16.23 &  &  \\ 
& {[0.000]} &  & {[0.000]} &  & {[0.004]} & 5 & {[0.000]} &  &  \\ 
&  & 2,4 & 990.9 & 1,4,5 & 110.3 &  &  &  &  \\ 
&  &  & {[0.000]} &  & {[0.000]} &  &  &  &  \\ 
&  & 2,5 & 63.17 & 2,3,4 & 765.8 &  &  &  &  \\ 
&  &  & {[0.000]} &  & {[0.000]} &  &  &  &  \\ 
&  & 3,4 & 825.2 & 2,3,5 & 16.96 &  &  &  &  \\ 
&  &  & {[0.000]} &  & {[0.001]} &  &  &  &  \\ 
&  & 3,5 & {11.14} & 2,4,5 & 549.0 &  &  &  &  \\ 
&  &  & {[0.025]} &  & {[0.000]} &  &  &  &  \\ 
&  & 4,5 & 445.2 & 3,4,5 & 18.05 &  &  &  &  \\ 
&  &  & {[0.000]} &  & {[0.000]} &  &  &  &  \\ \hline\hline
\end{tabular}
}} ~\newline
\noindent {\small \textbf{Table 4}: Kleibergen-Paap rank statistic testing $%
H_{0}:\text{rank}(\beta )=K-1$ ($K$ denotes the number of factors) and its
associated [$p$-value] in square brackets, for varying factor combinations.
The colum headed by (i), for i=1,\ldots ,5, states which factor combinations
are used when using i factors. All cases use excess returns on bonds with
maturities of 2, 3, 12, 60, and 120 months and different combinations of the
five PCA factors from Adrian et al. (2013). We mark with one star if the
lower bound of the limit sFAR (see Proposition \ref{prop:min boundary})
indicates bounded 95\% confidence sets in every direction, and mark with
double daggers if the associated 95\% confidence sets of the time-varying
risk premia parameters of one or more factor are unbounded. \bigskip }

\subsection{Power of the sFAR test}

To illlustrate the power of the sFAR test, we compute power curves for two
settings calibrated to the data discussed previously. Figure 6 therefore
shows the two-dimensional power curves that result when jointly testing the
two risk premia parameters associated with a single factor in a two factor
model. The left hand side of Figure 6 shows the power curves for a strongly
identified setting while the right hand side does so for a weakly identified
setting. The power curves on the right hand side show that the sFAR test is
not consistent for weakly identified settings since the rejection
frequencies do not converge to one when we move away from the hypothesized
value.\bigskip
 
\begin{figure}[htbp!]
	\includegraphics[width=\columnwidth]{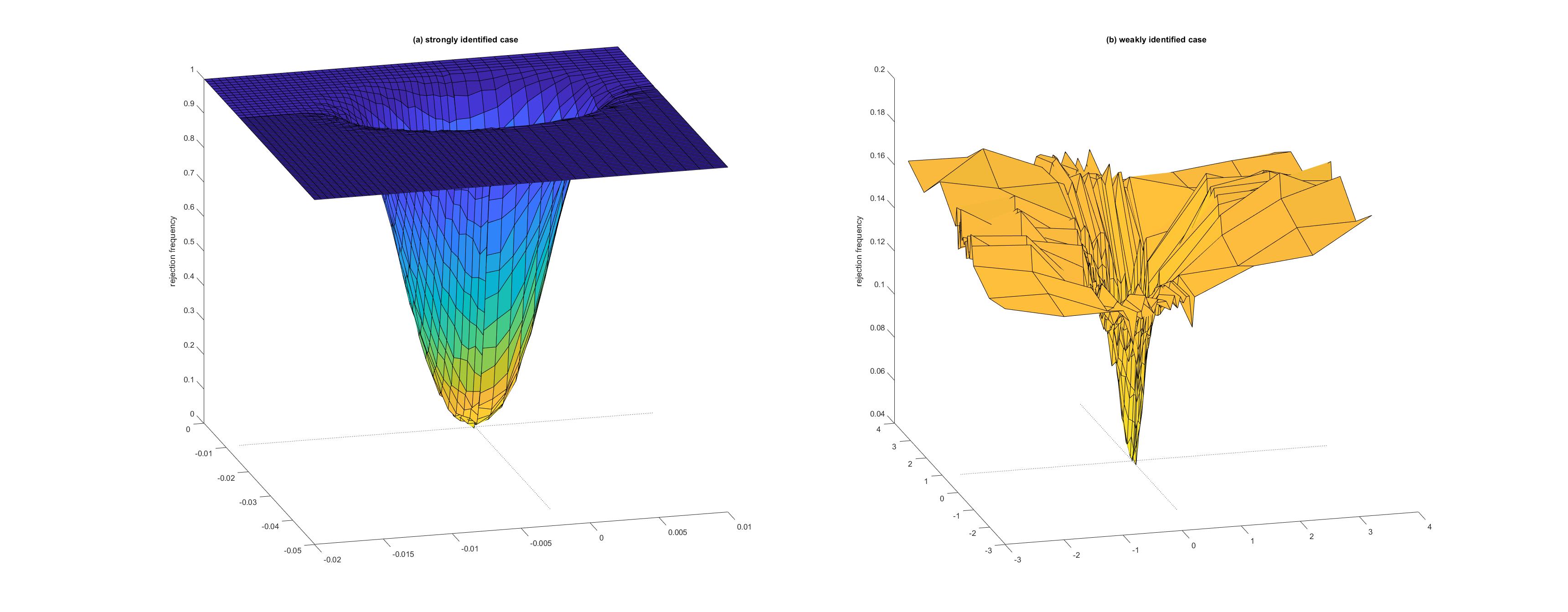}
	\caption*{{\small \textbf{Figure 6}: Simulated power surfaces (rejection frequencies)
			of sFAR tests on $\Lambda _{1}$'s w.r.t. the first factor of a two-factor
			model: the left panel (a) is a strong identified setting calibrated to the
			two-factor model with excess returns of bonds with maturities 3, 60, 120
			months using the first and second PCA factors , while the right panel (b) is
			a weakly identified setting calibrated to the two-factor model with excess
			returns of bonds with maturities 3, 60, 120 months using the third and fifth
			PCA factors. Dotted lines $(\gamma _{i},y,0.05)$ are drawn to mark the
			positions of the calibrated risk premia values, $\gamma _{i}$'s, at 5\%
			level.}}
\end{figure}

\subsection{Identification robust confidence sets for risk premia}

We use the sFAR test to construct confidence sets on the risk premia
resulting from two and three factor models. Figure 7 shows the 90, 95 and
99\% joint confidence sets that result for the two risk premia resulting for
one specific factor in a two factor model using the data from Adrian et al.
(2013) while Figure 8 does so for the three risk premia resulting for one
specific factor in a three factor model. Size correct confidence sets for
the individual risk premia result by projecting the joint confidence sets on
the axes. When using four or more factors, the number of risk premia
concerning one factor is at least four so we have to use projection-based
tests based on the sFAR statistic to be able to visualize these confidence
sets.\ For expository purposes and since Table 4 shows that some of these
confidence sets are unbounded, for example, the one that results when using
all five factors, we therefore refrain from using more than three factors.%
\footnote{%
The rank tests in Table 4 and Figures 7-8 are not identical. For expository
purposes, we choose a smaller number of test assets in Figures 7-8.}

Figure 7 shows all two dimensional confidence sets for the two risk premia
resulting for one factor for all different specifications using the five PCA
factors discussed previously in a two factor model. The two dimensional
confidence sets in Figure 2 vary a lot. Quite a few are empty so all values
of the parameters are rejected at significance levels which exceed 99\%.
This occurs, for example, when using the first and either the second, third
and fourth PCA factor so the model is misspecified. There are also settings
where the confidence set is bounded and well behaved which occurs, for
example, when using the third and fourth PCA factor. Other confidence sets
are unbounded and/or cover the whole two-dimensional space, which occurs,
for instance, when using the third and fifth PCA. For this combination the
90\% confidence set for the two risk premia on the third factor is unbounded
but excludes an area in the parameter space while the 90\% confidence set of
the two risk premia on the fifth factor covers the whole two-dimensional
space. Table 4 also shows that the combination of the third and fifth PCA
factors leads to a smaller rank test statistic than other factor
combinations when using two-factor models which is in line with the
unbounded confidence set in Figure 7 which relates to the third and fifth
PCA factors. This is all indicative of weak identification when using both
the third and fifth factors.

Figure 8 shows the joint confidence sets for the three risk premia
associated with a single factor in a three factor model. The first column of
Figure 8 does for a factor model containing the first three PCAs as factors
while the second column does so using the first, third and fifth PCAs as
factors. Unlike when using two factors, the first column shows that the
confidence sets are no longer empty but bounded which shows that the risk
premia for the first three PCA factors are well identified and that the
model is no longer misspecified. This is confirmed by the p-value of the
rank test on the $\beta $'s. This is in contrast when using the first, third
and fifth PCAs as factors. The confidence sets in the second column of
Figure 8 are namely all unbounded indicating weak identification of the risk
premia which is further reflected by the $p$-value of the rank test on the $%
\beta $'s. Table 4 also shows that the model including the first, third, and
fifth PCA factors has a much smaller rank test statistic than one of the
first three PCA factors within the three-factor model.~\newline
 
\begin{figure}[htbp!]
	\includegraphics[width=\columnwidth]{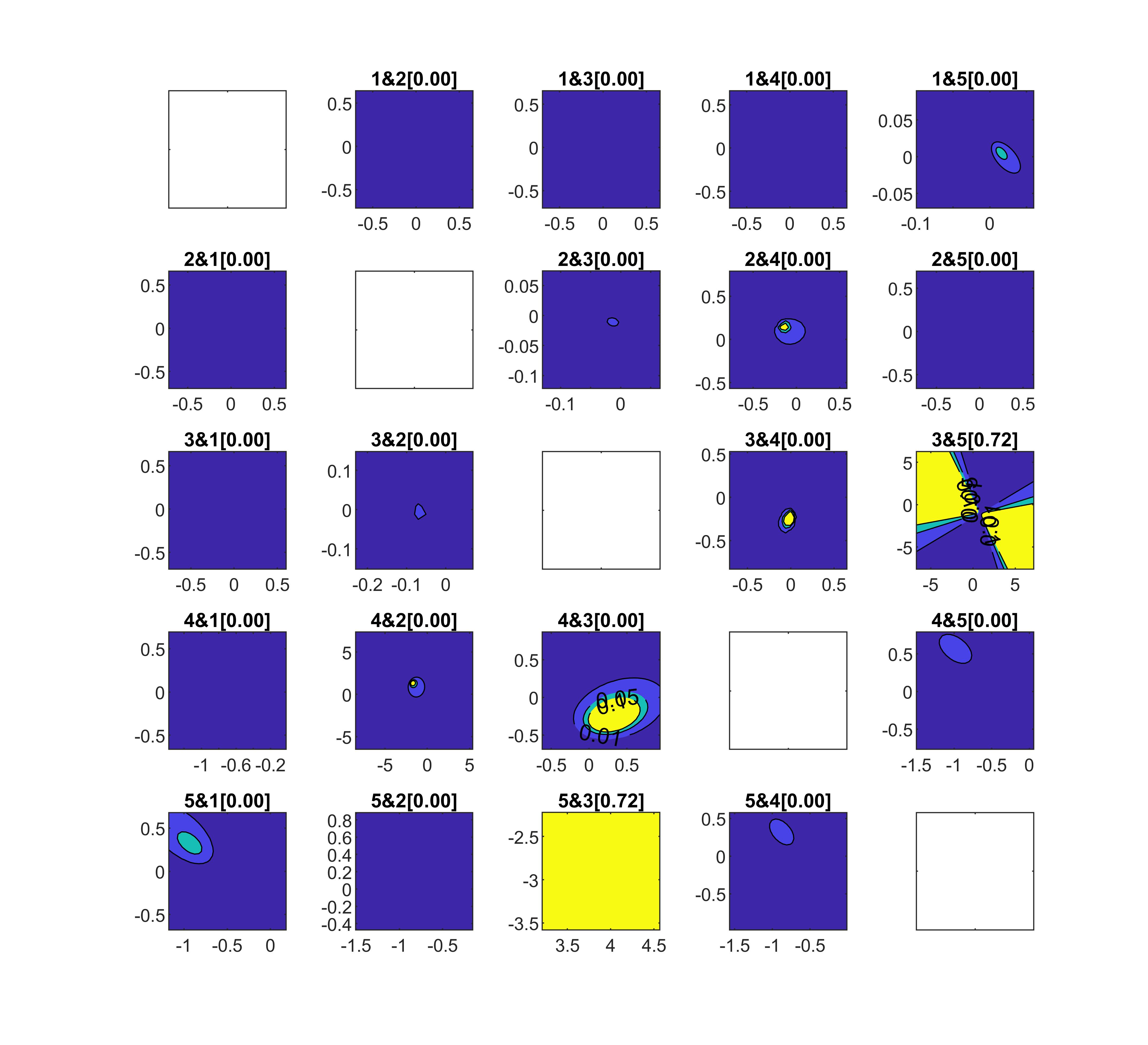}
	\caption*{{\small \textbf{Figure 7}: Joint confidence sets from the sFAR
			test for the two risk premia of the first of the two listed factors in a two
			factor model. (yellow 90\%, light green 95\%, light blue 99\%, dark blue
			area contains the remaining values). Excess returns on bonds with maturities
			3, 60, 120 months are used. [$p$-value] of Kleibergen-Paap rank statistic
			testing $H_0: \text{rank}(\beta)=1$ in square brackets.}}
\end{figure}
 \bigskip
 
\begin{figure}[htbp!]
	\includegraphics[width=\columnwidth]{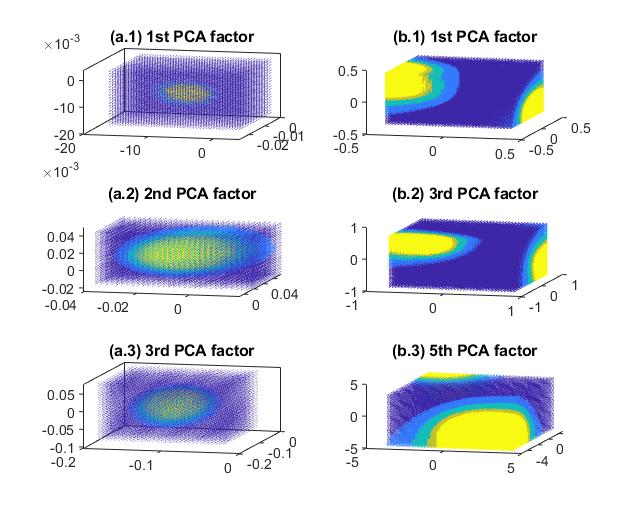}
	\caption*{{\small \textbf{Figure 8}: Joint confidence sets from the sFAR test for
			three-factor models (yellow 90\%, light green 95\%, light blue 99\%, dark
			blue area contains the remaining values) with excess returns on bonds with
			maturities 3, 60, 120 months, (a.i) testing on $\Lambda _{1}$'s w.r.t. the $%
			i $-th factor when using first, second and third PCA factors factors
			Kleibergen-Paap rank statistic testing $H_{0}$ :rank($\beta $)=2 equals
			310.9294 [p-value: 0.0000]); (b.i) testing on $\Lambda _{1}$'s w.r.t. the }$%
		( ${\small $i+2)$-th factor when using first, third and fifth PCA factors
			factors. Kleibergen-Paap rank statistic testing $H_{0}$ :rank($\beta $)=2
			equals 0.0654 [p-value 0.7981]).} }
\end{figure}
\bigskip

Figure 8 shows the three dimensional confidence sets that result from the
sFAR test. It results from partialling out the six risk premia associated
with the other factors. Hence, when we compute these confidence sets using
projection with the identification robust tests, we have to specify a
nine-dimensional grid for the risk premia and compute the identification
robust tests for all values on this nine-dimensional grid. This is, or is
close to be, computationally infeasible. Hence, the sFAR test provides a
computationally tractable manner to conduct identification robust tests on
larger number of parameters.

\section{Conclusion}

We propose identification robust test procedures for testing hypotheses on
risk premia in dynamic affine term structure models. The robust subset
factor Anderson-Rubin test extends the sFAR test from the linear asset
pricing model to allow for tests on multiple risk premia and, unlike
projection based testing, provides a computationally tractable manner to
conduct identification robust tests on larger number of parameters. Our
empirical results show that especially in case of multiple factors, weak
identification is pervasive and traditional tests are likely misleading. We
use the empirical settings from the literature on affine term structure
models, see e.g. \nocite{adrian2013pricing} Adrian et al. (2013) and \nocite%
{ang2003no}Ang and Piazzesi (2003)), to illustrate our results and the
importance of using weak identification robust test procedures.\newpage

\bibliographystyle{ier}
\bibliography{bibtexrefs}

\newpage \appendix
\renewcommand{\theequation}{\thesection.\arabic{equation}}

\section{Appendix: Proofs.}

\noindent

\subsection{Proof of Proposition \protect\ref{prop:KPS covariance}.}

Suppose that there exists a consistent estimator for $S$, $\hat{S}$, a
by-product of the KPS test is the KPS factorization for $\hat{S}$ (\nocite%
{guggenberger2022test}Guggenberger et al.(2022)). We briefly discuss how it
operates. For a matrix $A\in \mathbb{R}^{kp\times kp}$ with a block
structure 
\begin{equation*}
R:=\left( 
\begin{array}{ccc}
R_{11} & \cdots & R_{1p} \\ 
\vdots & \ddots & \vdots \\ 
R_{p1} & \cdots & R_{pp}%
\end{array}%
\right)
\end{equation*}%
where $R_{jl}\in \mathbb{R}^{k\times k},j,l=1,\ldots ,p$, define 
\begin{equation*}
\mathcal{R}(A):=\left( 
\begin{array}{c}
\displaystyle A_{1} \\ 
\vdots \\ 
A_{p}%
\end{array}%
\right) \in \mathbb{R}^{p^{2}\times k^{2}},
\end{equation*}%
with 
\begin{equation*}
A_{j}:=\left( 
\begin{array}{c}
\displaystyle\text{vec}\left( A_{1j}\right) ^{\prime } \\ 
\vdots \\ 
\text{vec}\left( A_{pj}\right) ^{\prime }%
\end{array}%
\right) \in \mathbb{R}^{p\times k^{2}}
\end{equation*}%
for $j=1,\ldots ,p$. Consider a singular value decomposition (SVD) of $%
\mathcal{R}(\hat{S})$: 
\begin{equation*}
\mathcal{R}(\hat{S})=\hat{L}\hat{\Sigma}\hat{N}^{\prime }
\end{equation*}%
where $p=2K,$ $k=N$, $\hat{\Sigma}:=\text{diag}\left( \hat{\sigma}_{1}\ldots 
\hat{\sigma}_{\min \left( p^{2},k^{2}\right) }\right) $ denotes a $%
p^{2}\times k^{2}$ dimensional diagonal matrix with the singular values $%
\hat{\sigma}_{j}$ $\left( j=1,\ldots ,\min \left( p^{2},k^{2}\right) \right) 
$ non-increasingly ordered on the main diagonal, and $\hat{L}\in $ $\mathbb{R%
}^{p^{2}\times p^{2}},$ $\hat{N}\in \mathbb{R}^{k^{2}\times k^{2}}$
orthonormal matrices. Decompose 
\begin{equation*}
\begin{array}{c}
\hat{L}:=\left( 
\begin{array}{ll}
\displaystyle\hat{L}_{11} & \hat{L}_{12} \\ 
\hat{L}_{21} & \hat{L}_{22}%
\end{array}%
\right) =\left( \hat{L}_{1}\text{ }\vdots \text{ }\hat{L}_{2}\right) ,\hat{%
\Sigma}:=\left( 
\begin{array}{cc}
\hat{\sigma}_{1} & 0 \\ 
0 & \hat{\Sigma}_{2}%
\end{array}%
\right) , \\ 
\hat{N}:=\left( 
\begin{array}{ll}
\displaystyle\hat{N}_{11} & \hat{N}_{12} \\ 
\hat{N}_{21} & \hat{N}_{22}%
\end{array}%
\right) =\left( \hat{N}_{1}\text{ }\vdots \text{ }\hat{N}_{2}\right)%
\end{array}%
\end{equation*}%
with $\hat{L}_{11}:1\times 1,$ $\hat{L}_{12}:1\times \left( p^{2}-1\right) ,$
$\hat{L}_{21}:\left( p^{2}-1\right) \times 1,$ $\hat{L}_{22}:\left(
p^{2}-1\right) \times \left( p^{2}-1\right) ,$ $\hat{\sigma}_{1}:1\times 1$, 
$\hat{\Sigma}_{2}:\left( p^{2}-1\right) \times \left( k^{2}-1\right) ,$ $%
\hat{N}_{11}:1\times 1,$ $\hat{N}_{12}:1\times \left( k^{2}-1\right) ,$ $%
\hat{N}_{21}:\left( k^{2}-1\right) \times 1,$ $\hat{N}_{22}:\left(
k^{2}-1\right) \times \left( k^{2}-1\right) $ dimensional matrices. The
consistency follows from the consistency of $\hat{S}$ and Theorem 1 from 
\nocite{guggenberger2022test}Guggenberger et al. (2022). \newline
$\square $

\subsection{Intermediate {results for the Proof of Proposition \protect\ref%
{prop:sFAR} .}}

\label{Intermediate}

We show some intermediate results needed for the proof of Proposition \ref%
{prop:sFAR}. To prove that critical values from a $\chi ^{2}$-distribution
serve as a proper upper bound, we extend the methodology employed in \nocite%
{guggenberger2019more}Guggenberger et al. (2019) by linking the test
statistic to the smaller roots of an eigenvalue problem and thereafter
analyzing the distributions of the eigenvalues.

We first provide Proposition \ref{prop:additional} below which shows that
for the upper bound, we only need to consider the distribution of the $K$
smallest roots of $\left\vert \mu I_{2K-1}-\Xi ^{\prime }\Xi \right\vert =0$
, where $\Xi ^{\prime }\Xi \sim _{a}\mathcal{W}_{2K-1}\left( N,I_{2K-1},%
\mathcal{M}^{\prime }\mathcal{M}\right) $, for the \textquotedblleft
strongly identified\textquotedblright\ setting and $\mathcal{W}_{2K-1}\left(
N,I_{2K-1},\mathcal{M}^{\prime }\mathcal{M}\right) $ indicates a non-central
Wishart distribution of a $(2K-1)\times (2K-1)$ dimensional matrix with $N$
degrees of freedom, scaling matrix $I_{2K-1}$ and non-centrality $\mathcal{M}%
^{\prime }\mathcal{M}$. Proposition \ref{prop:additional}.\ref%
{prop:additional 1} together with its proof show that the sFAR statistic
results from an eigenvalue problem as the summation of its $K$ smallest
eigenvalues. We can thus focus on the joint distribution of these smallest
eigenvalues to obtain the distributional behavior of the sFAR statistic.

Proposition \ref{prop:additional}.\ref{prop:additional 1} shows that the
distribution of the eigenvalues under the null hypothesis depends only on
the nuisance parameters $\mathcal{M}$. Proposition \ref{prop:additional}.\ref%
{prop:additional 3} argues that to derive the upper bound, we need only
focus on the nuisance parameters $\mathcal{M}$ under the \textquotedblleft
strongly identified\textquotedblright\ setting, where the eigenvalues of $%
\mathcal{M}^{\prime }\mathcal{M}$ satisfy the condition that $\kappa
_{1}>\kappa _{2}>\cdots >\kappa _{K-1}>\kappa _{K}=\kappa _{K+1}=\cdots =0$
and the $(K-1)$-th largest eigenvalue, $\kappa _{K-1}$, goes to infinity.
Because the \textquotedblleft strongly identified\textquotedblright\ is
shown to lead to the \textquotedblleft largest
probability\textquotedblright\ of having larger values of the $K$ smallest
eigenvalues.

In the following, we first provide Proposition \ref{prop:additional} and
then continue to discuss the joint distribution of the eigenvalues of the
eigenvalue problem from Proposition \ref{prop:additional}.\ref%
{prop:additional 1} and the associated approximate conditional distribution
of the smallest eigenvalues in Section \ref{secjoint}.

\begin{proposition}
\label{prop:additional}\label{lem:hat kappa small}\label{lem:hat kappa large}
Under Assumptions of Proposition \ref{prop:sFAR},

\begin{enumerate}
\item \label{prop:additional 1} The sFAR$(\lambda _{1}^{0})$ statistic
equals the sum of the K smallest roots of the following polynomial: $%
\left\vert \mu I_{2K-1}-\Xi ^{\prime }\Xi \right\vert =0$, where 
\begin{equation*}
\begin{array}{rl}
\displaystyle\Xi = & \left[ \xi _{u}\text{ }\vdots \text{ }\xi _{\beta _{2}}%
\right] \sim _{a}N\left( \mathcal{M},I_{N(2K-1)}\right) , \\ 
\displaystyle\sqrt{T}\hat{\Sigma}^{-\frac{1}{2}}\hat{u}\hat{\Psi}_{uu}^{-%
\frac{1}{2}}\underset{d}{\rightarrow } & \xi _{u} \\ 
\displaystyle\sqrt{T}\hat{\Sigma}^{-\frac{1}{2}}\left( \hat{\beta}_{2}-\hat{u%
}\hat{\Psi}_{uu}^{-1}\hat{\Psi}_{u\beta _{2}}^{\prime }\right) \hat{\Psi}%
_{\beta _{2}\beta _{2}\cdot u}^{-\frac{1}{2}}\underset{d}{\rightarrow } & 
\xi _{\beta _{2}} \\ 
\displaystyle\mathcal{M}= & \left( \sqrt{T}{\Sigma }^{-\frac{1}{2}}\beta
_{1}(\lambda _{1}-\lambda _{1}^{0})^{\prime }{\Psi }_{uu}^{-\frac{1}{2}}%
\text{ }\vdots \right. \\ 
\displaystyle & \left. \sqrt{T}{\Sigma }^{-\frac{1}{2}}\left( {\beta }%
_{2}-\beta _{1}(\lambda _{1}-\lambda _{1}^{0})^{\prime }{\Psi }_{uu}^{-1}{%
\Psi }_{u\beta _{2}}^{\prime }\right) {\Psi }_{\beta _{2}\beta _{2}\cdot
u}^{-\frac{1}{2}}\right) ,%
\end{array}%
\end{equation*}%
and \textquotedblleft $\sim _{a}$\textquotedblright\ denotes an approximate
distribution in large samples, and ${\Psi }_{uu}$, ${\Psi }_{u\beta _{2}}$, $%
{\Psi }_{\beta _{2}\beta _{2}\cdot u},$ $\hat{\Psi}_{uu}$, $\hat{\Psi}%
_{u\beta _{2}}$, $\hat{\Psi}_{\beta _{2}\beta _{2}\cdot u}$ are specified in
the proof.

\item \label{prop:additional 2} Let $\hat{\kappa}_{1}\geq \hat{\kappa}%
_{2}\geq \cdots \geq \hat{\kappa}_{2K-1}$ be the eigenvalues of $\Xi
^{\prime }\Xi $, and ${\kappa }_{1}\geq {\ \ \kappa }_{2}\geq \cdots \geq {%
\kappa }_{2K-1}$ be the eigenvalues of $\mathcal{M}^{\prime }\mathcal{M}$.
Under the hypothesis $H_{0}:\lambda _{1}=\lambda _{1}^{0}$, $\hat{\kappa}%
_{j}=O_{p}(1),$ $j=K,\cdots ,2K-1$; and when $\kappa _{K-1}\rightarrow
\infty $, we have $\hat{\kappa}_{i}\rightarrow _{p}\infty ,i=1,\cdots ,K-1.$

\item \label{prop:additional 3} $\{\mathcal{M}_{n},n\geq 1\}$ is a sequence
of the parameter matrix $\mathcal{M}$, and let $\mathbb{M}$ denote the
collection of all such sequences. $\mathbb{M}_{\infty }$ is a subset of $%
\mathbb{M}$ that is a collection of $\{\mathcal{M}_{n},n\geq 1:\kappa _{1}(%
\mathcal{M}_{n}^{\prime }\mathcal{M}_{n})>\cdots >\kappa _{K-1}(\mathcal{M}%
_{n}^{\prime }\mathcal{M}_{n})>\kappa _{K}(\mathcal{M}_{n}^{\prime }\mathcal{%
M}_{n})=\kappa _{K+1}(\mathcal{M}_{n}^{\prime }\mathcal{M}_{n})=\cdots
=0,\kappa _{K-1}(\mathcal{M}_{n}^{\prime }\mathcal{M}_{n})\rightarrow \infty
\}$. Let $\tilde{\Xi}_{\mathcal{M}}=Z+\mathcal{M}$ with $Z\sim N(0,I)$ ,then
under the hypothesis $H_{0}:\lambda _{1}=\lambda _{1}^{0}$, for any $\{%
\mathcal{M}_{n},n\geq 1\}\in \mathbb{M}$, we can find a parameter sequence $%
\{\tilde{\mathcal{M}}_{h},h\geq 1\}\in \mathbb{M}_{\infty }$ such that 
\begin{equation*}
\begin{array}{cl}
\displaystyle\limsup\limits_{n\rightarrow \infty }\hat{\kappa}_{j}(\tilde{\Xi%
}_{\mathcal{M}_{n}}^{\prime }\tilde{\Xi}_{\mathcal{M}_{n}})\leq
\liminf\limits_{h\rightarrow \infty }\hat{\kappa}_{j}(\tilde{\Xi}_{\tilde{%
\mathcal{M}}_{h}}^{\prime }\tilde{\Xi}_{\tilde{\mathcal{M}}_{h}}),j\geq K. & 
\end{array}%
\end{equation*}
\end{enumerate}
\end{proposition}

\paragraph{{Proof of Proposition \protect\ref{lem:hat kappa large}.}}

\paragraph{Proof of \protect\ref{lem:hat kappa large}.\protect\ref%
{prop:additional 1}.}

The FAR statistic reads as follows 
\begin{equation}
\begin{array}{cl}
\displaystyle\text{FAR(}{\Lambda }_{1}^{0}\mathbf{)=} & {T\times f}_{T}{%
(\Lambda }_{1}^{0}{,X)}^{\prime }{\hat{V}}_{ff}{(\Lambda }_{1}^{0}{)}^{-1}{f}%
_{T}{(\Lambda }_{1}^{0}{,X)},%
\end{array}
\label{1}
\end{equation}%
and the sample moments $f_{T}(\Lambda _{1},X)$ of the model under
consideration for given $\Lambda _{1}$ are 
\begin{equation*}
\begin{array}{c}
\displaystyle f_{T}(\Lambda _{1},X)=\frac{1}{T}\sum_{t=1}^{T}(\bar{X}%
_{t-1}\otimes (\bar{R}_{t}-\hat{\beta}\hat{V}_{t}))-\left( \hat{Q}%
_{XX}\otimes \hat{\beta}\right) \text{vec(}\Lambda _{1}).%
\end{array}%
\end{equation*}%
Note that for $\hat{d}$ the least squares estimator resulting from the 2-nd
step of the three step procedure, 
\begin{equation*}
\begin{array}{c}
\displaystyle\left( \hat{Q}_{XX}\otimes I_{N}\right) ^{-1}\frac{1}{T}%
\sum_{t=1}^{T}(\bar{X}_{t-1}\otimes (\bar{R}_{t}-\hat{\beta}\hat{V}_{t}))=%
\text{vec}(\hat{d}),%
\end{array}%
\end{equation*}%
which implies that 
\begin{equation}
\begin{array}{c}
\displaystyle f_{T}(\Lambda _{1},X)=\left( \hat{Q}_{XX}\otimes I_{N}\right) 
\text{vec}\left( \hat{\Phi}A(\lambda _{1})\pi _{\Lambda _{2}}\right) ,%
\end{array}
\label{2}
\end{equation}%
with%
\begin{equation*}
\begin{array}{c}
\hat{\Phi}=\frac{1}{T}\sum_{t=1}^{T}R_{t}\left( 
\begin{array}{c}
\bar{X}_{t} \\ 
\hat{v}_{t+1}%
\end{array}%
\right) ^{\prime }\left[ \frac{1}{T}\sum_{t=1}^{T}\left( 
\begin{array}{c}
\bar{X}_{t} \\ 
\hat{v}_{t+1}%
\end{array}%
\right) \left( 
\begin{array}{c}
\bar{X}_{t} \\ 
\hat{v}_{t+1}%
\end{array}%
\right) ^{\prime }\right] ^{-1}=\left( \hat{d}\ \text{{}}\vdots \ \text{{}}%
\hat{\beta}\right) \\ 
\displaystyle A(\lambda _{1})=\left( 
\begin{array}{ccc}
I_{K} & -\lambda _{1} & 0 \\ 
0 & 0 & I_{K-1}%
\end{array}%
\right) ^{\prime },\text{ }\pi _{\Lambda _{2}}=\left( 
\begin{array}{c}
I_{K} \\ 
-\Lambda _{2}%
\end{array}%
\right) .%
\end{array}%
\end{equation*}%
Provided that Assumption \ref{assum:KPS} holds, the above result implies
that we can choose 
\begin{equation}
\begin{array}{c}
\hat{V}_{ff}{(\Lambda }_{1}^{0})=\left( \hat{Q}_{XX}\pi _{\Lambda
_{2}^{0}}^{\prime }A(\lambda _{1}^{0})^{\prime }\otimes I_{N}\right) \hat{V}%
_{\hat{\Phi}}\left( A(\lambda _{1}^{0})\pi _{\Lambda _{2}^{0}}\hat{Q}%
_{XX}\otimes I_{N}\right) \\ 
=\left( \hat{Q}_{XX}\otimes I_{N}\right) \left( \pi _{\Lambda
_{2}^{0}}^{\prime }A(\lambda _{1}^{0})^{\prime }\hat{\Psi}A(\lambda
_{1}^{0})\pi _{\Lambda _{2}^{0}}\otimes \hat{\Sigma}\right) \left( \hat{Q}%
_{XX}\otimes I_{N}\right) .%
\end{array}
\label{3}
\end{equation}%
Substituting (\ref{2}) and (\ref{3}) into (\ref{1}) gives 
\begin{equation*}
\begin{array}{c}
\displaystyle\text{FAR(}\lambda _{1}^{0},\Lambda _{2}^{0}) \\ 
=T\text{vec}\left( \hat{\Phi}A(\lambda _{1})\pi _{\Lambda _{2}}\right)
^{\prime }\left( \pi _{\Lambda _{2}^{0}}^{\prime }A(\lambda
_{1}^{0})^{\prime }\hat{\Psi}A(\lambda _{1}^{0})\pi _{\Lambda
_{2}^{0}}\otimes \hat{\Sigma}\right) ^{-1}\text{vec}\left( \hat{\Phi}%
A(\lambda _{1})\pi _{\Lambda _{2}}\right) \\ 
=T\text{vec}\left( \hat{\Phi}A(\lambda _{1})\pi _{\Lambda _{2}}\right)
^{\prime }\text{vec}\left( \hat{\Sigma}^{-1}\hat{\Phi}A(\lambda _{1})\pi
_{\Lambda _{2}}\left( \pi _{\Lambda _{2}^{0}}^{\prime }A(\lambda
_{1}^{0})^{\prime }\hat{\Psi}A(\lambda _{1}^{0})\pi _{\Lambda
_{2}^{0}}\right) ^{-1}\right) ,%
\end{array}%
\end{equation*}%
which via the trace operator can be rewritten as 
\begin{equation*}
\begin{array}{l}
\displaystyle\text{FAR(}\lambda _{1}^{0},\Lambda _{2}^{0})=T\text{tr}\left(
\pi _{\Lambda _{2}^{0}}^{\prime }A(\lambda _{1}^{0})^{\prime }\hat{\Phi}%
^{\prime }\hat{\Sigma}^{-1}\hat{\Phi}A(\lambda _{1}^{0})\pi _{\Lambda
_{2}^{0}}\left[ \pi _{\Lambda _{2}^{0}}^{\prime }A(\lambda _{1}^{0})^{\prime
}\hat{\Psi}A(\lambda _{1}^{0})\pi _{\Lambda _{2}^{0}}\right] ^{-1}\right) .%
\end{array}%
\end{equation*}%
Denote 
\begin{equation*}
\begin{array}{l}
\displaystyle\widetilde{\text{FAR}}(\lambda _{1}^{0},q)=T\text{tr}\left(
q^{\prime }A(\lambda _{1}^{0})^{\prime }\hat{\Phi}^{\prime }\hat{\Sigma}^{-1}%
\hat{\Phi}A(\lambda _{1}^{0})q\left[ q^{\prime }A(\lambda _{1}^{0})^{\prime }%
\hat{\Psi}A(\lambda _{1}^{0})q\right] ^{-1}\right) ,%
\end{array}%
\end{equation*}%
where $q$ is a $(2K-1)\times K$ matrix of full column rank $K$. Let $q^{\ast
}$ denote the set of $(2K-1)\times K$ matrices of full column rank $K$, and $%
q^{\ast \ast }$ denote the set of all matrices $q$ for which $q^{\prime
}A(\Lambda _{0})^{\prime }\hat{\Psi}A(\Lambda _{0})q=I_{K}$, then it is a
straightforward result such that 
\begin{equation*}
\begin{array}{c}
\displaystyle\min_{q\in q^{\ast }}\widetilde{\text{FAR}}(\lambda
_{1}^{0},q)=\min_{q\in q^{\ast \ast }}\displaystyle\widetilde{\text{FAR}}%
(\lambda _{1}^{0},q)\leq \inf_{\Lambda _{2}^{0}}\text{FAR(}\lambda
_{1}^{0},\Lambda _{2}^{0}).%
\end{array}%
\end{equation*}%
Denote $\hat{q}=\arg \min_{q\in q^{\ast \ast }}\displaystyle\widetilde{\text{%
FAR}}(\lambda _{1}^{0},q)$, and we prove \ref{lem:hat kappa large}.\ref%
{prop:additional 1} by linking $\widetilde{\text{FAR}}(\lambda _{1}^{0},\hat{%
q})$ with an eigenvalue problem and showing that under the null hypothesis $%
\lambda _{1}=\lambda _{1}^{0}$, with probability approaching one we have 
\begin{equation}
\begin{array}{c}
\inf_{\Lambda _{2}^{0}}\text{FAR(}\lambda _{1}^{0},\Lambda _{2}^{0})=%
\widetilde{\text{FAR}}(\lambda _{1}^{0},\hat{q}).%
\end{array}
\label{eq sfar1}
\end{equation}

Theorem 1.2 from \nocite{sameh1982trace}Sameh and Wisniewski (1982) implies
that $\widetilde{\text{FAR}}(\lambda _{1}^{0},\hat{q})$ equals $T$ times the
sum of the $K$ smallest eigenvalues of the eigenvalue problem: 
\begin{equation*}
\begin{array}{c}
A(\Lambda _{0})^{\prime }\hat{\Phi}^{\prime }\hat{\Sigma}^{-1}\hat{\Phi}%
A(\Lambda _{0})x=\kappa \hat{\Psi}(\lambda _{1}^{0})x,%
\end{array}%
\end{equation*}%
where $\hat{\Psi}(\lambda _{1}^{0})=A(\Lambda _{0})^{\prime }\hat{\Psi}%
A(\Lambda _{0})$, $x$ is a $(2K-1)$-eigenvector, $\kappa $ is a scalar
eigenvalue. $T$ times the sum of the $K$ smallest eigenvalues, $\kappa $'s,
of the above eigenvalue problem equals $T$ times the sum of the $K$ smallest
roots of the characteristic polynomial: 
\begin{equation}
\begin{array}{lc}
\displaystyle\left\vert \mu \hat{\Psi}(\lambda _{1}^{0})-A(\Lambda
_{0})^{\prime }\hat{\Phi}^{\prime }\hat{\Sigma}^{-1}\hat{\Phi}A(\Lambda
_{0})\right\vert =0. & 
\end{array}
\label{eigen prob 1}
\end{equation}%
Pre/post-multiplying $\displaystyle\left\vert \mu \hat{\Psi}(\lambda
_{1}^{0})-TA(\Lambda _{0})^{\prime }\hat{\Phi}^{\prime }\hat{\Sigma}^{-1}%
\hat{\Phi}A(\Lambda _{0})\right\vert =0$ by $\displaystyle\left( 
\begin{array}{cc}
I_{K} & -\Lambda _{2}^{\prime } \\ 
0 & I_{K-1}%
\end{array}%
\right) $ gives 
\begin{equation}
\begin{array}{lc}
\displaystyle\left\vert \mu \hat{\Psi}(\lambda _{1}^{0},\Lambda
_{2})-T\left( \hat{u},\hat{\beta}_{2}\right) ^{\prime }\hat{\Sigma}%
^{-1}\left( \hat{u},\hat{\beta}_{2}\right) \right\vert & =0%
\end{array}%
,  \label{eq:A2 1}
\end{equation}%
where 
\begin{equation*}
\begin{array}{cc}
\displaystyle\hat{\Psi}(\lambda _{1}^{0},\Lambda _{2})=\left( 
\begin{array}{cc}
I_{K} & -\Lambda _{2}^{\prime } \\ 
0 & I_{K-1}%
\end{array}%
\right) ^{\prime }\hat{\Psi}(\lambda _{1}^{0})\left( 
\begin{array}{cc}
I_{K} & -\Lambda _{2}^{\prime } \\ 
0 & I_{K-1}%
\end{array}%
\right) =\left( 
\begin{array}{cc}
\hat{\Psi}_{u} & \hat{\Psi}_{u\beta _{2}} \\ 
\hat{\Psi}_{u\beta _{2}}^{\prime } & \hat{\Psi}_{\beta _{2}}%
\end{array}%
\right) , & 
\end{array}%
\end{equation*}%
and $\hat{u}=\hat{d}-\hat{\beta}_{1}\lambda _{1}^{0\prime }-\hat{\beta}%
_{2}\Lambda _{2}=\epsilon +\hat{\beta}_{1}\left( \lambda _{1}-\lambda
_{1}^{0}\right) ^{\prime }.$ Note that $C\hat{\Psi}(\lambda _{1}^{0},\Lambda
_{2})C^{\prime }=I_{2K-1}$ holds with 
\begin{equation*}
\begin{array}{lc}
\displaystyle C=\left( 
\begin{array}{cc}
\hat{\Psi}_{uu}^{-\frac{1}{2}} & 0 \\ 
-\hat{\Psi}_{\beta _{2}\beta _{2}\cdot u}^{-\frac{1}{2}}\hat{\Psi}_{u\beta
_{2}}^{\prime }\hat{\Psi}_{uu}^{-1} & \hat{\Psi}_{\beta _{2}\beta _{2}\cdot
u}^{-\frac{1}{2}}%
\end{array}%
\right) ,\displaystyle\hat{\Psi}_{\beta _{2}\beta _{2}\cdot u}=\hat{\Psi}%
_{\beta _{2}\beta _{2}}-\hat{\Psi}_{u\beta _{2}}^{\prime }\hat{\Psi}%
_{uu}^{-1}\hat{\Psi}_{u\beta _{2}}. & 
\end{array}%
\end{equation*}%
Pre/post-multiplying equation (\ref{eq:A2 1}) by $|C|$ yields $\left\vert
\mu I_{2K-1}-\Xi ^{\prime }\Xi \right\vert =0$. Therefore, the eigenvalue
problem (\ref{eigen prob 1}) is equivalent to the eigenvalue problem $%
\left\vert \mu I_{2K-1}-\Xi ^{\prime }\Xi \right\vert =0$, and $\widetilde{%
\text{FAR}}(\lambda _{1}^{0},\hat{q})$ equals the sum of the K smallest
eigenvalues of $\Xi ^{\prime }\Xi $.\newline

Next, we complete the proof by showing that (\ref{eq sfar1}) holds. Denote $%
\hat{q}=\left( \hat{q}_{1}^{\prime },\hat{q}_{2}^{\prime }\right) ^{\prime }$
with $\hat{q}_{1}$ a $K\times K$ submatrix of $\hat{q}$, and denote $\sigma
_{i}(\hat{q}_{1}),i=1,\cdots K$ the singular values of $\hat{q}_{1}$ in
descending order.

If $\hat{q}_{1}$ is of full rank, then by construction 
\begin{equation*}
\begin{array}{c}
\text{FAR(}\lambda _{1}^{0}, -\hat{q}_2) \\ 
=T \text{tr}\left(\pi_{-\hat{q}_2}^\prime A(\Lambda _{0})^{\prime }\hat{\Phi}%
^{\prime }\hat{\Sigma}^{-1}\hat{\Phi} A(\Lambda _{0})\pi_{-\hat{q}_2} \left[%
\pi_{-\hat{q}_2}^{\prime} A(\Lambda _{0})^{\prime }\hat{\Psi}A(\Lambda
_{0})\pi_{-\hat{q}_2} \right] ^{-1}\right) \\ 
=T \text{tr}\left(A_{\hat{q}_1}^\prime\pi_{-\hat{q}_2}^\prime A(\Lambda
_{0})^{\prime }\hat{\Phi}^{\prime }\hat{\Sigma}^{-1}\hat{\Phi} A(\Lambda
_{0})A_{\hat{q}_1}\pi_{-\hat{q}_2} \left[A_{\hat{q}_1}^\prime \pi_{-\hat{q}%
_2}^{\prime} A(\Lambda _{0})^{\prime }\hat{\Psi}A(\Lambda _{0})\pi_{-\hat{q}%
_2} A_{\hat{q}_1} \right] ^{-1}\right) \\ 
= \widetilde{\text{FAR}}(\lambda _{1}^{0}, \hat{q}),%
\end{array}%
\end{equation*}
where $A_{\hat{q}_1}=\left(%
\begin{matrix}
\hat{q}_1 & 0 \\ 
0 & I_{K-1}%
\end{matrix}
\right)$. Hence, $\inf_{ \Lambda _{2}^{0}} \text{FAR(}\lambda _{1}^{0},
\Lambda _{2}^{0} ) = \widetilde{\text{FAR}}(\lambda _{1}^{0}, \hat{q})$ when 
$\hat{q}_{1}$ is of full rank.

If $\hat{q}_{1}$ is not of full rank, consider a small-$\tilde{\epsilon}$
perturbation of $\hat{q}$, $\hat{q}_{\tilde{\epsilon}}=\hat{q}+\tilde{%
\epsilon}\pi _{\Lambda _{2}^{0}}=\left( \hat{q}_{\tilde{\epsilon},1}^{\prime
},\hat{q}_{\tilde{\epsilon},2}^{\prime }\right) ^{\prime }$, where $\hat{q}_{%
\tilde{\epsilon},1}$ is of full rank. We show that $\tilde{\epsilon}$ can be
chosen in a way such that $\widetilde{\text{FAR}}(\lambda _{1}^{0},\hat{q})$
can be arbitrarily close to $\widetilde{\text{FAR}}(\lambda _{1}^{0},\hat{q}%
_{\tilde{\epsilon}})$, then from $\inf_{\Lambda _{2}^{0}}\text{FAR(}\lambda
_{1}^{0},\Lambda _{2}^{0})\leq \widetilde{\text{FAR}}(\lambda _{1}^{0},\hat{q%
}_{\tilde{\epsilon}})$ we know $\inf_{\Lambda _{2}^{0}}\text{FAR(}\lambda
_{1}^{0},\Lambda _{2}^{0})=\widetilde{\text{FAR}}(\lambda _{1}^{0},\hat{q})$%
. For example, let $\tilde{\epsilon}_{\delta }=-\min \left\{ \delta ,\frac{1%
}{2}|\sigma _{i}(\hat{q}_{1})|,i=1,\cdots ,K\right\} $ with $\delta $ an
arbitrary positive number. Then $\sigma _{i}(\hat{q}_{\tilde{\epsilon}%
_{\delta },1})\neq 0$ due to the extension of Weyl's eigenvalue perturbation
inequality to perturbation of singular values, and thus $\hat{q}_{\tilde{%
\epsilon}_{\delta },1}$ is of full rank, which leads to 
\begin{equation*}
\inf_{\Lambda _{2}^{0}}\text{FAR(}\lambda _{1}^{0},\Lambda _{2}^{0})\leq 
\text{FAR(}\lambda _{1}^{0},-\hat{q}_{\tilde{\epsilon}_{\delta },2})=%
\widetilde{\text{FAR}}(\lambda _{1}^{0},\hat{q}_{\tilde{\epsilon}_{\delta
}}).
\end{equation*}%
By construction, $\widetilde{\text{FAR}}(\lambda _{1}^{0},\hat{q}_{\tilde{%
\epsilon}_{\delta }})=\widetilde{\text{FAR}}(\lambda _{1}^{0},\hat{q}%
)+O_{p}(\delta ^{2})$ and since $\delta $ can be chosen arbitrarily small
(e.g., $\delta =1/\sqrt{T}$), we know with probability approaching to one, $%
\inf_{\Lambda _{2}^{0}}\text{FAR(}\lambda _{1}^{0},\Lambda _{2}^{0})=%
\widetilde{\text{FAR}}(\lambda _{1}^{0},\hat{q})$ holds.

\paragraph{Proof of \protect\ref{lem:hat kappa large}.\protect\ref%
{prop:additional 2}.}

Under the hypothesis $H_0:\lambda_1=\lambda_1^0$, we have $\mathcal{M}%
=\left(0_{N\times K}, \sqrt{T}{\Sigma}^{-\frac{1}{2}} {\beta}_2{\Psi}%
_{\beta_2\beta_2\cdot u}^{-\frac{1}{2}} \right)$, and thus the rank of $%
\mathcal{M}$, denoted as $\rho(\mathcal{M})$, is smaller than or equal to $%
K-1$. Therefore, the null space of $\mathcal{M}$, denoted by $N(\mathcal{M})$%
, is at least $K$ dimensional. In the following, we use the decomposition $%
\Xi = Z + \mathcal{M}$ with $Z\sim N(0,I)$.

For an arbitrary $n$-by-$n$ real symmetric matrix $A$, the $k$-th largest
eigenvalue, via min-max characterization (also known as the Courant-Fisher
expression), can be expressed as 
\begin{equation*}
\begin{array}{c}
\displaystyle\displaystyle\kappa _{k}(A)=\min_{\substack{ U~s.t.  \\ \mathit{%
dim}(U)=n-k+1}}\max_{x\in U}\{x^{\prime }Ax:\lVert x\rVert =1\}%
\end{array}%
\end{equation*}%
where the first minimum is over all $(n-k+1)$-dimensional subspaces $U$ of $%
\mathbb{R}^{n}$. Therefore, employing this characterization, the $j$-th
eigenvalue of $\Xi ^{\prime }\Xi $ is 
\begin{equation*}
\begin{array}{cl}
\hat{\kappa}_{j}=\displaystyle\min_{\substack{ U~s.t.  \\ \text{dim}%
(U)=2K-1-j+1}}\max_{x\in U}\{x^{\prime }(Z+\mathcal{M})^{\prime }(Z+\mathcal{%
M})x:\lVert x\rVert =1\}. & 
\end{array}%
\end{equation*}%
For $j\geq K$, $2K-1-j+1\leq K$, thus we can choose $U\subseteq N(\mathcal{M}%
)$, and note that 
\begin{equation*}
\begin{array}{cl}
& \max_{x\in N(\mathcal{M})}\{x^{\prime }(Z+\mathcal{M})^{\prime }(Z+%
\mathcal{M})x:\lVert x\rVert =1\} \\ 
= & \max_{x\in N(\mathcal{M})}\{x^{\prime }Z^{\prime }Zx:\lVert x\rVert
=1\}=O_{p}(1),%
\end{array}%
\end{equation*}%
which implies that $\hat{\kappa}_{j}=O_{p}(1)$.

For $j<K$, $2K-1-j+1>K$, we can no longer choose $U\subseteq N(\mathcal{M})$
when $\kappa _{K-1}\rightarrow \infty $ as the null space would be $K$
dimensional. Let's consider $\mathcal{M}^{\prime }\mathcal{M}=QPQ^{\prime }$%
, $Q$ is an orthogonal matrix whose columns are eigenvectors of $\mathcal{M}%
^{\prime }\mathcal{M}$ and $P$ is a diagonal matrix whose entries are the
eigenvalues ${\kappa }_{i},i\geq 1$. Denote $Q_{i}$ the eigenvector
associated with ${\kappa }_{i}$.

Then for any $2K-1-j+1$-dimensional linear space $U$ in $\mathbb{R}^{2K-1}$, 
$U\cap \text{span}(Q_i, i\geq K-1)$ is not empty. Hence, 
\begin{equation*}
\begin{array}{cl}
\hat{\kappa}_j= & \displaystyle\min_{\substack{ U ~s.t.  \\ \text{dim}
(U)=2K-1-j+1}} \max_{x\in U} \{x^{\prime }(Z+\mathcal{M})^{\prime }(Z+ 
\mathcal{M}) x: \lVert x\rVert=1 \} \\ 
\geq & O_p(\kappa_{K-1}^{\frac{1}{2}}) + \kappa_{K-1}.%
\end{array}%
\end{equation*}
which implies that $\hat{\kappa}_j\rightarrow_p \infty$.

\paragraph{Proof of \protect\ref{lem:hat kappa large}.\protect\ref%
{prop:additional 3}.}

For an arbitrary sequence $\{\mathcal{M}_n, n\geq 1\}$ in $\mathbb{M}$, by
construction we have (see also Proposition \ref{prop:additional}.(b).(i)), 
\begin{equation}  \label{eq:small kappa Op1}
\begin{array}{cl}
\displaystyle \kappa_j(\tilde{\Xi}_{\mathcal{M}_n}^{\prime }\tilde{\Xi}_{ 
\mathcal{M}_n}) & = O_p(1), j\geq K.%
\end{array}%
\end{equation}
We complete the proof by constructing a sequence $\{\tilde{\mathcal{M}}_h,
h\geq 1\}\in \mathbb{M}_\infty$ such that 
\begin{equation*}
\begin{array}{cl}
\limsup\limits_{n\rightarrow \infty } \hat{\kappa}_j(\tilde{\Xi}_{\mathcal{M}
_n}^{\prime }\tilde{\Xi}_{\mathcal{M}_n}) \leq \liminf\limits_{h\rightarrow
\infty } \hat{\kappa}_j(\tilde{\Xi}_{\tilde{\mathcal{M}}_h}^{\prime }\tilde{
\Xi}_{\tilde{\mathcal{M}}_h}), j\geq K. & 
\end{array}%
\end{equation*}
Pick up a sub-sequence of $\{\mathcal{M}_n, n\geq 1\}$, denoted by $\{%
\mathcal{M}_{n_h}, h\geq 1\}$, such that 
\begin{equation*}
\begin{array}{cl}
\lim\limits_{h\rightarrow \infty } \hat{\kappa}_i(\tilde{\Xi}_{\mathcal{M}
_{n_h}}^{\prime }\tilde{\Xi}_{\mathcal{M}_{n_h}})=
\limsup\limits_{n\rightarrow \infty } \hat{\kappa}_i(\tilde{\Xi}_{\mathcal{M}
_n}^{\prime }\tilde{\Xi}_{\mathcal{M}_n}), i\geq 1; & 
\end{array}%
\end{equation*}
and $Q_{n_h} \rightarrow Q,$ where $\mathcal{M}_{n_h}^{\prime }\mathcal{M}%
_{n_h}=Q_{n_h}P_{n_h} Q_{n_h}^\prime$, $Q_{n_h}$ is an orthogonal matrix
whose columns are eigenvectors of $\mathcal{M}_{n_h}^{\prime }\mathcal{M}%
_{n_h}$ and $P_{n_h}$ is a diagonal matrix whose entries are the eigenvalues 
${\kappa}_i({\mathcal{\ \ M}_{n_h}}^{\prime }{\mathcal{M}_{n_h}}), i\geq 1 $%
. Denote $Q_{n_h,i}$ the eigenvector associated with ${\kappa}_i({\mathcal{M}%
_{n_h}}^{\prime }{\ \mathcal{M}_{n_h}})$, and by construction $Q_{n_h,i}
\rightarrow Q_i$. Since we allow weak identifications, and thus ${\kappa}%
_{i}({\mathcal{M}_{n_h}}^{\prime }{\mathcal{M}_{n_h}}), i\leq K-1$ can be
zero when all factors are unspanned or bounded when all factors are weak $%
\beta_2=O(1/\sqrt{T})$.

Now we construct the sequence $\{\tilde{\mathcal{M}}_h, h\geq 1\}$ in $%
\mathbb{M}_\infty$. Let $\tilde{\mathcal{M}}_h= \tilde{P}_h^{\frac{1}{2}%
}Q_{n_h}^\prime$ with $\tilde{P}_h$ a diagonal matrix with $\tilde{P}%
_{h,ii}=\kappa_{h,i}$ such that $\kappa_{h,1}> \cdots > \kappa_{h,K-1} >
\kappa_{h,K}=\kappa_{h,K+1}=\cdots=0$, $\kappa_{h,K-1} \rightarrow \infty$
and 
\begin{equation*}
\kappa_{h,1}= o\left({\left(\sum_{1\leq i\leq 2K-1} \lVert
Q_{n_h,i}-Q_i\rVert_{\max} +h^{-1}\right)^{-1}}\right).
\end{equation*}
By construction, the following two properties hold.

\begin{enumerate}
\item For $j\geq K$, 
\begin{equation}
\begin{array}{cl}
\displaystyle & \displaystyle\liminf_{h}\hat{\kappa}_{j}(\tilde{\Xi}_{\tilde{
\mathcal{M}}_{h}}^{\prime }\tilde{\Xi}_{\tilde{\mathcal{M}}_{h}}) \\ 
= & \displaystyle\liminf_{h}\min_{\substack{ U~s.t.  \\ \text{dim}
(U)=2K-1-j+1,  \\ U\perp \text{span}(Q_{i},i<K)}}\max_{x\in U}\{x^{\prime } 
\tilde{\Xi}_{\tilde{\mathcal{M}}_{h}}^{\prime }\tilde{\Xi}_{\tilde{\mathcal{%
\ M }}_{h}}x:\lVert x\rVert =1\},%
\end{array}
\label{eq:minmax kappa j 1}
\end{equation}
where the minimum is over all $(2K-1-j+1)$-dimensional subspaces $U$ of $%
\mathbb{R}^{2K-1}$ that are orthogonal to the linear spaces spanned by
vectors $Q_{i},i<K$.

\item For $j\geq K$, 
\begin{equation}
\begin{array}{cl}
\displaystyle\liminf\limits_{h\rightarrow \infty }\hat{\kappa}_{j}(\tilde{
\Xi }_{\tilde{\mathcal{M}}_{h}}^{\prime }\tilde{\Xi}_{\tilde{\mathcal{M}}
_{h}})=\lim\limits_{h\rightarrow \infty }\hat{\kappa}_{j}(\tilde{\Xi}_{ 
\tilde{\mathcal{M}}_{h}}^{\prime }\tilde{\Xi}_{\tilde{\mathcal{M}}_{h}}). & 
\end{array}
\label{eq:kappa j 0}
\end{equation}
\end{enumerate}

We prove the above two properties which lead directly to the final
conclusion.

\begin{enumerate}
\item {Proof of equation (\ref{eq:minmax kappa j 1})}. Note that 
\begin{equation*}
\begin{array}{cl}
\lVert \tilde{\mathcal{M}}_{h}x-\tilde{P}_{h}^{\frac{1}{2}}Q^{\prime
}x\rVert \leq \kappa _{h,1}^{\frac{1}{2}}\left( \sum_{1\leq i\leq
2K-1}\lVert Q_{n_h,i}-Q_{i}\rVert _{\max }\right) \lVert x\rVert =o(1) & 
\end{array}%
\end{equation*}
due to $\kappa _{h,1}=o\left( \frac{1}{\sum_{1\leq i\leq 2K-1}\lVert
Q_{n_h,i}-Q_{i}\rVert _{\max }+h^{-1}}\right) $. Hence, for the sequence $\{{%
\tilde{\mathcal{M}}_{h}},h\geq 1\}$, if, for example, for an arbitrary $i<K $
such that $Q_{i}^{\prime }x\neq 0,\lVert x\rVert =1$, then with probability
approaching one as $h$ increases, 
\begin{equation*}
\begin{array}{cl}
x^{\prime }\tilde{\Xi}_{\tilde{\mathcal{M}}_{h}}^{\prime }\tilde{\Xi}_{ 
\tilde{\mathcal{M}}_{h}}x & =x^{\prime }Z^{\prime }Zx+x^{\prime }{\tilde{ 
\mathcal{M}}_{h}}^{\prime }Zx+x^{\prime }Z^{\prime }{\tilde{\mathcal{M}}_{h}}
x+x^{\prime }{\tilde{\mathcal{M}}_{h}}^{\prime }{\tilde{\mathcal{M}}_{h}}x
\\ 
& \geq \frac{1}{2}\kappa _{h,K-1}x^{\prime }{Q_{i}}^{\prime }{Q_{i}}
x\rightarrow \infty ,%
\end{array}%
\end{equation*}
Combining the above equation with equation (\ref{eq:small kappa Op1}), we
know $\kappa _{j}(\Xi _{\mathcal{M}_{h}}^{\prime }\Xi _{\mathcal{M}
_{h}}),j\geq K$, the minimum in equation (\ref{eq:minmax kappa j 1}) should
be achieved over linear spaces that are orthogonal to $Q_{i},i<K$ as $h$
increases. This completes the proof of the first statement.

\item {Proof of equation (\ref{eq:kappa j 0})}. For $j\geq K$, by
construction, $\displaystyle \hat{\kappa}_j(\tilde{\Xi}_{\tilde{\mathcal{M}}%
_h}^{\prime } \tilde{\Xi}_{\tilde{\mathcal{M}}_h})=Q_{n_h,i}^{\prime }\tilde{%
\Xi}_{\tilde{ \mathcal{M}}_h}^{\prime }\tilde{\Xi}_{\tilde{\mathcal{M}}_h}
Q_{n_h,i}=Q_{n_h,i}^{\prime }Z^{\prime }Z Q_{n_h,i}\rightarrow Q_{i}^{\prime
}Z^{\prime }Z Q_{i}. $ Hence, $\liminf\limits_{h\rightarrow \infty } $ $\hat{%
\kappa}_j(\tilde{ \Xi }_{\tilde{\mathcal{M}}_h}^{\prime }\tilde{\Xi}_{\tilde{%
\mathcal{M}}_h})$ $= \lim\limits_{h\rightarrow \infty } \hat{\kappa}_j(%
\tilde{\Xi}_{\tilde{ \mathcal{M}}_h}^{\prime }\tilde{\Xi}_{\tilde{\mathcal{M}%
}_h}). $
\end{enumerate}

Therefore, we arrive at the final conclusion that for $j\geq K$, 
\begin{equation*}
\begin{array}{cl}
& \limsup\limits_{n\rightarrow \infty } \hat{\kappa}_j(\tilde{\Xi}_{\mathcal{%
\ \ M}_n}^{\prime }\tilde{\Xi}_{\mathcal{M}_n})=\lim\limits_{h\rightarrow
\infty } \hat{\kappa}_j(\tilde{\Xi}_{\mathcal{M}_{n_h}}^{\prime }\tilde{\Xi}
_{ \mathcal{M}_{n_h}}) \\ 
= & \displaystyle \lim\limits_{h\rightarrow \infty }\min_{\substack{ U ~s.t. 
\\ \text{dim}(U)=2K-1-j+1}} \max_{x\in U} \{x^{\prime }\tilde{\Xi}_{\mathcal{%
\ \ M}_{n_h}}^{\prime }\tilde{\Xi}_{\mathcal{M}_{n_h}} x: \lVert x\rVert=1 \}
\\ 
\leq & \displaystyle \lim\limits_{h\rightarrow \infty }\min_{\substack{ U
~s.t.  \\ \text{dim}(U)=2K-1-j+1,  \\ U\perp \text{span}(Q_{n_h,i}, i< K)}}
\max_{x\in U} \{x^{\prime }\tilde{\Xi}_{\mathcal{M}_{n_h}}^{\prime }\tilde{
\Xi}_{\mathcal{M}_{n_h}} x: \lVert x\rVert=1 \} \\ 
= & \displaystyle \lim\limits_{h\rightarrow \infty }\min_{\substack{ U ~s.t. 
\\ \text{dim}(U)=2K-1-j+1,  \\ U\perp \text{span}(Q_{n_h,i}, i< K)}}
\max_{x\in U} \{x^{\prime }Z^{\prime }Z x: \lVert x\rVert=1 \} \\ 
= & \displaystyle \lim\limits_{h\rightarrow \infty }\min_{\substack{ U ~s.t. 
\\ \text{dim}(U)=2K-1-j+1,  \\ U\perp \text{span}(Q_{i}, i< K)}} \max_{x\in
U} \{x^{\prime }Z^{\prime }Z x+o_p(1): \lVert x\rVert=1 \} \\ 
= & \displaystyle \lim\limits_{h\rightarrow \infty }\min_{\substack{ U ~s.t. 
\\ \text{dim}(U)=2K-1-j+1,  \\ U\perp \text{span}(Q_{i}, i< K)}} \max_{x\in
U} \{x^{\prime }\tilde{\Xi}_{\tilde{\mathcal{M}}_h}^{\prime }\tilde{\Xi}_{ 
\tilde{\mathcal{M}}_h} x : \lVert x\rVert=1 \} \\ 
= & \displaystyle \lim\limits_{h\rightarrow \infty } \hat{\kappa}_j(\tilde{
\Xi}_{\tilde{\mathcal{M}}_h}^{\prime }\tilde{\Xi}_{\tilde{\mathcal{M}}_h})=
\liminf\limits_{h\rightarrow \infty } \hat{\kappa}_j(\tilde{\Xi}_{\tilde{ 
\mathcal{M}}_h}^{\prime }\tilde{\Xi}_{\tilde{\mathcal{M}}_h}),%
\end{array}%
\end{equation*}
where the second last equality is due to construction that we restrict the
increasing rate of $\kappa_{h,1}$ such that $\kappa_{h,1}=o\left(\frac{1}{%
\sum_{1\leq i\leq 2K-1} \lVert Q_{n_h,i}-Q_i\rVert_{\max} +h^{-1}}\right)$.

~\newline
$\square$

\subsubsection{Joint distribution of the eigenvalues of the eigenvalue
problem in Proposition \protect\ref{prop:additional}.\protect\ref%
{prop:additional 1} and the approximate conditional distribution}

\label{secjoint}

Proposition \ref{prop:additional} links the sFAR statistic with the smallest
eigenvalues of an eigenvalue problem. We next first discuss the
distributional behavior of the eigenvalues of the eigenvalue problem in
Proposition \ref{prop:additional}.\ref{prop:additional 1} under the
\textquotedblleft strong identified\textquotedblright\ setting, which then
leads to the distributional behavior of the smallest eigenvalues. We derive
an approximate conditional density of the smallest eigenvalues given the
larger ones and then show the critical values constructed based on the
conditional density are bounded by those generated from a $\chi ^{2}$
distribution asymptotically in Corollary \ref{coro2}.

\nocite{guggenberger2019more}Guggenberger et al. (2019) discuss the case
where $\mathcal{M}^{\prime }\mathcal{M}$ is a reduced rank 2-by-2 matrix
with eigenvalues $\kappa _{1}>\kappa _{2}=2$ and $\kappa _{1}$ increasing to
infinity. This Section extends the analysis to the more general
\textquotedblleft strongly identified\textquotedblright\ specification that $%
\mathcal{M}^{\prime }\mathcal{M}$ is a reduced rank $L$-by-$L$ matrix of
rank $K-1$ with eigenvalues $\kappa _{1}>\cdots >\kappa _{K-1}>\kappa
_{K}=\cdots =\kappa _{L}$ and $\kappa _{K-1}$ increasing to infinity.

The following Proposition \ref{prop:a2} replicates the discussion in \nocite%
{muirhead1978latent}Muirhead (1978) for an approximation of the density of
all eigenvalues of noncentral Wishart matrices. This gives rise to an
approximate conditional density of the $(L-K+1)$ smallest eigenvalues of the
noncentral Wishart matrix given the remaining eigenvalues as shown in (\ref%
{eq conditional 1}).

\begin{proposition}
\label{prop:a2} The joint density of the eigenvalues $\hat{\kappa}_1\geq 
\hat{\kappa}_2\geq \cdots \geq \hat{\kappa}_L$ of $\tilde{\xi}_{\mathcal{M}
_h}^{\prime }\tilde{\xi}_{\mathcal{M}_h} \sim W_L(N, I_L, \mathcal{M}
_h^{\prime }\mathcal{M}_h)$ with $\{\mathcal{M}_h,h\geq 1\}\in \mathbb{M}
_\infty$ such that $\mathcal{M}_h^{\prime }\mathcal{M}_h$ is of rank $K-1$
and the $K-1$ ($K\geq 2$) largest eigenvalues of $\mathcal{M}_h^{\prime } 
\mathcal{M}_h$ are distinct and go to infinity, when $h$ is large, can be
approximated by 
\begin{equation*}  \label{equation joint latent 0}
\begin{array}{ll}
\displaystyle f_{\hat{\kappa}_1,\hat{\kappa}_2,\cdots, \hat{\kappa}
_{L}}(x_1, x_2,\cdots, x_{L}) = &  \\ 
\displaystyle g_1(\kappa_i, 1\leq i\leq L)g_{2, \hat{\kappa}_K,\cdots, \hat{
\kappa}_L}(x_i, K\leq i\leq L) \times
\prod_{i=1}^{K-1}\prod_{j=K-1+1}^{L}(x_i-x _j)^{\frac{1}{2}} &  \\ 
\displaystyle \times \exp{\left(-\frac{1}{2}\sum_{j=K-1+1}^{L}x _j\right) }
\prod_{i=K-1+1}^{L}x_i^{\frac{1}{2}(N-{L} -1)}\prod_{i<j;K-1+1}^{L}(x_i-x_j),
& 
\end{array}%
\end{equation*}
where 
\begin{equation*}
\begin{array}{ll}
\displaystyle g_1(\kappa_i, 1\leq i\leq L) = &  \\ 
\displaystyle \frac{\pi^{\frac{1}{2}L^2-\frac{1}{2}(K-1)K } \Gamma_{K-1}( 
\frac{1 }{2}N)\Gamma_{K-1}(\frac{1}{2}L) }{ \Gamma_{L}(\frac{1}{2}
N)\Gamma_{L}(\frac{1}{2}L2^{\frac{1}{2}LN-\frac{1}{2} (K-1)(N+L-(K-1) -3) }}
\exp\left({-\frac{1}{2}\sum_{i=1}^{L}\kappa_i }\right) &  \\ 
\times \displaystyle \prod_{i=1}^{K-1}\kappa_i^{-\frac{1}{4}(N+L-2(K-1))}, & 
\\ 
~ &  \\ 
\displaystyle g_{2, \hat{\kappa}_K,\cdots, \hat{\kappa}_L}(x_i, K\leq i\leq
L) = &  \\ 
\displaystyle \exp\left({-\frac{1}{2}\sum_{i=1}^{K-1}x_i +
\sum_{i=1}^{K-1}(\kappa_ix_i)^{\frac{1}{2}} } \right)\prod_{i=1}^{K-1} x_i^{ 
\frac{1}{4}}\prod_{i=1}^{K-1}x_i^{\frac{1}{4} (N-L-2)}
\prod_{i<j}^{K-1}\left(\frac{x_i-x_j}{{\ \kappa}_i-{\kappa}_j} \right). & 
\end{array}%
\end{equation*}
\end{proposition}

\paragraph{{Proof of Proposition \protect\ref{prop:a2}.}}

The joint distribution of the eigenvalues $\hat{\kappa}_{1}\geq \hat{\kappa}%
_{2}\geq \cdots \geq \hat{\kappa}_{L}$ of a $\mathcal{W}_{L}\left( N,I_{L},%
\mathcal{M}^{\prime }\mathcal{M}\right) $-distributed random matrix reads
(see, e.g., \nocite{james1964distributions}James (1964)), 
\begin{equation}
\begin{array}{lc}
\displaystyle\displaystyle f_{\hat{\kappa}_{1},\hat{\kappa}_{2},\cdots ,\hat{%
\kappa}_{L}}(x_{1},x_{2},\cdots ,x_{L})= & \displaystyle \\ 
\displaystyle\frac{\pi ^{\frac{1}{2}L^{2}}}{2^{\frac{1}{2}NL}\Gamma _{L}(%
\frac{1}{2}N)\Gamma _{L}(\frac{1}{2}L)}\exp \left( -\frac{1}{2}%
\sum_{i=1}^{L}x_{i}\right) \prod_{i=1}^{L}x_{i}^{\frac{1}{2}(N-2K)} &  \\ 
\displaystyle\times \prod_{i<j}(x_{i}-x_{j})\exp \left( -\frac{1}{2}%
\sum_{i=1}^{L}\kappa _{i}\right) ~_{0}F_{1}\left( \frac{1}{2}N;\frac{1}{4}%
\Omega _{\kappa },\Omega _{x}\right) , & 
\end{array}
\label{a21}
\end{equation}%
where $x_{1}\geq x_{2}\geq \cdots \geq x_{L}\geq 0$, $\kappa _{1}>\kappa
_{2}>\cdots >\kappa _{K-1}=\cdots =\kappa _{L}=0$ are eigenvalues of $%
\mathcal{M}^{\prime }\mathcal{M}$, $\Omega _{\kappa }=\text{diag}\left(
\kappa _{1},\kappa _{2},\cdots ,\kappa _{L}\right) ,$ $\Omega _{\kappa }=%
\text{diag}\left( \kappa _{1},\kappa _{2},\cdots ,\kappa _{L}\right) $ and $%
_{0}F_{1}$ is a hypergeometric function of the matrix argument. $_{0}F_{1}$
contains power series representations in terms of zonal polynomials, which
are very hard to find exact closed forms for except for limited cases. To
analyze the density, we consider asymptotic approximations under certain
model sequences (\nocite{muirhead1978latent}Muirhead (1978), \nocite%
{guggenberger2019more}Guggenberger et al. (2019)).

When $\kappa_1, \cdots, \kappa_{K-1}$ are large, \nocite{leach1969bessel}%
Leach (1969) shows that $~_0F_1\left(\frac{1}{2}N; \frac{1}{4}\Omega_\kappa,
\Omega_x \right)$ can be approximated by the following function, 
\begin{equation}
\begin{array}{ll}
\displaystyle & \displaystyle~_0F_1\left(\frac{1}{2}N; \frac{1}{4}
\Omega_\kappa, \Omega_x \right) \approx \displaystyle2^{\frac{1}{2}
(K-1)(N+K-3) \pi^{-\frac{1}{2}(K-1)K} \Gamma_{K-1}(\frac{1}{2} N)
\Gamma_{K-1}(\frac{1}{2}L) } \\ 
& \displaystyle\times \exp\left(\sum_{i=1}^{K-1}(x_i \kappa_i)^{\frac{1}{2}}
\right)\prod_{i=1}^{K-1} (x_i\kappa_i)^{\frac{1}{2}(L-N)}
\prod_{i=1}^{K-1}\prod_{j=1;i<j}^{L}c_{ij}^{-\frac{1}{2}}%
\end{array}
\label{eq:A2 }
\end{equation}%
where 
\begin{equation*}
\begin{array}{lll}
\displaystyle c_{ij} & = (\kappa_i-\kappa_j)(x_i-x_j)~~~~ & i,j=1,\cdots,K-1;
\\ 
\displaystyle & = \kappa_i (x_i-x_j)~~~~ & i=1,\cdots,K-1, j=K,\cdots,L;%
\end{array}%
\end{equation*}%
and 
\begin{equation*}
\begin{array}{lcc}
\displaystyle \Gamma_{m}(a)= \pi^{\frac{1}{2}m(m-1)} \prod_{j=1}^m \Gamma(a- 
\frac{1}{2} (j-1)). &  & 
\end{array}%
\end{equation*}%
Substituting (\ref{eq:A2 }) into (\ref{a21}) gives an asymptotic
representation of the joint density of all latent roots, 
\begin{equation*}
\begin{array}{lc}
\displaystyle f_{\hat{\kappa}_1,\hat{\kappa}_2,\cdots, \hat{\kappa}
_{L}}(x_1, x_2,\cdots, x_{L}) \sim \frac{\pi^{\frac{1}{2}L^2-\frac{1}{2}
(K-1)K } \Gamma_{K-1}(\frac{1 }{2}N)\Gamma_{K-1}(\frac{1}{2}L) }{
\Gamma_{L}( \frac{1}{2} N)\Gamma_{L}(\frac{1}{2}L2^{\frac{1}{2}LN-\frac{1}{2}
(K-1)(N+L-(K-1) -3) }} &  \\ 
\displaystyle \times \exp{\left(-\frac{1}{2}\sum_{i=1}^{L}\kappa_i \right)}
\prod_{i=1}^{K-1}\kappa_i^{-\frac{1}{4}(N+L-2(K-1))} &  \\ 
\displaystyle \times \exp\left({-\frac{1}{2}\sum_{i=1}^{K-1}x_i +
\sum_{i=1}^{K-1}(\kappa_ix_i)^{\frac{1}{2}} }\right) \prod_{i=1}^{K-1} x_i^{ 
\frac{1}{4}}\prod_{i=1}^{K-1}x_i^{\frac{1}{4} (N-L-2)}
\prod_{i<j}^{K-1}\left(\frac{x_i-x_j}{{\ \kappa}_i-{\kappa}_j} \right) &  \\ 
\displaystyle \times \prod_{i=1}^{K-1}\prod_{j=K-1+1}^{L}(x_i-x _j)^{\frac{1 
}{2}} \exp{\left(-\frac{1}{2}\sum_{j=K-1+1}^{L}x _j\right) } &  \\ 
\displaystyle \prod_{i=K-1+1}^{L}x_i^{\frac{1}{2}(N-{L} -1)}
\prod_{i<j;K-1+1}^{L}(x_i-x_j), & 
\end{array}%
\end{equation*}%
which corresponds to equation (6.5) in \nocite{muirhead1978latent}Muirhead
(1978). ~\newline
$\square$ ~\newline

{Proposition} \ref{prop:a2} leads to an approximate conditional density
function of the $(L-K+1)$ smallest roots $\hat{\kappa}_{j},$ $j=K,\cdots ,L$
given the $(K-1)$ largest roots $\hat{\kappa}_{i},i=1,\cdots ,K-1$ when $%
\kappa _{1}>\kappa _{2}>\cdots >\kappa _{K-1}>\kappa _{K}=\kappa
_{K+1}=\cdots =0$ and $\kappa _{1},\cdots ,\kappa _{K-1}$ are large is (see,
e.g., Corollary 2 in \nocite{james1969tests}James (1969), \nocite%
{muirhead1978latent}Muirhead (1978))) 
\begin{equation}
\begin{array}{lc}
\displaystyle f_{\hat{\kappa}_{K},\cdots ,\hat{\kappa}_{L}|\hat{\kappa}%
_{1},\cdots ,\hat{\kappa}_{K-1}}^{\ast }(x_{K-1+1},\cdots x_{L}|\hat{\kappa}%
_{1},\cdots ,\hat{\kappa}_{K-1}) &  \\ 
=\displaystyle g\left( \hat{\kappa}_{i},1\leq i\leq K-1\right)
\prod_{i=1}^{K-1}\prod_{j=K-1+1}^{L}(\hat{\kappa}_{i}-x_{j})^{\frac{1}{2}} & 
\\ 
\displaystyle\times \exp {\left( -\frac{1}{2}\sum_{j=K-1+1}^{L}x_{j}\right) }%
\prod_{i=K-1+1}^{L}x_{i}^{\frac{1}{2}(N-{L}-1)}%
\prod_{i<j;K-1+1}^{L}(x_{i}-x_{j}) &  \\ 
=\displaystyle\bar{g}\left( \kappa _{1},\cdots \kappa _{K}\right)
\prod_{i=1}^{K-1}\prod_{j=K-1+1}^{L}(1-\frac{x_{j}}{\hat{\kappa}_{i}})^{%
\frac{1}{2}}\times p_{N,L,K}\left( x_{K-1+1},\cdots x_{L}\right) & 
\end{array}
\label{eq conditional 1}
\end{equation}%
where $\hat{\kappa}_{K-1}\geq x_{K-1+1}\geq \cdots \geq x_{L}\geq 0$, $%
p_{N,L,K}\left( \cdot \right) $ is the joint density function of eigenvalues
of a $W_{L-K+1}(N-K+1,I_{L-K+1})$-distributed matrix such that 
\begin{equation*}
\begin{array}{l}
\displaystyle p_{N,L,K}\left( x_{K-1+1},\cdots x_{L}\right) = \\ 
\displaystyle\bar{G}\displaystyle\exp \left( -\frac{1}{2}%
\sum_{i=K-1+1}^{L}x_{i}\right) \displaystyle\times
\prod_{i=K-1+1}^{L}x_{i}^{(N-L-1)/2}\prod_{i<j;i\geq K-1+1}^{L}\left(
x_{i}-x_{j}\right) , \\ 
~ \\ 
\displaystyle\bar{G}=\frac{\pi ^{\frac{1}{2}(L-(K-1))^{2}}2^{-\frac{1}{2}%
(N-(K-1))(L-(K-1))}}{\Gamma _{L-(K-1)}(\frac{1}{2}(L-(K-1)))\Gamma
_{L-(K-1)}(\frac{1}{2}(N-(K-1)))},%
\end{array}%
\end{equation*}%
$\displaystyle g\left( \hat{\kappa}_{i},1\leq i\leq K-1\right) ,\bar{g}%
\left( \hat{\kappa}_{1},\cdots \hat{\kappa}_{K}\right) $ are functions only
of the largest $(K-1)$ eigenvalues and independent of values of $\kappa
_{i},1\leq i\leq K-1$, and 
\begin{equation*}
\begin{array}{c}
\displaystyle\bar{g}\left( \hat{\kappa}_{1},\cdots \hat{\kappa}_{K}\right)
=\prod_{i=1}^{K-1}\hat{\kappa}_{i}^{\frac{1}{2}(L-(K-1))}g\left( \hat{\kappa}%
_{i},1\leq i\leq K-1\right) /\bar{G}.%
\end{array}%
\end{equation*}

\nocite{guggenberger2019more}Guggenberger et al. (2019) provide the exact
analytical form of $f_{\hat{\kappa}_{2}|\hat{\kappa}_{1}}^{\ast }$, and
study its properties via numerical integration. We employ an alternative
approach by finding an asymptotic representation of $\displaystyle f_{\hat{%
\kappa}_{K},\cdots ,\hat{\kappa}_{L}|\hat{\kappa}_{1},\cdots ,\hat{\kappa}%
_{K-1}}^{\ast }$ when $\hat{\kappa}_{K-1}$ is large, since $\hat{\kappa}%
_{K-1}\rightarrow _{p}\infty $ by Proposition \ref{prop:additional}.\ref%
{prop:additional 2}. For example, \nocite{guggenberger2019more}Guggenberger
et al. (2019) show the exact form of the special case $L=K=2$ such that 
\begin{equation*}
\displaystyle f_{\hat{\kappa}_{2}|\hat{\kappa}_{1}}^{\ast }(x_{2}|\hat{\kappa%
}_{1})=\bar{g}(\hat{\kappa}_{1})\left( 1-\frac{x_{2}}{\hat{\kappa}_{1}}%
\right) ^{\frac{1}{2}}p_{N,2,2}\left( x_{2}\right) ,
\end{equation*}%
where 
\begin{equation*}
\displaystyle\bar{g}(\hat{\kappa}_{1})=\frac{\Gamma \left( \frac{K+2}{2}%
\right) 2^{\frac{N+1}{2}}}{\hat{\kappa}_{1}^{\frac{K}{2}}\sqrt{\pi }%
~_{1}F_{1}\left( \frac{N-1}{2},\frac{N+2}{2};-\frac{\hat{\kappa}_{1}}{2}%
\right) }\hat{\kappa}_{1}^{\frac{1}{2}},
\end{equation*}%
and $p_{N,2,2}\left( \cdot \right) $ is the joint density function of the
eigenvalue of a $W_{1}(N-1,1)$-distributed matrix and thus also the density
function of the $\chi ^{2}(N-1)$-distribution. The property 
that for $-z\rightarrow \infty $, 
\begin{equation*}
\displaystyle{\ }_{1}F_{1}(a;b;z)=\frac{\Gamma (b)(-z)^{-a}}{\Gamma (b-a)}%
\left( \sum_{k=0}^{n}\frac{(-1)^{k}(a)_{k}(a-b+1)_{k}z^{-k}}{k!}+O\left(
z^{-n-1}\right) \right)
\end{equation*}%
implies that 
\begin{equation*}
\begin{array}{l}
\displaystyle\displaystyle\bar{g}(\hat{\kappa}_{1})=\displaystyle\frac{1}{1-%
\frac{N-1}{2}\frac{1}{\hat{\kappa}_{1}}+O\left( \hat{\kappa}_{1}^{-2}\right) 
}=\left( 1+\frac{N-1}{2}\frac{1}{\hat{\kappa}_{1}}+O\left( \hat{\kappa}%
_{1}^{-2}\right) \right) ,%
\end{array}%
\end{equation*}%
and hence when $\hat{\kappa}_{1}$ is large we can rewrite $f_{\hat{\kappa}%
_{2}|\hat{\kappa}_{1}}^{\ast }$ in the following form 
\begin{equation}
\begin{array}{c}
\displaystyle f_{\hat{\kappa}_{2}|\hat{\kappa}_{1}}^{\ast }(x_{2}|\hat{\kappa%
}_{1})\displaystyle=\left( 1+\frac{N-1}{2}\frac{1}{\hat{\kappa}_{1}}+O\left( 
\hat{\kappa}_{1}^{-2}\right) \right) \left( 1-\frac{x_{2}}{\hat{\kappa}_{1}}%
\right) ^{\frac{1}{2}}p_{1}\left( x_{2}\right) \\ 
\displaystyle=\left( 1+O\left( \hat{\kappa}_{1}^{-1}\right) \right) \left( 1-%
\frac{x_{2}}{\hat{\kappa}_{1}}\right) ^{\frac{1}{2}}p_{1}\left( x_{2}\right)
,0\leq x_{2}\leq \hat{\kappa _{1}}.%
\end{array}
\label{a1}
\end{equation}

The structure of $f_{\hat{\kappa}_{2}|\hat{\kappa}_{1}}^{\ast }(x_{2}|\hat{%
\kappa}_{1})$ shows that $\bar{g}(\hat{\kappa}_{1})$ is close to one when $%
\hat{\kappa}_{1}$ is large, and for fixed $x_{2}$ the limit of $f_{\hat{%
\kappa}_{2}|\hat{\kappa}_{1}}^{\ast }(x_{2}|\hat{\kappa}_{1})$ is $%
p_{N,2,2}\left( x_{2}\right) $. When $\hat{\kappa}_{1}$ increases, the
maximum possible value can be taken by $x_{2}$ increases and the term $%
\left( 1-{x_{2}}/{\hat{\kappa}_{1}}\right) ^{{1}/{2}}$ increases as well for
fixed $x_{2}$. These observations imply that if we ignore terms of order $%
O\left( \hat{\kappa}_{1}^{-1}\right) $, increasing $\hat{\kappa}_{1}$ leads
to a conditional density with "fatter tails", and thus the critical values
generated based on the conditional density increase as $\hat{\kappa}_{1}$
increases and are bounded by those generated based on $p_{N,2,2}\left(
x_{2}\right) $. These observations should hold for general cases as well.

In the following, we extend the above representation (\ref{a1}) to the
general $f_{\hat{\kappa}_K, \cdots, \hat{\kappa}_L|\hat{\kappa}_1, \cdots, 
\hat{\kappa}_{K-1}}^*(\cdot)$, and then discuss in Corollary \ref{coro2} the
monotonicity and boundedness of the critical values. To derive the
representation of $f_{\hat{\kappa}_K, \cdots, \hat{\kappa}_L|\hat{\kappa}_1,
\cdots, \hat{\kappa}_{K-1}}^*(\cdot)$, Proposition \ref{largest eigenvalues}
analyzes the form of the first factor of $\displaystyle f_{\hat{\kappa}_K,
\cdots, \hat{\kappa}_L|\hat{\kappa}_1, \cdots, \hat{\kappa}_{K-1}}^*$ in (%
\ref{eq conditional 1}), i.e., $\bar{g}\left(\hat{\kappa}_{1},\cdots \hat{%
\kappa}_{K}\right)$, when $\hat{ \kappa}_{K-1}$ is large.

\begin{proposition}
\label{largest eigenvalues} When $\hat{ \kappa}_{K-1}$ is large, the
function $\bar{g}\left(\hat{\kappa }_{1},\cdots \hat{\kappa}_{K}\right)$
from (\ref{eq conditional 1}) satisfies, 
\begin{equation*}
\begin{array}{c}
\displaystyle \bar{g}\left(\kappa_{1},\cdots \kappa_{K}\right) = %
\displaystyle \prod_{i=1}^{K-1} \prod_{j=K-1+1}^{L} \left(1+ \frac{N-(K-1)}{%
2 \hat{\kappa}_i} + O\left(\hat{\kappa}_{i}^{-2}\right)\right).%
\end{array}%
\end{equation*}
\end{proposition}

\paragraph{{Proof of Proposition \protect\ref{largest eigenvalues}.}}

~\newline
Equation (\ref{eq conditional 1}) can be expressed as 
\begin{equation*}
\begin{array}{l}
\displaystyle f_{\hat{\kappa}_K, \cdots, \hat{\kappa}_L|\hat{\kappa}_1,
\cdots, \hat{\kappa}_{K-1}}^*(x_{K-1+1},\cdots x_{L}|\hat{\kappa}_1, \cdots, 
\hat{\kappa}_{K-1})= \displaystyle \bar{g}\left(\hat{\kappa}_1, \cdots, \hat{
\kappa}_{K-1}\right) \\ 
\displaystyle \times \left\{\left(1-\left( \sum_{i=1}^{K-1}\frac{1}{2\hat{
\kappa}_i}\right)\left(\sum_{j=K-1+1}^{2K-1} x_j\right) \right) \right. \\ 
\displaystyle \left.+ \prod_{i=1}^{K-1}\prod_{j=K-1+1}^{L}(1-\frac{x_j}{\hat{
\kappa}_i})^{\frac{1}{2}} - \left(1-\left( \sum_{i=1}^{K-1}\frac{1}{2\hat{
\kappa}_i}\right)\left(\sum_{j=K-1+1}^{2K-1} x_j\right) \right) \right\} \\ 
~ \\ 
\times p_{N,L,K}\left(x_{K-1+1},\cdots x_{L}\right) \\ 
~ \\ 
= \bar{g}\left(\hat{\kappa}_1, \cdots, \hat{\kappa}_{K-1}\right) \left(A_1 +
A_2 \right)p_{N,L,K}\left(x_{K-1+1},\cdots x_{L}\right),%
\end{array}%
\end{equation*}
with 
\begin{equation*}
\begin{array}{l}
\displaystyle A_1 = \left(1-\left( \sum_{i=1}^{K-1}\frac{1}{2\hat{\kappa}_i}
\right)\left(\sum_{j=K-1+1}^{2K-1} x_j\right) \right), \\ 
\displaystyle A_2= \prod_{i=1}^{K-1}\prod_{j=K-1+1}^{L}(1-\frac{x_j}{\hat{
\kappa}_i})^{\frac{1}{2}} - \left(1-\left( \sum_{i=1}^{K-1}\frac{1}{2\hat{
\kappa}_i}\right)\left(\sum_{j=K-1+1}^{2K-1} x_j\right) \right).%
\end{array}%
\end{equation*}
Theorem 3.2.20. in \nocite{muirhead2009aspects}Muirhead (2009) shows that 
\begin{equation*}
\begin{array}{l}
\displaystyle \int\cdots \int_{\substack{ x_{K-1+1}  \\ \geq \cdots \geq
x_{L} \geq 0 }} \left(\sum_{j=K-1+1}^{L} x_j \right)
p_{N,L,K}\left(x_{K-1+1},\cdots x_{L}\right)dx_{K}\cdots d x_{L} \\ 
\displaystyle = (N-(K-1))(L-K+1),%
\end{array}%
\end{equation*}
and thus integrating $\bar{g}\left(\hat{\kappa}_1, \cdots, \hat{\kappa}%
_{K-1}\right) A_1 p_{N,L,K}\left(x_{K-1+1},\cdots x_{L}\right)$ leads to 
\begin{equation*}
\begin{array}{l}
\displaystyle \int\cdots \int_{\substack{ \hat{\kappa}_{K-1}\geq x_{K-1+1} 
\\ \geq \cdots \geq x_{L} \geq 0 }} \bar{g}\left(\hat{\kappa}_1, \cdots, 
\hat{\kappa}_{K-1}\right) A_1 p_{N,L,K}\left(x_{K-1+1},\cdots x_{L}\right)
dx_{K}\cdots d x_{L} \\ 
\displaystyle = 1- \sum_{i=1}^{K-1} \frac{(N-(K-1))(L-K+1)}{\hat{\kappa}_i}
\\ 
\displaystyle - \int\cdots \int_{\substack{ x_{K-1+1}  \\ \geq \cdots \geq
x_{L} \geq 0;  \\ x_{K-1+1}\geq \hat{\kappa}_{K-1} }} \bar{g}\left(\hat{
\kappa}_1, \cdots, \hat{\kappa}_{K-1}\right) A_1
p_{N,L,K}\left(x_{K-1+1},\cdots x_{L}\right) dx_{K}\cdots d x_{L} \\ 
\displaystyle = 1- \sum_{i=1}^{K-1} \frac{(N-(K-1))(L-K+1)}{\hat{\kappa}_i}
+ O\left(\exp \left(-\frac{1}{8} \hat{\kappa}_{K-1} \right)\right),%
\end{array}%
\end{equation*}
where the last equality is due to the fact that $\displaystyle %
p_{N,L,K}\left(x_{K-1+1},\cdots x_{L}\right)$ is smaller than $\exp \left(-%
\frac{1}{4} \sum_{i=K-1+1}^{L} x_{i}\right) $ when $x_K$ is large.

~\newline
Next, we consider the integration concerning the term $A_2$. Note that $%
-2\leq A_2 \leq 2 $, and for $x\leq \sqrt{\hat{\kappa}}$, we would have that
for some real number $\xi_{x/\hat{\kappa}}$ between 0 and $x/\hat{\kappa}$, 
\begin{equation*}
\begin{array}{c}
\displaystyle \left(1-\frac{x}{\hat{\kappa}}\right)^{\frac{1}{2}}=1-\frac{1}{
2}\frac{x}{\hat{\kappa}} - \frac{1}{8\left(1- \xi_{x/\hat{\kappa}
}\right)^{3/2}}\left(\frac{x}{\hat{\kappa}} \right)^2,%
\end{array}%
\end{equation*}
which leads to the fact that for large $\hat{\kappa}_{K-1}$ and $x_{K}\leq 
\sqrt{\hat{\kappa}_{K-1}}$, 
\begin{equation*}
\begin{array}{c}
\displaystyle A_2=- \sum_{i=1}^{K-1} \left(\frac{1}{\hat{\kappa}_i}
\right)^2 \left( \sum_{j=K-1+1}^{2K-1} \frac{x_j^2}{8\left(1- \xi_{x_j/\hat{
\kappa}_i }\right)^{3/2}} \right) + \widetilde{A}_2 ,%
\end{array}%
\end{equation*}
where $\widetilde{A}_2$ is a summation of terms of the form 
\begin{equation*}
\begin{array}{c}
\displaystyle \prod_{i\in I_i,j\in I_j} -\frac{1}{8\left(1- \xi_{x_j/\hat{
\kappa}_i }\right)^{3/2}} \frac{x_j^2}{\hat{\kappa}_i^2} \prod_{l\in I_l,m
\in I_m} -\frac{x_l}{2\hat{\kappa}_m}.%
\end{array}%
\end{equation*}
Therefore, when $x_{K}\leq \sqrt{\hat{\kappa}_{K-1}}$, $A_2=
O\left(\sum_{i=1}^{K-1}\frac{1}{\hat{\kappa}_i^2} \right)\left(%
\sum_{j=K-1+1}^{2K-1}{x_j^2} \right)$, and thus integrating $\bar{g}\left(%
\hat{\kappa}_1, \cdots, \hat{\kappa}_{K-1}\right) A_2
p_{N,L,K}\left(x_{K-1+1},\cdots x_{L}\right)$ leads to 
\begin{equation*}
\begin{array}{l}
\displaystyle \int\cdots \int_{\substack{ \hat{\kappa}_{K-1}\geq x_{K-1+1} 
\\ \geq \cdots \geq x_{L} \geq 0 }} \bar{g}\left(\hat{\kappa}_1, \cdots, 
\hat{\kappa}_{K-1}\right) A_2 p_{N,L,K}\left(x_{K-1+1},\cdots x_{L}\right)
dx_{K}\cdots d x_{L} \\ 
\displaystyle = \int\cdots \int_{\substack{ \sqrt{\hat{\kappa}_{K-1}}\geq
x_{K-1+1}  \\ \geq \cdots \geq x_{L} \geq 0 }} \bar{g}\left(\hat{\kappa}_1,
\cdots, \hat{\kappa}_{K-1}\right) O\left(\sum_{i=1}^{K-1}\frac{1}{\hat{%
\kappa }_i^2} \right)\left(\sum_{j=K-1+1}^{2K-1} {x_j^2} \right) \\ 
\times \displaystyle p_{N,L,K}\left(x_{K-1+1},\cdots x_{L}\right)
dx_{K}\cdots d x_{L} \\ 
\displaystyle + \int\cdots \int_{\substack{ \hat{\kappa}_{K-1} \geq
x_{K-1+1}  \\ \geq \cdots \geq x_{L} \geq 0;  \\ x_{K} \geq \sqrt{\hat{%
\kappa }_{K-1}} }} \bar{g}\left(\hat{\kappa}_1, \cdots, \hat{\kappa}%
_{K-1}\right) O(1) p_{N,L,K}\left(x_{K-1+1},\cdots x_{L}\right) dx_{K}\cdots
d x_{L} \\ 
\displaystyle = O\left(\sum_{i=1}^{K-1}\frac{1}{\hat{\kappa}_i^2} \right) +
O\left(\exp \left(-\frac{1}{8} \sqrt{\hat{\kappa}_{K-1}} \right)\right) =
O\left(\sum_{i=1}^{K-1}\frac{1}{\hat{\kappa}_i^2} \right).%
\end{array}%
\end{equation*}
Combining the above results, we have 
\begin{equation*}
\begin{array}{l}
\displaystyle \int\cdots \int_{\substack{ \hat{\kappa}_{K-1}\geq x_{K-1+1} 
\\ \geq \cdots \geq x_{L} \geq 0 }} \bar{g}\left(\hat{\kappa}_1, \cdots, 
\hat{\kappa}_{K-1}\right) (A_1+A_2) p_{N,L,K}\left(x_{K-1+1},\cdots
x_{L}\right) dx_{K}\cdots d x_{L} \\ 
\displaystyle = 1- \sum_{i=1}^{K-1} \frac{(N-(K-1))(L-K+1)}{\hat{\kappa}_i}
+ O\left(\sum_{i=1}^{K-1}\frac{1}{\hat{\kappa}_i^2} \right).%
\end{array}%
\end{equation*}
Thus, 
\begin{equation*}
\begin{array}{l}
\displaystyle \bar{g}\left(\hat{\kappa}_1, \cdots, \hat{\kappa}_{K-1}\right)
\\ 
\displaystyle =\left( 1- \sum_{i=1}^{K-1} \frac{(N-(K-1))(L-K+1)}{\hat{%
\kappa }_i} + O\left(\sum_{i=1}^{K-1}\frac{1}{\hat{\kappa}_i^2}
\right)\right)^{-1} \\ 
\displaystyle =1+ \sum_{i=1}^{K-1} \frac{(N-(K-1))(L-K+1)}{\hat{\kappa}_i} +
O\left(\sum_{i=1}^{K-1}\frac{1}{\hat{\kappa}_i^2} \right) \\ 
= \displaystyle \prod_{i=1}^{K-1} \prod_{j=K-1+1}^{L} \left(1+ \frac{N-(K-1) 
}{2\hat{\kappa}_i} + O\left(\hat{\kappa}_{i}^{-2}\right)\right).%
\end{array}%
\end{equation*}

~\newline
$\square$

Proposition \ref{largest eigenvalues} directly leads to Corollary \ref{coro}%
, which shows the form of $\displaystyle f_{\hat{\kappa}_K, \cdots, \hat{%
\kappa}_L|\hat{\kappa}_1, \cdots, \hat{\kappa}_{K-1}}^*$ when $\hat{ \kappa}%
_{K-1}$ is large.

\begin{corollary}
\label{coro} Under the specification in Proposition \ref{prop:a2}, 
for arbitrary $L\geq K\geq 2$, the approximate conditional distribution of
the smallest $(L-K+1)$ eigenvalues $\hat{\kappa}_{K},\cdots ,\hat{\kappa}%
_{L} $ given the largest $(K-1)$ eigenvalues $\hat{\kappa}_{1},\cdots ,\hat{%
\kappa}_{K-1}$, satisfies 
\begin{equation}
\begin{array}{cc}
\displaystyle f_{\hat{\kappa}_{K},\cdots ,\hat{\kappa}_{L}|\hat{\kappa}%
_{1},\cdots ,\hat{\kappa}_{K-1}}^{\ast }(x_{K-1+1},\cdots x_{L}|\hat{\kappa}%
_{1},\cdots ,\hat{\kappa}_{K-1}) &  \\ 
=\displaystyle\prod_{i=1}^{K-1}\prod_{j=K-1+1}^{L}\left( 1+\frac{N-(K-1)}{2%
\hat{\kappa}_{i}}+O\left( \hat{\kappa}_{i}^{-2}\right) \right) (1-\frac{x_{j}%
}{\hat{\kappa}_{i}})^{\frac{1}{2}} &  \\ 
\times p_{N,L,K}\left( x_{K-1+1},\cdots x_{L}\right) , & 
\end{array}
\label{eq conditional 3}
\end{equation}%
where $\hat{\kappa}_{K-1}\geq x_{K-1+1}\geq \cdots x_{L}\geq 0$.
Furthermore, for $x_{K}\leq \hat{\kappa}_{K-1}-\hat{\kappa}_{K-1}^{-\frac{%
2(1-\epsilon )}{3}}$ with $\epsilon $ a small arbitrary positive real
number, satisfies 
\begin{equation}
\begin{array}{c}
\displaystyle f_{\hat{\kappa}_{K},\cdots ,\hat{\kappa}_{L}|\hat{\kappa}%
_{1},\cdots ,\hat{\kappa}_{K-1}}^{\ast }(x_{K-1+1},\cdots x_{L}|\hat{\kappa}%
_{1},\cdots ,\hat{\kappa}_{K-1})= \\ 
~ \\ 
\displaystyle\prod_{i=1}^{K-1}\prod_{j=K-1+1}^{L}\left( 1+\frac{(N-(K-1))-{%
x_{j}}}{2\hat{\kappa}_{i}}-\tilde{f}(x_{j},\hat{\kappa}_{i})+O\left( \frac{1%
}{\hat{\kappa}_{i}^{1+\epsilon }}\right) \right) \\ 
\times p_{N,L,K}\left( x_{K-1+1},\cdots x_{L}\right) ,%
\end{array}
\label{a2}
\end{equation}%
where $\tilde{f}(x,\hat{\kappa})\geq 0$ such that for some real number
between 0 and $\displaystyle\frac{x}{\hat{\kappa}}$ denoted by $\xi _{x/\hat{%
\kappa}}$, 
\begin{equation*}
\begin{array}{c}
\tilde{f}(x,\hat{\kappa})=\frac{1}{8\left( 1-\xi _{x/\hat{\kappa}}\right)
^{3/2}}\left( \frac{x}{\hat{\kappa}}\right) ^{2}\left( 1+\frac{N-(K-1)}{2}%
\frac{1}{\hat{\kappa}}\right) +\frac{N-(K-1)}{4}\frac{x}{\hat{\kappa}^{2}}.%
\end{array}%
\end{equation*}%
For $x_{K}\geq \hat{\kappa}_{K-1}-\hat{\kappa}_{K-1}^{-\frac{2(1-\epsilon )}{%
3}}$, 
\begin{equation}
\begin{array}{c}
\displaystyle\displaystyle f_{\hat{\kappa}_{K},\cdots ,\hat{\kappa}_{L}|\hat{%
\kappa}_{1},\cdots ,\hat{\kappa}_{K-1}}^{\ast }(x_{K-1+1},\cdots x_{L}|\hat{%
\kappa}_{1},\cdots ,\hat{\kappa}_{K-1})= \\ 
~ \\ 
O\left( \ {\hat{\kappa}_{K-1}^{-\frac{1-\epsilon }{3}}}\right)
p_{N,L,K}\left( x_{K-1+1},\cdots x_{L}\right) .%
\end{array}
\label{a2-2}
\end{equation}
\end{corollary}

\paragraph{{Proof of Corollary \protect\ref{coro}.}}

~\newline
~\newline
Substituting the result of {Proposition} \ref{largest eigenvalues} into (\ref%
{eq conditional 1}) gives equation (\ref{eq conditional 3}). Next, we prove
the remaining statements. The following result holds for some real number $%
\xi_{x/\hat{\kappa}} $ between 0 and $\frac{x}{\hat{\kappa}}$ when $x\leq 
\hat{\kappa}-\hat{ \kappa}^{-\frac{2(1-\epsilon)}{3}}$, 
\begin{equation}  \label{taylor}
\begin{array}{c}
\displaystyle \left(1-\frac{x}{\hat{\kappa}}\right)^{\frac{1}{2}}=1-\frac{1}{
2}\frac{x}{\hat{\kappa}} - \frac{1}{8\left(1- \xi_{x/\hat{\kappa}
}\right)^{3/2}}\left(\frac{x}{\hat{\kappa}} \right)^2.%
\end{array}%
\end{equation}
The expansion (\ref{taylor}) implies that for $x_{K}\leq \hat{\kappa}_{K-1}-%
\hat{\kappa}_{K-1}^{-\frac{2(1-\epsilon)}{3}}$ 
\begin{equation*}
\begin{array}{l}
\displaystyle \prod_{i=1}^{K-1}\prod_{j=K-1+1}^{L}(\hat{\kappa}_i-x_j)^{ 
\frac{1}{2}} \\ 
\displaystyle = \prod_{i=1}^{K-1}\hat{\kappa}_i^{\frac{1}{2}(L-(K-1)) }
\left(\prod_{i=1}^{K-1} \prod_{j=K-1+1}^{L} \left(1-\frac{x_j}{2\hat{%
\kappa_i }}- \frac{1}{8\left(1- \xi_{\hat{\kappa}_i/x_j}\right)^{3/2}}\left(%
\frac{x_j }{\hat{\kappa}_i} \right)^2 \right)\right)%
\end{array}%
\end{equation*}
and thus equation (\ref{eq conditional 1}) can be expressed as 
\begin{equation}  \label{eq conditional 2}
\begin{array}{l}
\displaystyle f_{\hat{\kappa}_K, \cdots, \hat{\kappa}_L|\hat{\kappa}_1,
\cdots, \hat{\kappa}_{K-1}}^*(x_{K-1+1},\cdots x_{L}|\hat{\kappa}_1, \cdots, 
\hat{\kappa}_{K-1})= \\ 
\displaystyle \bar{g}\left(\kappa_{1},\cdots
\kappa_{K}\right)\left(\prod_{i=1}^{K-1} \prod_{j=K-1+1}^{L} \left(1-\frac{
x_j}{2\hat{\kappa_i}}- \frac{1}{8\left(1- \xi_{\hat{\kappa}_i;
x_j}\right)^{3/2}}\left(\frac{x_j}{\hat{\kappa}_i} \right)^2 \right)\right)
\\ 
\times p_{N,L,K}\left(x_{K-1+1},\cdots x_{L}\right).%
\end{array}%
\end{equation}

Following the result in {Proposition} \ref{largest eigenvalues}, the
equation (\ref{eq conditional 2}) implies that

\begin{equation*}
\begin{array}{l}
\displaystyle f_{\hat{\kappa}_K, \cdots, \hat{\kappa}_L|\hat{\kappa}_1,
\cdots, \hat{\kappa}_{K-1}}^*(x_{K-1+1},\cdots x_{L}|\hat{\kappa}_1, \cdots, 
\hat{\kappa}_{K-1})= \\ 
\displaystyle \prod_{i=1}^{K-1} \prod_{j=K-1+1}^{L} \left(1+ \frac{N-(K-1)}{%
2 \hat{\kappa}_i} + O\left(\hat{\kappa}_{i}^{-2}\right)\right) \\ 
\times \displaystyle \prod_{i=1}^{K-1} \prod_{j=K-1+1}^{L} \left(1-\frac{x_j 
}{2\hat{\kappa_i}}- \frac{1}{8\left(1- \xi_{\hat{\kappa}_i/x_j}\right)^{3/2} 
}\left(\frac{x_j}{\hat{\kappa}_i} \right)^2 \right) \\ 
\displaystyle \times p_{N,L,K}\left(x_{K-1+1},\cdots x_{L}\right) \\ 
~ \\ 
\displaystyle =\prod_{i=1}^{K-1} \prod_{j=K-1+1}^{L} \left(1 + \frac{
(N-(K-1))-{x_j} }{2\hat{\kappa}_i} - \tilde{f}(x_j,\hat{\kappa}_i) +O\left( 
\frac{1}{\hat{\kappa}_i^{1+\epsilon}} \right) \right) \\ 
\times p_{N,L,K}\left(x_{K-1+1},\cdots x_{L}\right)%
\end{array}%
\end{equation*}
As for the case when $x_{K}\geq \hat{\kappa}_{K-1}-\hat{\kappa}_{K-1}^{- 
\frac{2(1-\epsilon)}{3}}$, we would have 
\begin{equation*}
\begin{array}{l}
\displaystyle \prod_{i=1}^{K-1}\prod_{j=K-1+1}^{L}\left(1-\frac{x_j}{\hat{
\kappa}_i}\right)^{ \frac{1}{2}} = O\left(\ {\hat{\kappa}_{K-1}^{-\frac{
1-\epsilon}{3}} } \right),%
\end{array}%
\end{equation*}
which leads to equation (\ref{a2-2}). ~\newline
$\square$ ~\newline

Corollary \ref{coro} derives the form of $\displaystyle f_{\hat{\kappa}%
_{K},\cdots ,\hat{\kappa}_{L}|\hat{\kappa}_{1},\cdots ,\hat{\kappa}%
_{K-1}}^{\ast }$ when $\hat{\kappa}_{K-1}$ is large, and establishes a link
between the approximate conditional density and the density function of
eigenvalues of a central Wishart matrix. When $L=K=2$, the results are the
same as we have previously discussed. Similarly, it is evident from equation
(\ref{eq conditional 3}) and $\hat{\kappa}_{K-1}\rightarrow _{p}\infty $
from Proposition \ref{prop:additional}.\ref{prop:additional 2} that the
limit of $f_{\hat{\kappa}_{K},\cdots ,\hat{\kappa}_{L}|\hat{\kappa}%
_{1},\cdots ,\hat{\kappa}_{K-1}}^{\ast }(\cdot )$ is $p_{N,L,K}(\cdot )$.
This limiting result generalizes the limiting property of $f_{\hat{\kappa}%
_{2}|\hat{\kappa}_{1}}^{\ast }$ from \nocite{guggenberger2019more}%
Guggenberger et al. (2019) to $\displaystyle f_{\hat{\kappa}_{K},\cdots ,%
\hat{\kappa}_{L}|\hat{\kappa}_{1},\cdots ,\hat{\kappa}_{K-1}}^{\ast }$. 
\nocite{guggenberger2019more}Guggenberger et al. (2019) state that the
quantiles of the distribution with density $f_{\hat{\kappa}_{2}|\hat{\kappa}%
_{1}}^{\ast }$ are strictly increasing in $\hat{\kappa}_{1}$ based on
numerical integration. Corollary \ref{coro2} provides an analytical proof
for the boundedness by the limiting distribution and the monotonicity of
critical values if we ignore terms of order $O(\hat{\kappa}_{K-1}^{-1})$.

\begin{corollary}
\label{coro2} Denote $t_{1}$ the sum of $\hat{\kappa}_{K}\geq \cdots \geq 
\hat{\kappa}_{L}$ generated from a distribution with density function $%
\displaystyle f_{\hat{\kappa}_{K},\cdots ,\hat{\kappa}_{L}|\hat{\kappa}%
_{1},\cdots ,\hat{\kappa}_{K-1}}^{\ast }(\cdot |\hat{\kappa}_{1},\cdots ,%
\hat{\kappa}_{K-1})$ specified in Corollary \ref{coro}, and $t_{2}$ follows $%
\chi ^{2}\left( (L-(K-1))(N-(K-1))\right) $-distribution.  
Denote the $1-\alpha $ quantile of $t_{i}$, $q_{1-\alpha ,t_{i}},i=1,2$.
There exists $\tilde{q}_{1-\alpha ,t_{1}}$ such that $\tilde{q}_{1-\alpha
,t_{1}}=q_{1-\alpha -\tilde{\eta}_{\alpha },t_{1}}$, $0\leq \tilde{\eta}%
_{\alpha }\leq 1-\alpha $, $\tilde{\eta}_{\alpha }=O(\hat{\kappa}%
_{K-1}^{-1}) $, $\tilde{q}_{1-\alpha ,t_{1}}$ is increasing in $\hat{\kappa}%
_{i},i\leq K-1 $, and 
\begin{equation*}
\begin{array}{c}
\tilde{q}_{1-\alpha ,t_{1}}<q_{1-\alpha ,t_{2}}.%
\end{array}%
\end{equation*}%
Additionally, when $\hat{\kappa}_{K-1}$ is large, for small $\alpha $'s, $%
q_{1-\alpha ,t_{1}}<q_{1-\alpha ,t_{2}}$, and $q_{1-\alpha ,t_{1}}$ is
increasing in $\hat{\kappa}_{i},i\leq K-1$.
\end{corollary}

\paragraph{{Proof of Corollary \protect\ref{coro2}.}}

These are direct results of Corollary \ref{coro} and Theorem 3.2.20 in 
\nocite{muirhead2009aspects}Muirhead (2009). Theorem 3.2.20 in \nocite%
{muirhead2009aspects}Muirhead (2009) shows that the sum of eigenvalues $\ell
_{1}\geq \cdots \geq \ell _{L-K+1}$ of a $W_{L-(K-1)}(N-(K-1),I)$%
-distributed matrix follows a $\chi ^{2}\left( (L-(K-1))(N-(K-1))\right) $%
-distribution. \newline
~\newline
Equation (\ref{eq conditional 3}) implies that when $\hat{\kappa}_{K-1}$ is
large, 
\begin{equation*}
\begin{array}{cc}
\displaystyle\displaystyle f_{\hat{\kappa}_{K},\cdots ,\hat{\kappa}_{L}|\hat{%
\kappa}_{1},\cdots ,\hat{\kappa}_{K-1}}^{\ast }(x_{K-1+1},\cdots x_{L}|\hat{%
\kappa}_{1},\cdots ,\hat{\kappa}_{K-1}) &  \\ 
=(1+\eta _{\hat{\kappa}_{K-1}})\tilde{f}^{\ast }(x_{K-1+1},\cdots x_{L}|\hat{%
\kappa}_{1},\cdots ,\hat{\kappa}_{K-1}), & 
\end{array}%
\end{equation*}%
with $\eta _{\hat{\kappa}_{K-1}}>0,\eta _{\hat{\kappa}_{K-1}}=O\left( \hat{%
\kappa}_{K-1}^{-1}\right) $ and 
\begin{equation*}
\begin{array}{cc}
\tilde{f}^{\ast }(x_{K-1+1},\cdots x_{L}|\hat{\kappa}_{1},\cdots ,\hat{\kappa%
}_{K-1})=\displaystyle\prod_{i=1}^{K-1}\prod_{j=K-1+1}^{L}(1-\frac{x_{j}}{%
\hat{\kappa}_{i}})^{\frac{1}{2}}p_{N,L,K}\left( x_{K-1+1},\cdots
x_{L}\right) . & 
\end{array}%
\end{equation*}%
Denote 
\begin{equation*}
\displaystyle\tilde{q}_{1-\alpha ,t_{1}}=\inf \left\{ \tilde{q}\geq
0:\int\limits_{\sum_{j=K-1+1}^{L}x_{j}\geq \tilde{q}}\tilde{f}^{\ast
}(x_{K-1+1},\cdots x_{L}|\hat{\kappa}_{1},\cdots ,\hat{\kappa}_{K-1})\leq
\alpha \right\} .
\end{equation*}%
By construction, $\tilde{q}_{1-\alpha ,t_{1}}$ is increasing in $\hat{\kappa}%
_{i},i\leq K-1$, $\tilde{q}_{1-\alpha ,t_{1}}<q_{1-\alpha ,t_{2}}$ and $%
\tilde{q}_{1-\alpha ,t_{1}}={q}_{1-(1+\eta _{\hat{\kappa}_{K-1}})\alpha
,t_{1}}$. ~\newline
~\newline
As for the remaining results, note that when $x_{L}\geq N-(K-1)$, 
\begin{equation*}
\begin{array}{c}
\displaystyle\prod_{i=1}^{K-1}\prod_{j=K-1+1}^{L}\left( 1-\frac{%
x_{j}-(N-(K-1))}{2\hat{\kappa}_{i}}-\tilde{f}(x_{j},\hat{\kappa}_{i})\right)
-1<0,%
\end{array}%
\end{equation*}%
and as the term on the left side is dominating terms of order $O\left( \hat{%
\kappa}_{K-1}^{-(1+\epsilon )}\right) $, when $\hat{\kappa}_{K-1}$ is large,
equations (\ref{a2})-(\ref{a2-2}) implies that 
\begin{equation*}
\begin{array}{c}
\displaystyle f_{\hat{\kappa}_{K},\cdots ,\hat{\kappa}_{L}|\hat{\kappa}%
_{1},\cdots ,\hat{\kappa}_{K-1}}^{\ast }(x_{K-1+1},\cdots x_{L}|\hat{\kappa}%
_{1},\cdots ,\hat{\kappa}_{K-1}) \\ 
\displaystyle<p_{N,L,K}\left( x_{K-1+1},\cdots x_{L}\right) .%
\end{array}%
\end{equation*}%
The above inequality implies $q_{1-\alpha ,t_{1}}<q_{1-\alpha ,t_{2}}$.
Similarly, when $x_{L}\geq N-(K-1)$, the derivative $\displaystyle\left( 1+%
\frac{N-(K-1)}{2\hat{\kappa}_{i}}+O\left( \hat{\kappa}_{i}^{-2}\right)
\right) (1-\frac{x_{j}}{\hat{\kappa}_{i}})^{\frac{1}{2}}$ with respect to $%
\hat{\kappa}_{i}$ is positive when $\hat{\kappa}_{i}$ is large, and jointly
with equation (\ref{eq conditional 3}) it implies that $q_{1-\alpha ,t_{1}}$
is increasing in $\hat{\kappa}_{i},i\leq K-1$. ~\newline
$\square $

Corollary \ref{coro2} further verifies the numerical analysis in \nocite%
{guggenberger2019more}Guggenberger et al. (2019) and extends to a general
setting. Results in Section \ref{Intermediate} essentially lead to the proof
of Proposition \ref{prop:sFAR}, and we briefly summarize it in \ref%
{sec:bound}.

\subsection{Proof of Proposition \protect\ref{prop:sFAR}.}

\label{sec:bound} The first result that the statistic sFAR is equal to the
sum of $K$ smallest eigenvalues follows directly the proof of Proposition %
\ref{prop:additional}.\ref{prop:additional 1}. Next, we prove that 
\begin{equation*}
\begin{array}{c}
\displaystyle \lim_{T\rightarrow \infty }\text{sFAR(}\lambda _{1}^{0})\prec
\chi ^{2}(K(N-(K-1)).%
\end{array}%
\end{equation*}

As suggested by Proposition \ref{prop:additional}.\ref{prop:additional 1},
we consider the sum of the $K$ smallest roots of $\left\vert \mu I_{2K-1} -
\Xi^{\prime }\Xi\right\vert =0$ , where 
\begin{equation*}
\begin{array}{c}
\Xi^{\prime }\Xi \sim \mathcal{W}_{2K-1}\left(N, I_{2K-1}, \mathcal{M}
^{\prime }\mathcal{M}\right).%
\end{array}%
\end{equation*}
Denote $\hat{\kappa}_{1}\geq \cdots \geq \hat{\kappa}_{L}$ eigenvalues of $%
\Xi^{\prime }\Xi$, $t$ the sum of the K smallest eigenvalues, $\hat{\kappa}%
_{K}, \cdots, \hat{\kappa}_{L}$, and eigenvalues of $\mathcal{M}^{\prime }%
\mathcal{M}$ satisfy that $\kappa_1> \kappa_2 > \cdots > \kappa_{K-1} >
\kappa_K=\kappa_{K+1}=\cdots=0$ and $\kappa_{K-1}$ goes to infinity. Then
Proposition \ref{prop:additional}.\ref{prop:additional 3} suggests that
under the null hypothesis $\lambda_{1}=\lambda_1^0$, 
\begin{equation*}
\begin{array}{c}
\lim\limits_{T\rightarrow \infty} P\left(\text{sFAR(}\lambda _{1}^{0}) >
q_{1-\alpha, t_2} \right) \leq \lim\limits_{T\rightarrow \infty} P\left(t >
q_{1-\alpha, t_2} \right)%
\end{array}%
\end{equation*}
with $q_{1-\alpha, t_2}$ the $1-\alpha$ quantile of of $\chi^2(K(N-(K-1)))$.
As $\kappa_{K-1} \rightarrow \infty$, Proposition \ref{prop:additional}.\ref%
{prop:additional 2} shows that $\hat{\kappa}_{K-1}^{-1} \rightarrow_p 0$,
and thus Corollary\ref{coro2} implies that 
\begin{equation*}
\begin{array}{c}
\lim\limits_{T\rightarrow \infty} P\left(t > q_{1-\alpha, t_2}
\right)=\lim\limits_{T\rightarrow \infty} E\left(\alpha + \tilde{\eta}%
_\alpha\right) =\alpha.%
\end{array}%
\end{equation*}
Therefore, $\lim\limits_{T\rightarrow \infty} P\left(\text{sFAR(}\lambda
_{1}^{0}) > q_{1-\alpha, t_2} \right) \leq \alpha$, and hence $\displaystyle %
\lim_{T\rightarrow \infty }\text{sFAR(}\lambda _{1}^{0})\prec \chi
^{2}(K(N-(K-1)).$

~\newline
$\square$

~\newline

\subsection{{Proof of Proposition \protect\ref{prop:sFAR far away}.}}

The smallest $K$ roots are calculated from the polynomial 
\begin{equation*}
\begin{array}{lc}
\left\vert \left( 
\begin{array}{cc}
I_{K} & 0 \\ 
-\lambda _{1}^{0\prime } & 0 \\ 
0 & I_{K-1}%
\end{array}
\right) ^{\prime }\left[ \mu \hat{\Psi}-\hat{\Phi}^{\prime }\hat{\Sigma}%
^{-1} \hat{\Phi}\right] \left( 
\begin{array}{cc}
I_{K} & 0 \\ 
-\lambda _{1}^{0\prime } & 0 \\ 
0 & I_{K-1}%
\end{array}
\right) \right\vert & =0,%
\end{array}%
\end{equation*}%
where $\hat{\Psi} = \hat{\Psi}^\prime = \left( 
\begin{array}{cc}
\hat{\Psi}_{X} & \hat{\Psi}_{XV} \\ 
\hat{\Psi}_{VX} & \hat{\Psi}_{V}%
\end{array}
\right)$ and $\hat{\Psi}_{X}$ and $\hat{\Psi}_{V}$ are $K\times K$
submatrices of $\hat{\Psi}$, We specify $\lambda _{1}^{0}=c\bar{\lambda}%
_{1}^{0},$ with $c=(\lambda _{1}^{0\prime }\lambda _{1}^{0})^{\frac{1}{2}},$
so $\bar{\lambda}_{1}^{0\prime }\bar{\lambda}_{1}^{0}=1$ and $\bar{\lambda}%
_{1,\perp }^{0}:K\times (K-1),$ $\bar{\lambda}_{1,\perp }^{0\prime }\bar{%
\lambda}_{1}^{0}\equiv 0,$ $\bar{\lambda}_{1,\perp }^{0\prime }\bar{\lambda}%
_{1,\perp }^{0}\equiv I_{K-1}.$ Hence, $\text{diag}(( \bar{\lambda}_{1,\perp
}^{0} \text{ }\vdots \text{ } \bar{\lambda}_{1}^{0}), I_{K-1})$ is an
invertible orthonormal matrix. Pre- and post multiplying the matrices in the
determinant with it does therefore not affect the characteristic roots of he
following polynomial:%
\begin{equation*}
\begin{array}{ll}
\displaystyle \left\vert \left( 
\begin{array}{ccc}
\bar{\lambda}_{1,\perp }^{0} & \bar{\lambda}_{1}^{0} & 0 \\ 
0 & 0 & I_{K-1}%
\end{array}%
\right)^{\prime }\left( 
\begin{array}{cc}
I_{K} & 0 \\ 
-\lambda _{1}^{0^{\prime }} & 0 \\ 
0 & I_{K-1}%
\end{array}
\right) ^{\prime } \left[ \mu \hat{\Psi}-\hat{\Phi}^{\prime }\hat{\Sigma}%
^{-1} \hat{\Phi}\right] \left( 
\begin{array}{cc}
I_{K} & 0 \\ 
-\lambda _{1}^{0\prime } & 0 \\ 
0 & I_{K-1}%
\end{array}
\right) \right. &  \\ 
\times \left. \left( 
\begin{array}{ccc}
\bar{\lambda}_{1,\perp }^{0} & \bar{\lambda}_{1}^{0} & 0 \\ 
0 & 0 & I_{K-1}%
\end{array}%
\right) \right\vert =0, & 
\end{array}%
\end{equation*}%
which can be rewritten as 
\begin{equation*}
\begin{array}{ll}
\left\vert \left( 
\begin{array}{ccc}
1 & 0 & 0 \\ 
0 & -c & 0 \\ 
0 & 0 & I_{2(K-1)-1}%
\end{array}
\right) ^{\prime }\left( 
\begin{array}{ccc}
\bar{\lambda}_{1,\perp }^{0} & -\bar{\lambda}_{1}^{0}/c & 0 \\ 
0 & 1 & 0 \\ 
0 & 0 & I_{K-1}%
\end{array}
\right) ^{\prime } \right. &  \\ 
\left. \times \left[ \mu \hat{\Psi}-\hat{\Phi}^{\prime }\hat{\Sigma}^{-1} 
\hat{\Phi}\right] \left( 
\begin{array}{ccc}
\bar{\lambda}_{1,\perp }^{0} & -\bar{\lambda}_{1}^{0}/c & 0 \\ 
0 & 1 & 0 \\ 
0 & 0 & I_{K-1}%
\end{array}
\right) \left( 
\begin{array}{ccc}
1 & 0 & 0 \\ 
0 & -c & 0 \\ 
0 & 0 & I_{2(K-1)-1}%
\end{array}
\right) \right\vert & =0.%
\end{array}%
\end{equation*}%
Hence, when $c$ goes to infinity, the characteristic polynomial becomes:%
\begin{equation*}
\begin{array}{rl}
\left\vert \left( 
\begin{array}{cc}
\bar{\lambda}_{1,\perp }^{0} & 0 \\ 
0 & I_{K}%
\end{array}
\right)^{\prime }\left[ \mu \hat{\Psi}-\hat{\Phi}^{\prime }\hat{\Sigma}^{-1} 
\hat{\Phi}\right] \left( 
\begin{array}{cc}
\bar{\lambda}_{1,\perp }^{0} & 0 \\ 
0 & I_{K}%
\end{array}
\right) \right\vert & =0.%
\end{array}%
\end{equation*}
Since $\hat{\Phi}=(\hat{d}$ $\vdots $ $\hat{\beta})$, we have $\hat{\Phi}%
\times\text{diag}(\bar{\lambda}_{1,\perp }^{0}, I_{K-1}) =( \text{ }\hat{d}%
\bar{\lambda}_{1,\perp }^{0}\text{ }\vdots \text{ }\hat{\beta}).$ The subset
AR statistic now equals the sum of the $K$ smallest root of the above
characteristic polynomial which depend on $\bar{\lambda}_{1,\perp }^{0}.$
For example, when $\bar{\lambda}_{1}=e_{1,k},$ $\bar{ \lambda}_{1,\perp
}^{0}=\binom{0}{I_{k-1}}$ so $\hat{d}\bar{\lambda}_{1,\perp }^{0}=\hat{d}%
_{2} $ etc.

Let $\hat{\Psi}(\bar{\lambda}_{1,\perp }^{0}) = \left( 
\begin{array}{cc}
\bar{\lambda}_{1,\perp }^{0} & 0 \\ 
0 & I_{K}%
\end{array}%
\right)^{\prime }\hat{\Psi}\left( 
\begin{array}{cc}
\bar{\lambda}_{1,\perp }^{0} & 0 \\ 
0 & I_{K}%
\end{array}%
\right)$. Then the above characteristic polynomial can be rewritten, by pre-
and post multiplying the matrices in the determinant with $\hat{\Psi}(\bar{%
\lambda}_{1,\perp }^{0})^{-1/2\prime} $ and $\hat{\Psi}(\bar{\lambda}%
_{1,\perp }^{0})^{-1/2}$ respectively, as 
\begin{equation*}
\begin{array}{rl}
\left\vert \mu I_{2K-1} -\hat{\Psi}(\bar{\lambda}_{1,\perp
}^{0})^{-1/2\prime}( \text{ }\hat{d}\bar{\lambda} _{1,\perp }^{0}\text{ }%
\vdots \text{ }\hat{\beta})^{\prime }\hat{\Sigma}^{-1} ( \text{ }\hat{d}\bar{%
\lambda} _{1,\perp }^{0}\text{ }\vdots \text{ }\hat{\beta})\hat{\Psi}(\bar{%
\lambda}_{1,\perp }^{0})^{-1/2} \right\vert & =0.%
\end{array}%
\end{equation*}
The lower $K\times K$ principal submatrix of 
\begin{equation*}
\hat{\Psi}(\bar{\lambda}_{1,\perp }^{0})^{-1/2\prime}( \text{ }\hat{d}\bar{%
\lambda}_{1,\perp }^{0}\text{ }\vdots \text{ }\hat{\beta})^{\prime }\hat{%
\Sigma}^{-1} ( \text{ }\hat{d}\bar{\lambda}_{1,\perp }^{0}\text{ }\vdots 
\text{ }\hat{\beta})\hat{\Psi}(\bar{\lambda}_{1,\perp }^{0})^{-1/2}
\end{equation*}
by construction is 
\begin{equation*}
\hat{\Psi}_{V}^{-1/2\prime } \hat{\beta}^\prime \hat{\Sigma}^{-1} \hat{\beta}%
\hat{\Psi}_{V}^{-1/2},
\end{equation*}
and thus Cauchy's interlacing inequality implies the sum of the K smallest
roots of the above polynomial is bounded from below by the minimum
eigenvalue of $\hat{\Psi}_{V}^{-1/2\prime } \hat{\beta}^\prime \hat{\Sigma}%
^{-1} \hat{\beta}\hat{\Psi}_{V}^{-1/2}$. Therefore, the limit sFAR statistic
is bounded from below uniformly by the minimum eigenvalue of $\hat{\Psi}%
_{V}^{-1/2\prime } \hat{\beta}^\prime \hat{\Sigma}^{-1} \hat{\beta}\hat{\Psi}%
_{V}^{-1/2}$.

~\newline
$\square$

~\newline

 ~\newpage
\title{{\Large Online Supplementary Appendix: }\\
	Identification Robust Inference for the Risk Premium in Term Structure
	Models }
\author{Frank Kleibergen\thanks{
		Amsterdam School of Economics, University of Amsterdam, Roetersstraat 11,
		1018 WB Amsterdam, The Netherlands. Email: f.r.kleibergen@uva.nl.} \and %
	Lingwei Kong\thanks{
		Faculty of Economics and Business, University of Groningen, Nettelbosje 2,
		9747 AE Groningen, The Netherlands. Email: l.kong@rug.nl. } }
\date{}
\maketitle
\setcounter{page}{1}
\thispagestyle{empty}\bigskip

\newpage \appendix
\renewcommand{\theequation}{\thesection.\arabic{equation}}

\section{Comparisons of settings of term structure model}

The framework we use is comparable to other models, and our proposed tests
can be easily adapted to these models. We provide several examples.

\subsection{Linear asset pricing models}

With extra restrictions that $\Phi=0, \lambda_{1}=0$ and $g^{(n-1)}(\cdot)$
is a constant function, the first two equations in Assumption 2.1 are
equivalent to the linear asset pricing model from Kleibergen et al. (2022)%
\nocite{kleibergen2019identification}: 
\begin{equation*}
	\begin{array}{l}
		X_t = \mu + v_{t}, \\ 
		r_{t+1, n-1}= c_0+ \beta^{(n-1)^{\prime }} \lambda_0 +\beta^{(n-1)^{\prime
		}} v_{t+1} +e_{t+1,n-1}.%
	\end{array}%
\end{equation*}
With non-zero $\lambda_1$, the latter equation would be 
\begin{equation*}
	\begin{array}{c}
		r_{t+1, n-1}= c_0+ \beta^{(n-1)^{\prime }} \left(\lambda_0+ \lambda_1
		X_t\right) +\beta^{(n-1)^{\prime }} v_{t+1} +e_{t+1,n-1},%
	\end{array}%
\end{equation*}
which indicates our approach can be used for linear asset pricing models
with time-varying risk premia, $\lambda_0+ \lambda_1 X_t$.

\subsection{Dynamic affine term structure Models}

\subsubsection{Adrian et al. (2013)}

\nocite{adrian2013pricing} \label{s} This section discusses the basic
framework of \nocite{adrian2013pricing}Adrian et al. (2013). Consider a $%
K\times 1$ vector of state variables $X_{t}$ (VAR(1)): 
\begin{equation}
	\begin{array}{c}
		X_{t+1}=\mu +\Phi X_{t}+v_{t+1}.\label{s1}%
	\end{array}%
\end{equation}%
The pricing kernel is assumed to be exponential affine in innovation factors 
$v_{t}\sim _{i.i.d}N(0,\Sigma _{v})$ such that 
\begin{equation*}
	\begin{array}{c}
		M_{t+1}=\exp \left( -r_{t}-\frac{1}{2}\lambda _{t}^{\prime }\lambda
		_{t}-\lambda _{t}^{\prime }\Sigma _{v}^{-1/2}v_{t+1}\right) ,\label{s2}%
	\end{array}%
\end{equation*}%
where the market price of risk $\lambda _{t}$ is also an affine function in $%
X_{t}$, 
\begin{equation}
	\begin{array}{c}
		\lambda _{t}=\Sigma _{v}^{-1/2}\left( \lambda _{0}+\lambda _{1}v_{t}\right)
		. \label{s3}%
	\end{array}%
\end{equation}%
\nocite{adrian2013pricing}Adrian et al. (2013) employ, instead of yields,
one-period excess holding returns $R_{t+1,n}:=R_{t+1,\tau _{n}}$ ($%
R_{t+1,\tau _{n}}$ denotes the log excess holding return of a bond maturing
in $\tau _{n}$ periods) for analysis: 
\begin{equation*}
	\begin{array}{c}
		R_{t+1,\tau _{n}}=\ln P_{t+1,\tau _{n}}-\ln P_{t,\tau _{n}+1}-r_{t},\label%
		{eq:excess returns}%
	\end{array}%
\end{equation*}%
where $r_{t}=\ln P_{t,1}$ represents the continuously compounded risk-free
rate. The property of the pricing kernel such that $P_{t,\tau _{n}+1}=%
\mathbb{E}M_{t+1}P_{t+1,\tau _{n}}$ implies that 
\begin{equation}
	\begin{array}{c}
		\mathbb{E}_{t}\exp \left( R_{t+1,\tau _{n}}-\frac{1}{2}\lambda _{t}^{\prime
		}\lambda _{t}-\lambda _{t}^{\prime }\Sigma _{v}^{-1/2}v_{t+1}\right) =1. %
		\label{eq:r}%
	\end{array}%
\end{equation}%
%
%
%
From the moment generating function of a normal distribution and the
assumption that $\{R_{t+1},v_{t+1}\}$ are jointly normal,\footnote{$E\left(
	e^{t^{\prime }(X+Y)}\right) =e^{t^{\prime }\mu _{X+Y}+\frac{1}{2}t^{\prime
		}\Sigma _{X+Y}t}$} we know from (\ref{eq:r}) that 
\begin{equation}
	\begin{array}{l}
		E_{t}\left( R_{t+1,\tau }\right) =E_{t}\left( \frac{1}{2}\lambda
		_{t}^{\prime }\lambda _{t}+\lambda _{t}^{\prime }\Sigma
		_{v}^{-1/2}v_{t+1}\right) -\frac{1}{2}\text{Var}\left( R_{t+1,\tau }-\lambda
		_{t}^{\prime }\Sigma _{v}^{-1/2}v_{t+1}\right) \nonumber \\ 
		=\frac{1}{2}\lambda _{t}^{\prime }\lambda _{t}-\frac{1}{2}\text{Var}\left(
		R_{t+1,\tau }\right) -\frac{1}{2}\text{Var}\left( \lambda _{t}^{\prime
		}\Sigma _{v}^{-1/2}v_{t+1}\right) +\text{Cov}\left( R_{t+1,\tau },\left(
		\lambda _{t}^{\prime }\Sigma _{v}^{-1/2}v_{t+1}\right) ^{\prime }\right) %
		\nonumber \\ 
		=\text{Cov}\left( R_{t+1,\tau },v_{t+1}^{\prime }\Sigma _{v}^{-1/2}\lambda
		_{t}\right) -\frac{1}{2}\text{Var}\left( R_{t+1,\tau }\right) \nonumber \\ 
		=\beta ^{(\tau )^{\prime }}\Sigma _{v}^{1/2}\lambda _{t}-\frac{1}{2}\text{%
			Var }\left( R_{t+1,\tau }\right) \nonumber \\ 
		=\beta ^{(\tau )^{\prime }}(\lambda _{0}+\lambda _{1}X_{t})-\frac{1}{2}\text{
			Var}\left( R_{t+1,\tau }\right) .\label{eq1}%
	\end{array}%
\end{equation}%
We can decompose unexplained returns into two parts: one explained by the
innovation shock ($v_{t+1}$), the other by errors i.i.d $e_{t+1}^{(n-1)}$
orthogonal to $v_{t+1}$ with variance $\sigma ^{2}$ such that 
\begin{equation}
	\begin{array}{c}
		R_{t+1,\tau }-E_{t}\left( R_{t+1,\tau }\right) =\gamma _{t}^{(\tau )^{\prime
		}}v_{t+1}+e_{t+1,\tau },\label{eq:2}%
	\end{array}%
\end{equation}%
from which we know $\text{Cov}\left( R_{t+1,\tau },v_{t+1}^{\prime }\Sigma
_{v}^{-1/2}\right) =\text{Cov}\left( \gamma _{t}^{(\tau )^{\prime
}}v_{t+1}+e_{t+1,\tau },v_{t+1}^{\prime }\Sigma _{v}^{-1/2}\right) =\gamma
_{t}^{(\tau )^{\prime }}\Sigma _{v}^{1/2}$ and thus $\gamma _{t}^{(\tau
	)^{\prime }}=\text{Cov}\left( R_{t+1,\tau },v_{t+1}^{\prime }\Sigma
_{v}^{-1/2}\right) \Sigma _{v}^{-1/2}=\beta _{t}^{(\tau )^{\prime }}$, and
then~\newline
$\text{Var}\left( R_{t+1,\tau }\right) =\text{Var}\left( \gamma _{t}^{(\tau
	)^{\prime }}v_{t+1}+e_{t+1,\tau }\right) =\beta _{t}^{(\tau )^{\prime
}}\Sigma _{v}\beta _{t}^{(\tau )^{\prime }}+\sigma _{e}^{2}$. The above
results imply that from (\ref{eq:2}) we have, 
\begin{equation*}
	\begin{array}{l}
		R_{t+1,\tau }=E_{t}\left( R_{t+1,\tau }\right) +\gamma _{t}^{(\tau )^{\prime
		}}v_{t+1}+e_{t+1,\tau }\nonumber \\ 
		=\beta ^{(\tau )^{\prime }}(\lambda _{0}+\lambda _{1}X_{t})-\frac{1}{2}\text{
			Var}\left( R_{t+1,\tau }\right) +\beta ^{(\tau )^{\prime
		}}v_{t+1}+e_{t+1,\tau }\nonumber \\ 
		=\beta ^{(\tau )^{\prime }}(\lambda _{0}+\lambda _{1}X_{t})-\frac{1}{2}
		\left( \beta _{t}^{(\tau )^{\prime }}\Sigma _{v}\beta _{t}^{(n-1)^{\prime
		}}+\sigma _{e}^{2}\right) +\beta ^{(\tau )^{\prime }}v_{t+1}+e_{t+1,\tau },%
	\end{array}%
\end{equation*}%
which coresponds to the second equation in Assumption 2.1.

~~\newline
The following Proposition \ref{s:prop1} displays one specification nested
within the Adrian et al. (2015) framework that naturally satisfies the
Kronecker structure in Assumption 5.1.

\begin{proposition}
	\label{s:prop1}~~\newline
\end{proposition}

\subsubsection{Adrian et al. (2015)}

\nocite{adrian2015regression} Adrian et al. (2015) consider a less
restricted model setting with a different pricing kernel specification. The
model closely resembles affine term structure models, and it takes into
account the unspanned factors. Unspanned factors refer to factors that are
not correlated with the contemporaneous excess returns but contribute to the
forecasts of excess returns. Though usually, factors are either significant
factors for the cross section or significant for forecasts, this model takes
into account the possibility that factors act in both dimensions.

The state variables ($X_{t}$) satisfy equation (\ref{s1}), $X_{t+1}=\mu
+\Phi X_{t}+v_{t+1}$, and fall into three categories: $X_{1,t}\in \mathbb{R}%
^{k_{1}}$, $X_{2,t}\in \mathbb{R}^{k_{2}}$, $X_{3,t}\in \mathbb{R}^{k_{3}}$,
where $K_{C}\times 1$ vector $C_{t}=\left( X_{1,t}^{\prime },X_{2,t}^{\prime
}\right) ^{\prime }$ is for the cross section, whose innovations have
significant non-zero $\beta $'s in (\ref{hah}), $K_{F}\times 1$ vector $%
F_{t}=\left( X_{2,t}^{\prime },X_{3,t}^{\prime }\right) ^{\prime }$ for
forecasts and $u_{t}=\left( v_{1,t}^{\prime },v_{2,t}^{\prime }\right)
^{\prime }$ denotes innovations in the $C_{t}$ factors. The new pricing
kernel $M_{t+1}$ in Adrian et al. (2015) satisfies 
\begin{equation*}
	\begin{array}{c}
		\mathbb{E}_{t}\left( M_{t+1}R_{i,t+1}\right) =\mathbb{E}\left(
		M_{t+1}R_{i,t+1}|\mathcal{F}_{t}\right) =0, \\ 
		\frac{M_{t+1}-\mathbb{E}\left( M_{t+1}|\mathcal{F}_{t}\right) }{\mathbb{E}
			\left( M_{t+1}|\mathcal{F}_{t}\right) }=-\lambda _{t}^{\prime }\text{var}
		\left( u_{t+1}|\mathcal{F}_{t}\right) ^{-\frac{1}{2}}u_{t+1},%
	\end{array}%
\end{equation*}%
where $R_{i,t+1}$ denotes holding period return in excess of the risk free
rate of asset $i$. As a consequence, 
\begin{equation*}
	\begin{array}{c}
		\mathbb{E}\left( R_{i,t+1}|\mathcal{F}_{t}\right) =-\frac{\text{cov}\left(
			M_{t+1},R_{i,t+1}|\mathcal{F}_{t}\right) }{\mathbb{E}\left( M_{t+1}|\mathcal{%
				\ F}_{t}\right) }=\beta _{i,t}^{\prime }\left( \lambda _{0}+\lambda
		_{1}F_{t},\right) ,%
	\end{array}%
\end{equation*}%
where $\beta _{i,t}=\text{var}\left( u_{t+1}|\mathcal{F}_{t}\right) ^{-1}%
\text{ cov}\left( u_{t+1},R_{i,t+1}|\mathcal{F}_{t}\right) $. Therefore, we
can write 
\begin{equation}
	\begin{array}{c}
		R_{i,t+1}=\beta _{i,t}^{\prime }\left( \lambda _{0}+\lambda _{1}F_{t}\right)
		+\beta _{i,t}^{\prime }u_{t+1}+e_{i,t+1},\label{hah}%
	\end{array}%
\end{equation}%
the structure of which implies that our proposed tests are applicable when
we assume time-constant $\beta $'s in this framework.

\subsubsection{Hamilton and Wu (2012)}

Hamilton and Wu (2012)\nocite{hamilton2012identification} analyze the yield,
and they assume that equation (\ref{s1})-(\ref{s3}) hold\footnote{%
	Our notations are slightly from the original ones in Hamilton and Wu (2012) 
	\nocite{hamilton2012identification} to be consistent from previous
	discussions.}, and the risk free one-period yield $y_{t,r}^{(1)}$, the yield
of a $n$-period bond, $y_{t}^{(n)}$, satisfies 
\begin{equation*}
	\begin{array}{c}
		y_{t,r}^{(1)}=\delta _{0}+\delta _{1}X_{t},y_{t}^{(n)}=a_{n}+b_{n}^{\prime
		}X_{t}+u_{{t},{n}},%
	\end{array}%
\end{equation*}%
where 
\begin{equation*}
	\begin{array}{l}
		\mu ^{Q}=\mu -\lambda _{0}, \\ 
		\Phi ^{Q}=\Phi -\lambda _{1}, \\ 
		b_{n}=\left( I_{K}+\Phi ^{Q^{\prime }}+\cdots +(\Phi ^{Q^{\prime
		}})^{n-1}\right) \delta _{1}/n, \\ 
		a_{n}=\delta _{0}+\left( b_{1}^{\prime }+2b_{2}^{\prime }+\cdots
		+(n-1)b_{n-1}^{\prime }\right) \mu ^{Q}/n \\ 
		-\left( b_{1}^{\prime }\Sigma _{v}b_{1}+(2b_{2})^{\prime }\Sigma
		_{v}(2b_{2})+\cdots +(n-1)b_{n-1}^{\prime }\Sigma _{v}((n-1)b_{n-1})\right)
		/(2n).\label{restriction extra}%
	\end{array}%
\end{equation*}%
If we consider the data transformation: $%
r_{t}^{(n)}=ny_{t}^{(n)}-(n+1)y_{t-1}^{(n+1)}+y_{t-1,r}^{(1)}$, the above
equations imply that 
\begin{equation*}
	\begin{array}{l}
		r_{t}^{(n)}=c^{(n)^{\prime }}+\beta ^{(n)^{\prime }}\lambda
		_{1}X_{t-1}+\beta ^{(n)^{\prime }}v_{t}+e_{t},%
	\end{array}%
\end{equation*}%
with 
\begin{equation*}
	\begin{array}{l}
		c^{(n)^{\prime }}=na_{n}-(n+1)a_{n+1}+\delta _{0}=\beta ^{(n)^{\prime
		}}\lambda _{0}+(\beta ^{(n)^{\prime }}\Sigma _{v}\beta ^{(n)}/2-\beta
		^{(n)^{\prime }}\mu ), \\ 
		\beta ^{(n)^{\prime }}=nb_{n}^{\prime }, \\ 
		e_{t}=nu_{t,n}-(n+1)u_{t-1,n+1},%
	\end{array}%
\end{equation*}%
which then indicates our tests are valid with proper restrictions on $e_{t}$%
. The constant term $c^{(n)}$ has different structure than the one indicated
in \nocite{adrian2013pricing}Adrian et al. (2013) (or Assumption 2.1.(b)),
but this can be resolved by minor adaptions.

\newpage

\section{Additional discussions concerning the estimation strategies in
	Adrian et al. (2013)}

\textit{Proof of Proposition 3.1}~\newline
In the following, $\beta$ is an $N\times K$ matrix, which is the transpose
of the factor loading matrix we use in the main context, and we denote $%
f=(\mu, d)$, $\Psi_\beta$ as a stack form of $\psi_\beta$, the asymptotic
distribution of $\sqrt{T} \text{vec}\left( \hat{\beta} ^{\prime
}-\beta^{\prime }\right)$, such that $\left(\Psi_\beta\right) _{\substack{ %
		i,j  \\ 1\leq i\leq K, 1\leq j\leq N }} = \left(\psi_\beta\right)_{(i-1)N+j} 
$, and denote $\mathcal{K}_{a,b}$ the commutation matrix such that $\mathcal{%
	K}_{a,b}\text{vec}(A)=\text{vec} (A^{\prime })$ with $A$ a $a\times b$
matrix.~\newline
~~~\newline
The statement (a) is a direct result following \nocite{adrian2013pricing}%
Adrian et al.(2013). Here we only list the specifications of $\mathcal{V}%
_\beta, \mathcal{C}_{\Lambda,\beta}, \mathcal{V}_\Lambda$. 
\begin{equation}
	\begin{array}{l}
		\mathcal{V}_\beta= \Sigma_v^{-1} \otimes \sigma^2I_N, \\ 
		\mathcal{V}_\Lambda= 1^{\mu}+C+C^{\prime }\nonumber \\ 
		+ \sigma^2\left( \gamma_{ZZ}^{-1}\otimes \left(\beta\beta^{\prime
		}\right)^{-1} \right) +\sigma^2\left( \Lambda^{\prime }\Sigma_v^{-1}\Lambda
		\otimes \left(\beta\beta^{\prime }\right)^{-1} \right)\nonumber \\ 
		+\sigma^2 \left(\varrho^{\prime }\varrho\otimes \left(({\hat{\beta}}\hat{%
			\beta}^{\prime})^{-1}{\beta}A_\beta^{\prime }(I\otimes
		\Sigma_v)\right)A_\beta \beta^{\prime }({\hat{\beta}}\hat{\beta}%
		^{\prime})^{-1}\right)\nonumber \\ 
		+\frac{1}{4}\left( \varrho^{\prime }\varrho\otimes ({\hat{\beta}}\hat{\beta}%
		^{\prime})^{-1}{\beta} \left({B}^*
		\right)(I_{K^2}+\kappa_{K,K})(\Sigma_v\otimes \Sigma_v){B}^{*^{\prime }}
		\beta^{\prime }({\hat{\beta}}\hat{\beta}^{\prime})^{-1}\right) \nonumber \\ 
		\times \frac{\sigma^4}{2}\left(\varrho^{\prime }\otimes ({\hat{\beta}}\hat{%
			\beta}^{\prime})^{-1}{\beta} \right) \iota_N \iota_N^{\prime
		}\left(\varrho\otimes {\ \beta^{\prime }}({\hat{\beta}}\hat{\beta}%
		^{\prime})^{-1} \right), \\ 
		\mathcal{C}_{\Lambda,\beta}=-\kappa_{K+1,K}\left(\sigma^2(({\hat{\beta}}\hat{%
			\beta}^{\prime})^{-1}{\beta}\otimes \Lambda^{\prime }\Sigma_v^{-1} ) \right)
		+ \sigma^2\left(\varrho^{\prime }\otimes \left(({\hat{\beta}}\hat{\beta}%
		^{\prime})^{-1}{\ \beta}A_\beta^{\prime }\right)\right),%
	\end{array}%
\end{equation}
where $1^{\mu}=\gamma_{ZZ}^{-1}\otimes \Sigma_v$, 
\begin{equation*}
	\begin{array}{c}
		C=-\left( \Lambda^{\prime }\otimes ({\hat{\beta}}\hat{\beta}^{\prime})^{-1}{%
			\beta} \right)\left(\sigma^2\Sigma_v^{-1}\otimes I \right)
		\kappa_{N,K}^{\prime }\left(\varrho^{\prime }\otimes \left(({\hat{\beta}}%
		\hat{\beta}^{\prime})^{-1}{\beta} A_\beta^{\prime }(I\otimes
		\Sigma_v)\right)\right)^{\prime },%
	\end{array}%
\end{equation*}
	and $A_\beta$ is a $NK\times K$ matrix $A_\beta= \oplus_{i=1}^N \beta^{(i)}$
	, where $\oplus$ denotes the matrix direct sum, such that $A \oplus B =
	\left( 
	\begin{matrix}
		A & 0 \\ 
		0 & B%
	\end{matrix}
	\right)$.~\newline
	~~~\newline
	We discuss two special cases to show the statement (b). We first introduce
	some new notation, given the three-stage estimator as in the form in Adrian
	et al. (2013): 
	\begin{equation*}
		\begin{array}{c}
			\sqrt{T}\hat{\Lambda}=\sqrt{T}(({\hat{\beta}}\hat{\beta}^{\prime})^{-1}\hat{
				\beta }\left(rx +\frac{1}{2}\left( \hat{B}^*\text{vec}(\hat{\Sigma_v}) +\hat{
				\sigma }^2\iota_N\right)\iota_T^{\prime }\right)M_{\hat{V}^{\prime }}
			Z_-^{\prime }\left( Z_-M_{\hat{V}^{\prime }}Z_-^{\prime }\right)^{-1},%
		\end{array}%
	\end{equation*}
	where $\sqrt{T}\hat{\Lambda}$ is the summation of 
	\begin{equation*}
		\begin{array}{l}
			\mathcal{T}_{1}=\sqrt{T} \left(\hat{\beta}\hat{\beta}^{\prime }\right)^{-1} 
			\hat{\beta}\left(rx +\frac{1}{2}\left({B}^*\text{vec}({\Sigma_v}) + {\sigma}
			^2\iota_N\right)\iota_T^{\prime }\right)M_{\hat{V}^{\prime }} Z_-^{\prime
			}\left( Z_-M_{\hat{V}^{\prime }}Z_-^{\prime }\right)^{-1}, \\ 
			\mathcal{T}_{2}=\frac{1}{2} \sqrt{T} \left(\hat{\beta}\hat{\beta}^{\prime
			}\right)^{-1}\hat{\beta} \left( \hat{B}^*\text{vec}(\hat{\Sigma_v})- {B}^* 
			\text{vec}({\Sigma_v}) \right)\rho^{\prime }, \\ 
			\mathcal{T}_{3}=\frac{1}{2} \sqrt{T} \left(\hat{\beta}\hat{\beta}^{\prime
			}\right)^{-1}\hat{\beta} \left( \hat{\sigma}^2 - \hat{\sigma}^2
			\right)\iota_N\rho^{\prime },%
		\end{array}%
	\end{equation*}
	and $\mathcal{T}_{i}$ can be decomposed as summation of $\mathcal{T} _{i,j}$
	: 
	\begin{equation*}
		\begin{array}{l}
			\mathcal{T}_{1,1}= \sqrt{T}\Lambda, \\ 
			\mathcal{T}_{1,2}= \sqrt{T}V M_{\hat{V}^{\prime }} Z_-^{\prime }\left(
			Z_-M_{ \hat{V}^{\prime }}Z_-^{\prime }\right)^{-1}, \\ 
			\mathcal{T}_{1,3}=(({\hat{\beta}}\hat{\beta}^{\prime})^{-1}\hat{\beta} \sqrt{%
				T} \left(\hat{f}-f \right)+ o_p(1), \\ 
			\mathcal{T}_{1,4}= -(({\hat{\beta}}\hat{\beta}^{\prime})^{-1}\hat{\beta}%
			\sqrt{T} \left(\hat{\beta}^{\prime }-\beta^{\prime }\right) \Lambda, \\ 
			\mathcal{T}_{1,5}= \sqrt{T}(({\hat{\beta}}\hat{\beta}^{\prime})^{-1}\hat{%
				\beta} \left(\hat{\beta}^{\prime }-\beta^{\prime }\right) V M_{\hat{V}%
				^{\prime }} Z_-^{\prime }\left( Z_-M_{\hat{V}^{\prime }}Z_-^{\prime
			}\right)^{-1}, \\ 
			\mathcal{T}_{2,1}= (({\hat{\beta}}\hat{\beta}^{\prime})^{-1}\hat{\beta}
			A_\beta^{\prime }(I\otimes \Sigma_v) \text{vec}\left(\sqrt{T} \left(\hat{
				\beta}- {\beta}\right)\right)\varrho+o_p(1), \\ 
			\mathcal{T}_{2,2}=\frac{1}{2}(({\hat{\beta}}\hat{\beta}^{\prime})^{-1}\hat{%
				\beta} \left({B}^*\sqrt{T}\left( \text{vec}\left(\frac{VV^{\prime }}{T}-{%
				\Sigma_v} \right)\right)\right)\varrho+o_p(1), \\ 
			\mathcal{T}_{3~~}=\frac{1}{2}(({\hat{\beta}}\hat{\beta}^{\prime})^{-1}\hat{%
				\beta} \left(\sqrt{T}\left( \frac{\text{tr}\left( EE^{\prime }\right) }{NT}-{%
				\sigma} ^2\right) \right)\iota_N\varrho+o_p(1).%
		\end{array}%
	\end{equation*}
	We have 
	\begin{equation*}
		\begin{array}{l}
			\text{vec}\left(\mathcal{T}_{1,1} \right)=\sqrt{T}\Lambda, \\ 
			\text{vec}\left(\mathcal{T}^{\mu=0}_{1,2}
			\right)=\left(\gamma_{xx}^{-1}\otimes I\right)\text{vec}\left[
			0~~VX_-^{\prime }/\sqrt{T} \right]+o_p(1), \\ 
			\text{vec}\left(\mathcal{T}^{\mu}_{1,2} \right)=
			\left(\gamma_{zz}^{-1}\otimes I \right)\text{vec}\left(VZ_-^{\prime }/\sqrt{
				T }\right) +o_p(1), \\ 
			\text{vec}\left(\mathcal{T}_{1,3} \right)=\left(I\otimes(({\hat{\beta}}\hat{%
				\beta}^{\prime})^{-1}\hat{\beta} \right)\text{vec}\left(\sqrt{T}\left( \hat{f%
			} -f\right) \right)+o_p(1), \\ 
			\text{vec}\left(\mathcal{T}_{1,4} \right)=-\left( \Lambda^{\prime }\otimes ({%
				\hat{\beta}}\hat{\beta}^{\prime})^{-1}\hat{\beta} \right)\kappa_{K_C,N} 
			\text{ vec}\left(\sqrt{T}\left(\hat{\beta}^{\prime }-\beta^{\prime }\right)
			\right)+ o_p(1), \\ 
			\text{vec}\left(\mathcal{T}_{1,5} \right)=o_p(1), \\ 
			\text{vec}\left(\mathcal{T}_{2,1} \right)=\left(\varrho^{\prime }\otimes
			\left(({\hat{\beta}}\hat{\beta}^{\prime})^{-1}\hat{\beta} A_\beta^{\prime
			}(I\otimes \Sigma_v)\right)\right) \text{vec}\left(\sqrt{T} \left(\hat{\beta}
			- {\beta}\right)\right) +o_p(1), \\ 
			\text{vec}\left(\mathcal{T}_{2,2} \right)=\left( \varrho^{\prime }\otimes 
			\frac{1}{2}({\hat{\beta}}\hat{\beta}^{\prime})^{-1}\hat{\beta} \left({B}^*
			\right)\right)\text{vec}\left(\sqrt{T}\left(\text{vec}\left( \frac{
				VV^{\prime }}{T}\right)-\text{vec}({\Sigma_v})\right)\right)+o_p(1), \\ 
			\text{vec}\left(\mathcal{T}_{3} \right)=\left(\varrho^{\prime }\otimes \frac{
				1}{2}({\hat{\beta}}\hat{\beta}^{\prime})^{-1}\hat{\beta} \right)
			\iota_N\left( \sqrt{T}\left( \frac{\text{tr}\left( EE^{\prime }\right)}{NT}-{%
				\sigma} ^2\right) \right)+o_p(1).%
		\end{array}%
	\end{equation*}
	where the term $\mathcal{T}_{1,2}$ depends on whether or not we impose the
	assumption that $\mu=0$. Next we would like to look for a non-full rank case
	of $\beta$ and describe the asymptotic properties of $\hat{\Lambda}$ under
	those cases. In the following, we abuse the equal sign a bit where we may
	directly ignore the asymptotically negligible terms.~\newline
	~~ ~\newline
	Denote 
	\begin{equation*}
		\begin{array}{c}
			\left[ 
			\begin{matrix}
				\text{vec}\left(\hat{f}\right) \\ 
				\text{vec}\left(\hat{\beta}\right)%
			\end{matrix}
			\right]=\left( 
			\begin{matrix}
				\text{vec}\left({f}\right) \\ 
				\text{vec}\left({\beta}\right)%
			\end{matrix}
			\right)+\frac{1}{\sqrt{T}} \left( 
			\begin{matrix}
				\psi_{f} \\ 
				\psi_{\beta}%
			\end{matrix}
			\right) +o_p\left(\frac{1}{\sqrt{T}}\right), \label{eq:a:42}%
		\end{array}%
	\end{equation*}
	and thus 
	\begin{equation*}
		\begin{array}{c}
			\hat{\beta}=\beta + \frac{1}{\sqrt{T}}\Psi_{\beta}+ o_p\left(\frac{1}{\sqrt{
					T }} \right). \label{eq:a:43}%
		\end{array}%
	\end{equation*}
	If $\beta=0$, then $\left( \hat{\beta}\hat{\beta}^{\prime }\right)^{-1}\hat{
		\beta}= \left( \frac{\Psi_\beta}{\sqrt{T}}\frac{\Psi_\beta^{\prime }}{\sqrt{
			T}}\right)^{-1} \frac{\Psi_\beta}{\sqrt{T}}=\sqrt{T}\left(
	\Psi_\beta\Psi_\beta^{\prime }\right)^{-1}\Psi_\beta.$ We can look at the
	asymptotic properties of $\hat{\Lambda}$ by looking at $\mathcal{T}_{i,j}$
	's. For convenience in the following sections otherwise well mentioned, we
	denote $\mathcal{T}_{i,j}$ as the previous $\mathcal{T}_{i,j}$ divided by $%
	\sqrt{T}$ such that: $\mathcal{T}^{{new}}_{i,j}=\frac{\mathcal{T}^{{old}
			_{i,j}}}{\sqrt{ T}}$. 
	\begin{equation}
		\begin{array}{l}
			\text{vec}\left(\mathcal{T}_{1,1} \right)=\Lambda, \\ 
			\text{vec}\left(\mathcal{T}^{\mu=0}_{1,2} \right)=\frac{\left(
				\gamma_{xx}^{-1}\otimes I\right)\text{vec}\left[0~~VX_-^{\prime }/\sqrt{T} %
				\right] }{\sqrt{T}} +o_p\left(\frac{1}{\sqrt{T}}\right) = o_p(1), \\ 
			\text{vec}\left(\mathcal{T}^{\mu}_{1,2} \right)= \frac{\left(
				\gamma_{zz}^{-1}\otimes I \right)\text{vec}\left(VZ_-^{\prime }/\sqrt{T}
				\right)}{\sqrt{T}} +o_p\left(\frac{1}{\sqrt{T}}\right) = o_p(1), \\ 
			\text{vec}\left(\mathcal{T}_{1,3} \right)=\left(I\otimes({\hat{\beta}}\hat{%
				\beta}^{\prime})^{-1}\hat{\beta} \right)\text{vec}\left(\sqrt{T}\left( \hat{f%
			} -f\right) \right)+o_p\left(\frac{1}{\sqrt{T}}\right)\nonumber \\ 
			=\left(I\otimes\left( \Psi_\beta\Psi_\beta^{\prime }\right)^{-1}\Psi_\beta
			\right) \text{vec}\left(\sqrt{T}\left( \hat{f}-f\right) \right) \nonumber \\ 
			= \left(I\otimes\left( \Psi_\beta\Psi_\beta^{\prime }\right)^{-1}\Psi_\beta
			\right)\psi_f, \\ 
			\text{vec}\left(\mathcal{T}_{1,4} \right)=-\left( \Lambda^{\prime }\otimes ({%
				\hat{\beta}}\hat{\beta}^{\prime})^{-1}\hat{\beta} \right) \text{vec}
			\left(\left( \hat{\beta}^{\prime }-\beta^{\prime }\right) \right)+o_p\left( 
			\frac{1}{ \sqrt{T}}\right)\nonumber \\ 
			= -\left( \Lambda^{\prime }\otimes\left( \Psi_\beta\Psi_\beta^{\prime
			}\right)^{-1}\Psi_\beta \right)\kappa_{K_C,N}\psi_\beta, \\ 
			\mathcal{T}_{1,5}=({\hat{\beta}}\hat{\beta}^{\prime})^{-1}\hat{\beta}\left( 
			\hat{\beta}^{\prime }-\beta^{\prime }\right) M_{\hat{V}^{\prime }}
			Z_-^{\prime }\left( Z_-M_{\hat{V}^{\prime }}Z_-^{\prime }\right)^{-1} o_p(1) %
			\nonumber \\ 
			=\left(\Psi_\beta\Psi_\beta^{\prime
			}\right)^{-1}\Psi_\beta\Psi_\beta^{\prime }M_{\hat{V}^{\prime }} Z_-^{\prime
			}\left( Z_-M_{\hat{V}^{\prime }}Z_-^{\prime }\right)^{-1} o_p(1) \nonumber
			\\ 
			=\frac{M_{\hat{V}^{\prime }} Z_-^{\prime }}{T}\left( \frac{Z_-Z_-^{\prime } 
			}{T} \right)^{-1} o_p(1)=o_p(1), \\ 
			\text{vec}\left(\mathcal{T}_{2,1} \right)=\left(\varrho^{\prime }\otimes
			\left(({\hat{\beta}}\hat{\beta}^{\prime})^{-1}\hat{\beta} A_\beta^{\prime
			}(I\otimes \Sigma_v)\right)\right) \text{vec}\left( \left(\hat{\beta}- {\
				\beta }\right)\right) +o_p(\frac{1}{\sqrt{T}})\nonumber \\ 
			= \left(\varrho^{\prime }\otimes \left(\left(\Psi_\beta\Psi_\beta^{\prime
			}\right)^{-1}\Psi_\beta A_\beta^{\prime }(I\otimes
			\Sigma_v)\right)\right)\Psi_\beta, \\ 
			\text{vec}\left(\mathcal{T}_{2,2} \right)=\left( \varrho^{\prime }\otimes 
			\frac{1}{2}({\hat{\beta}}\hat{\beta}^{\prime})^{-1}\hat{\beta} \left({B}^*
			\right)\right)\text{vec}\left(\left(\text{vec}\left( \frac{VV^{\prime }}{T}
			\right)-\text{vec}({\Sigma_v})\right)\right)+o_p(\frac{1}{\sqrt{T}}), \\ 
			\text{vec}\left(\mathcal{T}_{3} \right)=\left(\varrho^{\prime }\otimes \frac{
				1}{2}({\hat{\beta}}\hat{\beta}^{\prime})^{-1}\hat{\beta} \right)
			\iota_N\left( \frac{\text{tr}\left( EE^{\prime }\right)}{NT}-{\sigma}^2
			\right)+o_p(\frac{1 }{\sqrt{T}}).%
		\end{array}%
	\end{equation}
	~~~\newline
	Based on the above derivations, we can see that $\mathcal{T}_{1,2}, \mathcal{%
		\ \ T}_{1,5}$ converge to zero in probability but not the rest terms, and
	thus given $\beta=0$, which implies that the estimated parameter $\hat{%
		\Lambda}$ does not converge to $\Lambda$ in probability but converges to $%
	\Lambda + \epsilon$ with $\epsilon$ having a non-standard distribution.~%
	\newline
	~~~\newline
	If $\beta=\frac{B}{\sqrt{T}}$, where $B$ is full rank, then 
	\begin{equation*}
		( \hat{\beta} \hat{\beta}^{\prime } )^{-1}\hat{\beta}= \sqrt{T}\left(
		\left(B+\Psi_\beta\right)\left(B+\Psi_\beta\right)^{\prime
		}\right)^{-1}\left(B+\Psi_\beta\right).
	\end{equation*}
	Again, we consider those decomposed terms (we only show those terms that do
	not converge to zero in probability): 
	\begin{equation}
		\begin{array}{l}
			\text{vec}\left(\mathcal{T}_{1,3} \right)=\left(I\otimes (\hat{\beta}\hat{
				\beta}^{\prime} )^{-1}\hat{\beta} \right)\text{vec}\left(\left( \hat{f}
			-f\right) \right)+o_p\left(\frac{1}{\sqrt{T}}\right)\nonumber \\ 
			=\left(I\otimes\left(
			\left(B+\Psi_\beta\right)\left(B+\Psi_\beta\right)^{\prime
			}\right)^{-1}\left(B+\Psi_\beta\right) \right) \text{vec}\left(\sqrt{T}
			\left( \hat{f}-f\right) \right) \nonumber \\ 
			= \left(I\otimes\left(
			\left(B+\Psi_\beta\right)\left(B+\Psi_\beta\right)^{\prime
			}\right)^{-1}\left(B+\Psi_\beta\right) \right)\psi_f, \\ 
			\text{vec}\left(\mathcal{T}_{1,4} \right)=-\left( \Lambda^{\prime }\otimes (%
			\hat{\beta}\hat{\beta}^{\prime} )^{-1}\hat{\beta} \right) \text{vec}
			\left(\left( \hat{\beta}^{\prime }-\beta^{\prime }\right) \right)+o_p\left( 
			\frac{1}{ \sqrt{T}}\right)\nonumber \\ 
			= -\left( \Lambda^{\prime }\otimes\left(
			\left(B+\Psi_\beta\right)\left(B+\Psi_\beta\right)^{\prime
			}\right)^{-1}\left(B+\Psi_\beta\right) \right)\kappa_{K_C,N}\psi_\beta, \\ 
			\text{vec}\left(\mathcal{T}_{2,1} \right)=\left(\varrho^{\prime }\otimes
			\left( (\hat{\beta}\hat{\beta}^{\prime} )^{-1}\hat{\beta} A_\beta^{\prime
			}(I\otimes \Sigma_v)\right)\right) \text{vec}\left( \left(\hat{\beta}- {\
				\beta }\right)\right) +o_p(\frac{1}{\sqrt{T}})\nonumber \\ 
			= \left(\varrho^{\prime }\otimes
			\left(\left(\left(B+\Psi_\beta\right)\left(B+\Psi_\beta\right)^{\prime
			}\right)^{-1}\left(B+\Psi_\beta\right) A_\beta^{\prime }(I\otimes
			\Sigma_v)\right)\right)\left(B+\Psi_\beta\right), \\ 
			\text{vec}\left(\mathcal{T}_{2,2} \right)=\left( \varrho^{\prime }\otimes 
			\frac{1}{2} (\hat{\beta}\hat{\beta}^{\prime} )^{-1}\hat{\beta} \left({B}^*
			\right)\right)\text{vec}\left(\left(\text{vec}\left( \frac{VV^{\prime }}{T}
			\right)-\text{vec}({\Sigma_v})\right)\right)+o_p(\frac{1}{\sqrt{T}}), \\ 
			\text{vec}\left(\mathcal{T}_{3} \right)=\left(\varrho^{\prime }\otimes \frac{
				1}{2} (\hat{\beta}\hat{\beta}^{\prime} )^{-1}\hat{\beta} \right)
			\iota_N\left( \frac{\text{tr}\left( EE^{\prime }\right)}{NT}-{\sigma}^2
			\right)+o_p(\frac{1}{\sqrt{T}}),%
		\end{array}%
	\end{equation}
	which imply that $\hat{\Lambda}$ does not converge to $\Lambda$ in
	probability but again converges to $\Lambda + \epsilon$ with $\epsilon$
	having a non-standard distribution. ~\newline
	$\square$ ~~~\newline
	~~~\newline
	~~~\newline
	We only analyze one out of two approaches (without knowledge of unspanned
	factors) proposed by Adrian et al. (2013) since these two appraoches are
	equivalent, as suggested by the following proposition.
	
	\begin{proposition}
		\label{prop:two are the same} Under Assumptions 2.1, estimation results via
		the following two three-stage procedures proposed in Adrian et al. (2013), I
		and II, are numerically identical:
		
		\begin{itemize}
			\item[I.] (1) the first step is to obtain estimates of $\mu, \Phi $ via
			linear regression using the first equation in Assumption 2.1;~\newline
			(2) the second step is to obtain estimates of $a^{(n-1)}, d^{(n-1)},
			\beta_{I}^{(n-1)^{\prime }}, \Sigma_{e}$ by regressing excess returns on a
			constant, the lagged and the contemporaneous factors according to 
			\begin{equation*}
				\begin{array}{c}
					rx_{t+1,n-1}= a^{(n-1)} + d^{(n-1)} X_{t} +\beta_{I}^{(n-1)^{\prime }}
					X_{t+1} +e_{t+1,n-1},%
				\end{array}%
			\end{equation*}
			(3) the final step is to obtain the estimates of $\lambda_{0}, \lambda_{1}$: 
			\begin{equation*}
				\begin{array}{l}
					\widehat{\lambda}_{0,I} = \sum_n \left(\sum_m \widehat{\beta}_{I}^{(m-1)} 
					\widehat{\beta}_{I}^{(m-1)^{\prime }} \right)^{-1} \widehat{\beta}
					_{I}^{(n-1)} \\ 
					\times \left(\widehat{a}^{(n-1)}+ \widehat{g}_{I}^{(n-1)^{\prime }}\left( {\
						\beta}_{I}, {\Sigma}_v, {\Sigma}_e \right)+\widehat{\beta}
					_{I}^{(n-1)^{\prime }}\widehat{\mu}\right), \\ 
					\widehat{\lambda}_{1,I}= \sum_n \left(\sum_m \widehat{\beta}_{I}^{(m-1)} 
					\widehat{\beta}_{I}^{(m-1)^{\prime }} \right)^{-1} \widehat{\beta}
					_{I}^{(n-1)} \\ 
					\times \left(\widehat{d}^{(n-1)}+ \widehat{\beta}^{(n-1)^{\prime }}_{I} 
					\widehat{\Phi} \right).%
				\end{array}%
			\end{equation*}
			
			\item[II.] (1) the first step is to obtain estimates of innovations $%
			\widehat{v}_t$ via linear regression using the first equation in Assumption
			2.1; ~\newline
			(2) the second step is to obtain estimates of $b^{(n-1)}, c^{(n-1)}, \beta_{ 
				\text{II}}^{(n-1)^{\prime }}, \Sigma_{e}$ by regressing excess returns on a
			constant, the lagged and the contemporaneous innovation factors ($v_{t+1}$
			is replaced by estimates $\widehat{v}_{t+1}$ in practice) according to 
			\begin{equation*}
				\begin{array}{c}
					rx_{t+1,n-1}= b^{(n-1)} + c^{(n-1)} X_{t} +\beta_{\text{II}}^{(n-1)^{\prime
					}}{v}_{t+1} +e_{t+1,n-1},%
				\end{array}%
			\end{equation*}
			which is derived by plugging $v_{t+1}=X_{t+1}-\mu-\Phi X_t$ into the second
			equation in Assumption 2.1. ~\newline
			(3) the final step is to obtain the estimates of $\lambda_{0}, \lambda_{1}$: 
			\begin{equation*}
				\begin{array}{l}
					\widehat{\lambda}_{0,II} = \sum_n \left(\sum_m \widehat{\beta}_{\text{II}
					}^{(m-1)}\widehat{\beta}_{\text{II}}^{(m-1)^{\prime }} \right)^{-1} \widehat{
						\beta}_{\text{II}}^{(n-1)} \\ 
					\times \left(\widehat{b}^{(n-1)}+ \widehat{g}_{\text{II} }^{(n-1)^{\prime
					}}\left({\beta}, {\Sigma}_v, {\Sigma}_e \right)\right), \\ 
					\widehat{\lambda}_{1,II}= \sum_n \left(\sum_m \widehat{\beta}_{\text{II}
					}^{(m-1)}\widehat{\beta}_{\text{II}}^{(m-1)^{\prime }} \right)^{-1} \widehat{
						\beta}_{\text{II}}^{(n-1)} \widehat{c}^{(n-1)}.%
				\end{array}%
			\end{equation*}
		\end{itemize}
	\end{proposition}
	
	\noindent \textit{Proof of Proposition \ref{prop:two are the same}}. 
Here we only show that $\widehat{\lambda}_{1,I}=\widehat{\lambda}_{1,II}$,
and the equality $\widehat{\lambda}_{0,I}=\widehat{\lambda}_{0,II}$ follows
the same argument. By the Frisch-Waugh-Lovell Theorem, $\widehat{\beta}_{I}
= \widehat{\beta}_{\text{II}}$. Denote $\xi_c=X_-M_{(\iota_T,~
	\xi_\beta^{\prime })}, ~ \xi_d=X_- M_{(\iota_T,~X^{\prime })}, ~\xi_\beta =
XM_{Z_-}$, and let 
\begin{equation*}
	\begin{array}{l}
		\widehat{c} = (\xi_c\xi_c^{\prime})^{-1}\xi_c R^{\prime }, \\ 
		\widehat{d} = (\xi_d\xi_d^{\prime})^{-1}\xi_d R^{\prime }, \\ 
		\widehat{\beta} = (\xi_\beta\xi_\beta^{\prime})^{-1}\xi_\beta R^{\prime },
		\\ 
		\widehat{\Phi} = (X_-M_{\iota_T}X_-^{\prime})^{-1}X_-M_{\iota_T}X^{\prime }.%
	\end{array}%
\end{equation*}
Notice the following two equations hold 
\begin{equation}
	\begin{array}{l}
		M_{(\iota_T, \widehat{V}^{\prime })} X_-^{\prime }(X_-^{\prime }M_{(\iota_T, 
			\widehat{V}^{\prime })} X_-^{\prime})^{-1} =M_{\iota_T} X_-^{\prime
		}(X_-^{\prime }M_{\iota_T} X_-^{\prime})^{-1}, \\ 
		\xi_c^{\prime }(\xi_c\xi_c^{\prime -1} =M_{\iota_T} X_-^{\prime
		}(X_-^{\prime }M_{\iota_T} X_-^{\prime})^{-1} \label{eq:proof two equality 1}%
		.%
	\end{array}%
\end{equation}
\begin{equation}
	\begin{array}{c}
		M_{\iota_T} X_-^{\prime }(X_-^{\prime }M_{\iota_T} X_-^{\prime})^{-1} =
		\xi_\beta^{\prime }(\xi_\beta\xi_\beta^{\prime -1}\widehat{\Phi}^{\prime }+
		\xi_d^{\prime }(\xi_d\xi_d^{\prime -1},%
	\end{array}
	\label{eq:proof two equality 2}
\end{equation}
which directly leads to the equality $\widehat{c}=\widehat{d}+ \widehat{\Phi}
\widehat{\beta}$. The equality (\ref{eq:proof two equality 1}) is obvious
due to the orthogonality $\xi_\beta X_-^{\prime }=0$. Therefore, we only
need to show the equation (\ref{eq:proof two equality 2}). 
\begin{equation*}
	\begin{array}{l}
		\xi_\beta^{\prime }(\xi_\beta\xi_\beta^{\prime })^{-1}\widehat{\Phi}^{\prime
		}+ \xi_d^{\prime }(\xi_d\xi_d^{\prime})^{-1} \\ 
		=\xi_\beta^{\prime }(\xi_\beta\xi_\beta^{\prime})^{-1}\widehat{\Phi}^{\prime
		} \\ 
		+ (X_-(M_{\iota_T} -P_{M_{\iota_T}X^{\prime }}))^{\prime }\left(
		(X_-(M_{\iota_T} -P_{M_{\iota_T}X^{\prime }} ) )(X_-(M_{\iota_T}
		-P_{M_{\iota_T}X^{\prime }} ) )^{\prime }\right)^{-1} \\ 
		= \xi_\beta^{\prime }(\xi_\beta\xi_\beta^{\prime})^{-1}\widehat{\Phi}%
		^{\prime } \\ 
		+ (X_-(M_{\iota_T} -P_{M_{\iota_T}X^{\prime }} ) )^{\prime }\left(
		X_-M_{\iota_T}X_-^{\prime }-X_-P_{M_{\iota_T}X^{\prime }} X_-^{\prime
		}\right)^{-1} \\ 
		= \xi_\beta^{\prime }(\xi_\beta\xi_\beta^{\prime})^{-1}\widehat{\Phi}%
		^{\prime } \\ 
		+ (X_-(M_{\iota_T} -P_{M_{\iota_T}X^{\prime }} ) )^{\prime }\left[\left(
		X_-M_{\iota_T}X_-^{\prime }\right)^{-1}+ \left( X_-M_{\iota_T}X_-^{\prime
		}\right)^{-1} X_-P_{M_{\iota_T}X^{\prime }} X_-^{\prime
		}(\xi_d\xi_d^{\prime})^{-1} \right] \\ 
		= \xi_\beta^{\prime }(\xi_\beta\xi_\beta^{\prime})^{-1}\widehat{\Phi}%
		^{\prime }+ \xi_c^{\prime }(\xi_c\xi_c^{\prime})^{-1}+ \\ 
		\left[-P_{M_{\iota_T}X^{\prime }}X_-^{\prime }\left(
		X_-M_{\iota_T}X_-^{\prime }\right)^{-1}+ (X_-(M_{\iota_T}
		-P_{M_{\iota_T}X^{\prime }} ) )^{\prime }\left( X_-M_{\iota_T}X_-^{\prime
		}\right)^{-1} X_-P_{M_{\iota_T}X^{\prime }} X_-^{\prime
		}(\xi_d\xi_d^{\prime})^{-1} \right] \\ 
		= \xi_\beta^{\prime }(\xi_\beta\xi_\beta^{\prime})^{-1}\widehat{\Phi}%
		^{\prime }+ \xi_c^{\prime }(\xi_c\xi_c^{\prime})^{-1}+ \\ 
		\left[-P_{M_{\iota_T}X^{\prime }}X_-^{\prime }\left(
		X_-M_{\iota_T}X_-^{\prime }\right)^{-1}+ M_{\iota_T}X_-^{\prime }\left(
		X_-M_{\iota_T}X_-^{\prime }\right)^{-1} X_-P_{M_{\iota_T}X^{\prime }}
		X_-^{\prime }(\xi_d\xi_d^{\prime})^{-1} \right] \\ 
		\left[ -P_{M_{\iota_T}X^{\prime }} X_-^{\prime }\left(
		X_-M_{\iota_T}X_-^{\prime }\right)^{-1} X_-P_{M_{\iota_T}X^{\prime }}
		X_-^{\prime }(\xi_d\xi_d^{\prime})^{-1} \right] \\ 
		= M_{Z_-}X^{\prime }\left(XM_{Z_-}X^{\prime
		}\right)^{-1}XM_{\iota_T}X_-^{\prime }\left( X_-M_{\iota_T}X_-^{\prime
		}\right)^{-1} + \xi_c^{\prime }(\xi_c\xi_c^{\prime})^{-1} \\ 
		+\left[-P_{M_{\iota_T}X^{\prime }}X_-^{\prime }\left(
		X_-M_{\iota_T}X_-^{\prime }\right)^{-1}+ P_{M_{\iota_T}X_-^{\prime
		}}P_{M_{\iota_T}X^{\prime }} X_-^{\prime }(\xi_d\xi_d^{\prime})^{-1} \right.
		\\ 
		\left. - P_{M_{\iota_T}X^{\prime }} P_{M_{\iota_T}X_-^{\prime
		}}P_{M_{\iota_T}X^{\prime }} X_-^{\prime }(\xi_d\xi_d^{\prime})^{-1}\right]
		\\ 
		= \xi_c^{\prime }(\xi_c\xi_c^{\prime })^{-1},%
	\end{array}%
\end{equation*}
where the last equality is due to the facts such that $(I-P_{M_{\iota_T}X^{
		\prime }}P_{M_{\iota_T}X_-^{\prime }})\Delta=0$, $P_{M_{\iota_T}X^{\prime
}}(I-P_{M_{\iota_T}X_-^{\prime }})\Delta=0$ with $\Delta=
M_{\iota_T}X^{\prime }\left(XM_{Z_-}X^{\prime }\right)^{-1}
XM_{\iota_T}X_-^{\prime }\left( X_-M_{\iota_T}X_-^{\prime }\right)^{-1} -
P_{M_{\iota_T}X^{\prime }} X_-^{\prime }\left(\xi_d\xi_d^{\prime
}\right)^{-1}$, and hence 
\begin{equation*}
	\begin{array}{l}
		M_{Z_-}X^{\prime }\left(XM_{Z_-}X^{\prime
		}\right)^{-1}XM_{\iota_T}X_-^{\prime }\left( X_-M_{\iota_T}X_-^{\prime
		}\right)^{-1} \\ 
		= \left( M_{\iota_T} -P_{M_{\iota_T}X_-^{\prime }} \right)X^{\prime } \\ 
		\times \left(\left(XM_{\iota_T}X^{\prime }\right)^{-1}
		+\left(XM_{\iota_T}X^{\prime }\right)^{-1} \left( XP_{M_{\iota_T}X_-^{\prime
		}}X^{\prime }\right) \left(XM_{Z_-}X^{\prime }\right)^{-1}
		\right)XM_{\iota_T}X_-^{\prime }\left( X_-M_{\iota_T}X_-^{\prime
		}\right)^{-1} \\ 
		=P_{M_{\iota_T}X^{\prime }}X_-^{\prime }\left( X_-M_{\iota_T}X_-^{\prime
		}\right)^{-1}- P_{M_{\iota_T}X_-^{\prime }}P_{M_{\iota_T}X^{\prime }}
		X_-^{\prime }\left( X_-M_{\iota_T}X_-^{\prime }\right)^{-1} \\ 
		+ \left( M_{\iota_T} -P_{M_{\iota_T}X_-^{\prime }} \right)X^{\prime
		}\left(XM_{\iota_T}X^{\prime }\right)^{-1} \left( XP_{M_{\iota_T}X_-^{\prime
		}}X^{\prime }\right) \left(XM_{Z_-}X^{\prime }\right)^{-1}
		XM_{\iota_T}X_-^{\prime }\left( X_-M_{\iota_T}X_-^{\prime }\right)^{-1} \\ 
		=P_{M_{\iota_T}X^{\prime }}X_-^{\prime }\left( X_-M_{\iota_T}X_-^{\prime
		}\right)^{-1}- P_{M_{\iota_T}X_-^{\prime }}P_{M_{\iota_T}X^{\prime }}
		X_-^{\prime }\left(\xi_d\xi_d^{\prime }\right)^{-1} \\ 
		+ P_{M_{\iota_T}X_-^{\prime }}P_{M_{\iota_T}X^{\prime }} X_-^{\prime
		}\left(X_-M_{\iota_T}X_-^{\prime }\right)^{-1} \left(
		X_-P_{M_{\iota_T}X^{\prime }}X_-^{\prime }\right) \left(\xi_d\xi_d^{\prime
		}\right)^{-1} \\ 
		+ \left( M_{\iota_T} -P_{M_{\iota_T}X_-^{\prime }} \right)X^{\prime
		}\left(XM_{\iota_T}X^{\prime }\right)^{-1} \left( XP_{M_{\iota_T}X_-^{\prime
		}}X^{\prime }\right) \left(XM_{Z_-}X^{\prime }\right)^{-1}
		XM_{\iota_T}X_-^{\prime }\left( X_-M_{\iota_T}X_-^{\prime }\right)^{-1} \\ 
		=P_{M_{\iota_T}X^{\prime }}X_-^{\prime }\left( X_-M_{\iota_T}X_-^{\prime
		}\right)^{-1}- P_{M_{\iota_T}X_-^{\prime }}P_{M_{\iota_T}X^{\prime }}
		X_-^{\prime }\left(\xi_d\xi_d^{\prime }\right)^{-1} +
		P_{M_{\iota_T}X^{\prime }}P_{M_{\iota_T}X_-^{\prime
		}}P_{M_{\iota_T}X^{\prime }} X_-^{\prime }\left(\xi_d\xi_d^{\prime
		}\right)^{-1}.%
	\end{array}%
\end{equation*}
~~\newline
$\square$ ~~~\newline

%

\end{document}